\documentclass[11pt]{article}
\usepackage{fullpage}
\usepackage{url}
\usepackage[colorlinks=true,linkcolor=Emerald,citecolor=RoyalBlue]{hyperref}
\usepackage{amsmath,nccmath, amsfonts,amsthm,amssymb,multirow}

\usepackage{times}
\usepackage{paralist}
\usepackage{graphicx}
\usepackage{floatpag}
\usepackage{algorithm}
\usepackage[noend]{algpseudocode}
\usepackage{bbold}
\usepackage{enumitem}
\usepackage{subcaption} 
\usepackage{soul}
\usepackage{comment}
\usepackage{calc}
\usepackage{upgreek}
\usepackage{mathtools}
\usepackage{color}
\usepackage[dvipsnames]{xcolor}
\usepackage{xspace}
\allowdisplaybreaks
\newenvironment{reminder}[1]{\bigskip
	\noindent {\bf Restatement of #1  }\em}{\smallskip}

\DeclareMathOperator*{\E}{\mathbb{E}}
\DeclareMathOperator*{\pE}{\widetilde{\mathbb{E}}}
\DeclareMathOperator*{\pPr}{\widetilde{\Pr}}
\DeclareMathOperator*{\Ind}{\mathbb{I}}

\DeclarePairedDelimiter\ip{\langle}{\rangle}

\def \F {{\mathbb F}}
\def \Q {{\mathbb Q}}
\def \R {\mathbb{R}}
\def \Z {\mathbb{Z}}
\def \N {\mathbb{N}}

\def \T {\widetilde{\E}}

\def \A {\mathcal{A}}
\def \cA {\mathcal{A}}
\def \cE {\mathcal{E}}
\def \cF {\mathcal{F}}

\def \Id {\mathrm{Id}}

\def \sos {SoS\xspace}

\def \SSE {\mathrm{SSE}}
\def \poly {\mathrm{poly}}

\def \eps {{\varepsilon}}
\def \one {\mathbb{1}}
\def \val {\mathrm{val}}
\def \viol {\mathrm{viol}}

\def \Chi {\raisebox{2pt}{$\chi$}}

\newcommand{\simp}[1]{\,\, \vdash_{#1} \,\,}
\renewcommand{\hat}[1]{\widehat{#1}}

\newcommand{\Tnote}[1]{}
\newcommand{\Bnote}[1]{}
\newcommand{\Pnote}[1]{}
\newcommand{\Dnote}[1]{}
\newcommand{\Mnote}[1]{}

\newtheorem{theorem}{Theorem}[section]
\newtheorem{observation}{Observation}[section]
\newtheorem{fact}[theorem]{Fact}

\newtheorem{corollary}[theorem]{Corollary}

\newtheorem{lemma}[theorem]{Lemma}
\newtheorem*{lemma*}{Lemma}
\newtheorem*{theorem*}{Theorem}
\newtheorem{claim}[theorem]{Claim}
\newtheorem*{claim*}{Claim}

\theoremstyle{definition}
\newtheorem{remark}[theorem]{Remark}
\newtheorem{definition}[theorem]{Definition}
\newtheorem{algorithm-thm}[theorem]{Algorithm}

\newcommand\MYcurrentlabel{xxx}

\newcommand{\MYstore}[2]{%
  \global\expandafter \def \csname MYMEMORY #1 \endcsname{#2}%
}

\newcommand{\MYload}[1]{%
  \csname MYMEMORY #1 \endcsname%
}

\newcommand{\MYnewlabel}[1]{%
  \renewcommand\MYcurrentlabel{#1}%
  \MYoldlabel{#1}%
}

\newcommand{\MYdummylabel}[1]{}

\newcommand{\torestate}[1]{%
  \let\MYoldlabel\label%
  \let\label\MYnewlabel%
  #1%
  \MYstore{\MYcurrentlabel}{#1}%
  \let\label\MYoldlabel%
}

\newcommand{\restatetheorem}[1]{%
  \let\MYoldlabel\label
  \let\label\MYdummylabel
  \begin{theorem*}[Restatement of Theorem~\ref{#1}]
    \MYload{#1}
  \end{theorem*}
  \let\label\MYoldlabel
}

\newcommand{\restatelemma}[1]{%
  \let\MYoldlabel\label
  \let\label\MYdummylabel
  \begin{lemma*}[Restatement of Lemma~\ref{#1}]
    \MYload{#1}
  \end{lemma*}
  \let\label\MYoldlabel
}

\newcommand{\restateprop}[1]{%
  \let\MYoldlabel\label
  \let\label\MYdummylabel
  \begin{proposition*}[Restatement of Proposition~\ref{#1}]
    \MYload{#1}
  \end{proposition*}
  \let\label\MYoldlabel
}

\newcommand{\restateclaim}[1]{%
  \let\MYoldlabel\label
  \let\label\MYdummylabel
  \begin{claim*}[Restatement of Claim~\ref{#1}]
    \MYload{#1}
  \end{claim*}
  \let\label\MYoldlabel
}

\newcommand{\restatefact}[1]{%
  \let\MYoldlabel\label
  \let\label\MYdummylabel
  \begin{fact*}[Restatement of Fact~\ref{#1}]
    \MYload{#1}
  \end{fact*}
  \let\label\MYoldlabel
}

\newcommand{\restate}[1]{%
  \let\MYoldlabel\label
  \let\label\MYdummylabel
  \MYload{#1}
  \let\label\MYoldlabel
}

\makeatletter
\let\c@fconjecture\c@conjecture
\makeatother

\makeatletter
\let\c@fconj\c@conj
\makeatother

\title{Playing Unique Games on Certified Small-Set Expanders}

\author{Mitali Bafna\thanks{Harvard University, \texttt{mitalibafna@g.harvard.edu}. Supported in part by a Simons Investigator Award and NSF Award CCF 1715187.}\and 
Boaz Barak\thanks{Harvard University, \texttt{b@boazbarak.org}. Supported by NSF awards CCF 1565264 and CNS 1618026 and a Simons Investigator Fellowship.}\and 
Pravesh K. Kothari\thanks{Carnegie Mellon University, \texttt{praveshk@cs.cmu.edu}. Supported by NSF CAREER Award 2047933}\and
Tselil Schramm \thanks{Stanford University, \texttt{tselil@stanford.edu}.}\and 
David Steurer \thanks{ETH Zurich, \texttt{dsteurer@inf.ethz.ch}.} \and}
\date{\today}

\begin{document}
\maketitle
\thispagestyle{empty}

\begin{abstract}
We give an algorithm for solving unique games (UG) instances whenever low-degree sum-of-squares proofs certify good bounds on the small-set-expansion of the underlying constraint graph via a hypercontractive inequality.
Our algorithm is in fact more versatile, and succeeds even when the constraint graph is not a small-set expander as long as the structure of non-expanding small sets is (informally speaking) ``characterized'' by a low-degree sum-of-squares proof.
Our results are obtained by rounding \emph{low-entropy} solutions --- measured via a new global potential function --- to sum-of-squares (SoS) semidefinite programs. 
This technique adds to the (currently short) list of general tools for analyzing SoS relaxations for \emph{worst-case} optimization problems. 

As corollaries, we obtain the first polynomial-time algorithms for solving any UG instance where the constraint graph is either the \emph{noisy hypercube}, the \emph{short code} or the \emph{Johnson} graph.
The prior best algorithm for such instances was the eigenvalue enumeration algorithm of Arora, Barak, and Steurer (2010) which requires quasi-polynomial time for the noisy hypercube and nearly-exponential time for the short code and Johnson graphs.
All of our results achieve an approximation of $1-\epsilon$ vs $\delta$ for UG instances, where $\epsilon>0$ and $\delta > 0$ depend on the expansion parameters of the graph but are independent of the alphabet size. 

\end{abstract}

\clearpage

\tableofcontents

\thispagestyle{empty}
\setcounter{page}{0}

\clearpage

\section{Introduction}

The \emph{Unique Games Conjecture (UGC)}~\cite{Khot02} is a central open question in computational complexity and algorithms.
In short, the UGC stipulates that distinguishing between almost satisfiable (value $\geq 1-\epsilon$) and highly unsatisfiable (value $\leq \epsilon$) instances of a certain 2-variable constraint satisfaction problem called \emph{Unique Games} is NP-hard.
The UGC is known to imply a vast number of hardness-of-approximation results in combinatorial optimization (e.g. Vertex Cover \cite{KhotR08}, Max Cut \cite{KhotKMO07}, constraint satisfaction problems \cite{Raghavendra08}, and Sparsest Cut \cite{ChawlaKKRS06}), but it is still not known whether the conjecture is true or false.
In a significant recent breakthrough, Khot, Minzer, and Safra \cite{KhotMS18} (building on \cite{KhotMS17,DinurKKMS18,BarakKS19}) showed that it is NP-hard to distinguish $\frac{1}{2}$-satisfiable instances from $\eps$-satisfiable instances. 
While  \cite{KhotMS18}'s result leads to some hardness-of-approximation results \cite{BhangaleK19}, it is not sufficient to recover most of the striking consequences of the UGC.
Moreover, regardless of the UGC's truth, there may be mild, natural conditions on instances that allow for polynomial-time algorithms for both the UG problem itself as well as ``downstream'' problems such as Max Cut. 

The community-wide quest to potentially refute the UGC, as well as to understand conditions under which UG instances are easy, has produced numerous advances in the broader theory of algorithms over the past two decades.
Notable examples include sophisticated graph partitioning tools \cite{ABS,lee2014multiway,kwok2013improved,louis2012many}, local random walks and similar stochastic processes \cite{gharan2012approximating,andersen2016almost}, and perhaps most of all, new tools for analyzing and rounding semidefinite programs (SDPs).
The groundbreaking result of Raghavendra \cite{Raghavendra08} exposed a deep connection between the UGC and the performance of a semidefinite programming relaxation called the {\em basic SDP}, showing that the UGC implies that the basic SDP is the optimal polynomial-time algorithm for any constraint satisfaction problem. 
Efforts to refute the UGC thus naturally led to the study of more powerful SDPs such as the sum-of-squares (SoS) hierarchy \cite{Lasserre00,Parrilo00}. The study of SoS (and specifically SoS algorithms for unique games) has since blossomed in the algorithms community, leading to many algorithmic advances.
These include the development of general techniques for analyzing and rounding SoS SDPs, such as global correlation rounding \cite{BarakRS11}, and the proofs-to-algorithms perspective \cite{BarakBHKSZ12}.
These techniques have in turn led to numerous algorithmic breakthroughs, for problems originating in high dimensional statistics (e.g. \cite{BarakKS14,ma2016polynomial,hopkins2018mixture,DBLP:journals/corr/abs-1711-07465,DBLP:journals/corr/abs-1711-11581,DBLP:journals/corr/abs-2005-02970,DBLP:journals/corr/abs-2005-06417,DBLP:conf/stoc/CherapanamjeriH20}), quantum computation \cite{BKS17}, statistical physics \cite{jain2019mean}, and more \cite{raghavendra2012approximating}.

\medskip

In this work, we give new algorithms for a large family of structured instances of the Unique Games problem.
Our algorithms are obtained via a novel analysis of Sum-of-Squares SDPs.
Specifically, we define a new {\em global} potential function which is a proxy for the entropy of the distribution over non-integral SDP solutions to the UG instance, and we show that when the entropy of the SDP solution is low, it is easy to round. 
We are then able to control our potential function if the UG constraint graph satisfies certain properties.
Our potential function offers an alternative to the {\em global correlation} function introduced by \cite{BarakRS11}, which is one of very few known tools for analyzing SoS relaxations of worst-case problems.
Using our new techniques, we show that polynomial-sized SoS relaxations solve Unique Games on graphs which were out of reach of previous techniques, including Short Code graphs, the Noisy Hypercube, and the Johnson graph.

To control our potential function, we exploit and deepen the connection between Unique Games and the related Small-Set Expansion problem.
A graph is said to be a $(\delta,\eta)$-small-set expander if all sets of measure at most $\delta$ have edge-expansion at least $\eta$, and the Small Set Expansion Hypothesis (SSEH) states that for each $\eps > 0$ there exists a sufficiently small constant $\delta$ such that it is NP-hard to decide whether a given graph is a $(\delta,1-\eps)$-small set expander, or whether the graph contains a set of measure $\le \delta$ with expansion $< \eps$.
A sequence of works in this area uncovered a fundamental relationship between the two problems~\cite{RaghavendraS10,RaghavendraST12}.
Our current state of knowledge regarding the relationship between these problems can be roughly summarized as follows. Raghavendra and Steurer \cite{RaghavendraS10} gave a reduction from the Small Set Expansion problem to Unique Games. 
Raghavendra, Steurer, and Tulsiani \cite{RaghavendraST12} reduced Unique games to Small Set Expansion, under the additional assumption that the UG  constraint graph is itself a small-set expander. 
Though a reduction in the opposite direction (without this additional assumption) has been postulated, it is still not known whether UGC implies SSEH.

Both the reductions above were \emph{worst-case reductions}, showing that if one problem is easy on \emph{all} instances (or all instances of certain type, in the case of \cite{RaghavendraST12}'s work) then the other is also easy on all instances.
In this work we show a ``point-wise'' reduction from UG instances on small-set expanders to the small-set expansion problem {\em within the sum-of-squares framework}.
Specifically, we show that for a graph $G$ for which SOS can certify small-set expansion, SOS can also solve {\em any} UG instance on $G$.
More precisely, we show that our global SoS potential function always reflects the fact that the entropy is low in a small-set expander.

In addition to this pleasing qualitative statement, our result yields polynomial time algorithms for solving arbitrary Unique Games instances on algebraic families of constraint graphs such as the \emph{noisy hypercube}~\cite{KhotV15,BarakBHKSZ12} and \emph{short code} graphs~\cite{BGHMRS15} that have been extensively investigated in the context of constructing integrality gaps for UG and related problems. 
The quantitative guarantees of our rounding algorithm are substantially stronger than previously known and in particular we give the first polynomial-time algorithms for instances over these graphs in the UGC parameter regime of distinguishing between $1-\epsilon$ satisfiable and $\leq \delta$ satisfiable instances for small constant $\epsilon,\delta>0$.

Our rounding technique is, in fact, more versatile and succeeds even when the constraint graph admits non-expanding sets so long as the structure of non-expanding small sets is (informally speaking) ``understood'' by the low-degree sum-of-squares proof system. Somewhat curiously, the main technical innovation in the recent proof of NP-hardness of the $2$-to-$1$-Games problem due to Khot, Minzer and Safra~\cite{KhotMS18} (building on \cite{KhotMS17,DinurKKMS18,BarakKS19,KMMS}) involves an exhaustive characterization of the structure of small non-expanding sets in algebraic families such as the \emph{Johnson} and the \emph{Grassmann} Graphs that establish the truth of the $2$-to-$1$ conjecture. We show that their proof in fact yields a low-degree sum-of-squares certificate characterizing  the non-expanding sets in the Johnson graph. Building on this, we obtain a polynomial time algorithm for solving arbitrary Unique Games instances when the constraint graph is the Johnson graph\footnote{The second largest eigenvalue of the Grassman graph's random walk matrix is already $1/2$, and hence it is not an interesting constraint graph for the UGC regime of nearly satisfiable instances which we study in this work.}.

\medskip

\subsection{Our Results} \label{sec:results}

We now formally state our results.
Our first theorem shows that unique games is easy on graphs which are ``certifiable small-set expanders.''
In order to state our theorem we first need to define certifiable small-set expanders.
We use the well known relationship between \emph{hypercontractivity} and small set expansion (e.g., \cite{KKL88}).
This is a relation between a polynomial inequality derived from the graph and the combinatorial property that small sets have large expansion.

For a graph $G=(V,E)$ and $\lambda \geq 0$, we let $V_\lambda(G)$ denote the linear subspace of $\mathbb{R}^V$ that is spanned by the eigenvectors of $G$'s normalized adjacency matrix that correspond to eigenvalues of value at least $1-\lambda$.
We say that $G$ is \emph{$(\lambda,C)$ hypercontractive} if every $f \in V_\lambda(G)$ satisfies $\E_{v \sim V}[ f_v^4] \leq C \E_{v \sim V} [ f_v^2 ]^2$.\footnote{This is also called \emph{2 to 4 hypercontractivity}; we drop the ``2 to 4'' modifier as it is the only notion of hypercontractivity we use.}
It is known that if $G$ is hypercontractive then subsets of size $\poly(\lambda)/C$ have expansion at least $\Omega(\lambda)$ and a certain converse was given in \cite{BarakBHKSZ12}.

We say that $G$ is {\em $(\lambda,C,D)$-certifiably  hypercontractive} if $G$ is $(\lambda,C)$ hypercontractive and furthermore this fact is certifiable by a degree-$D$ \sos proof  (see Definition~\ref{def:hyp}).
Our main theorem shows that when a graph is certifiably  hypercontractive, it is also a tractable constraint graph for unique games instances.\footnote{To reduce clutter, we state many of our results with explicit numerical constants. We have made no attempt to optimize these.}

\begin{theorem}[Unique games on certifiable small-set expanders]\label{thm:intro-main}
For every $C>0$, $\lambda \in (0,1)$, $D \in \N$ there exists a polynomial-time algorithm $A$ such that if:
\begin{compactitem}
    \item $G$ is $(\lambda,C,D)$-certifiably  hypercontractive, \emph{and}  
    \item $I$ is an affine unique games instance with constraint graph $G$, and $\val(I) = 1 - \eps$, for $\eps \leq \lambda^2/100$.
\end{compactitem}
Then $A(I)$ outputs an assignment to $I$ with value at least  $\frac{\eps\lambda^4}{64 C}$.
\end{theorem}

The algorithm is obtained by rounding the standard degree-$D'$  \sos relaxation for unique games, where $D'$ is a constant depending on $C,\eps,\lambda,D$.
The degree-$D'$ \sos relaxation for a unique games over constraint graph $G = (V,E)$ and alphabet $\Sigma$ can be computed in $(|V|\cdot |\Sigma|)^{O(D')}$ time (see \cite{RaghavendraW17}).

\begin{remark}[Completeness Gap vs Set Size]
The completeness bound in Theorem~\ref{thm:intro-main} is \emph{independent} of the the parameter $C$ corresponding to the set size. This is important, since, just as it works for expander graphs, the basic SDP can solve unique games instances on small set expanders if the completeness parameter can depend on the size of the sets that expand (see \cite{DBLP:journals/eccc/AroraIMS10} and  Theorem 1.1 of \cite{RaghavendraS09}).
In contrast, obtaining a guarantee where the set-size $\delta$ is independent of the completeness such as the one in Theorem~\ref{thm:intro-main} inherently requires using higher levels of the SOS hierarchy, since there are known integrality gap instances for the basic SDP where the constraint graphs are certifiably hypercontractive (e.g., the short-code graph, see \cite[Cor.~7.2]{BGHMRS15} and \cite{BarakBHKSZ12}).
As we discuss below, our algorithm solves such instances in polynomial time. 
\end{remark}

We prove Theorem~\ref{thm:intro-main} and give more precise quantitative bounds in Section~\ref{sec:cert-sse}.
From this theorem, we are able to obtain corollaries for the Noisy Hypercube and the Noisy Short Code graphs, since the latter are known to have sum-of-squares certificates of small-set expansion via hypercontractivity \cite{BarakBHKSZ12}.

\begin{corollary}[Unique Games on the Noisy Hypercube]\label{cor:intro-hc}
For every $0.001 > \eps > 0$ and $\tfrac{1}{4} > \alpha >0$ there is a polynomial time algorithm $A$ and a constant $\tau = \tau(\alpha,\eps)>0$, such that if $I$ is an affine unique games instance over the $\alpha$-noisy hypercube with $\val(I) \geq 1-\eps$ then $A(I)$ outputs an assignment to $I$ with value at least $\tau$.
\end{corollary}

\begin{corollary}[Unique Games on the Noisy Short Code Graph]\label{cor:intro-sc}
There exists constant $\eps_0 > 0$ such that for every $\eps \in [0,\eps_0)$, $\alpha \in (0,1)$ there exists a polynomial-time algorithm $A$ and a constant $\tau = \tau(\alpha,\eps)>0$, such that if $I$ is an affine unique games instance over the $\alpha$-noisy shortcode graph with $\val(I) \geq 1-\eps$, then $A(I)$ outputs an assignment to $I$  with value at least $\tau$.
\end{corollary}
The value $\tau$ in both corollaries is of the form $\poly(\eps)\exp(-c\sqrt{\eps}/\alpha)$ for $c>0$ a fixed constant.
Crucially, $\tau$ is independent of the alphabet size of $I$.
Though there was prior work giving SOS certificates for \emph{specific} instances of this type (see Section~\ref{sec:related}), this is the first algorithm for \emph{all} affine unique games instances over these graphs in the UGC parameter regime.
We derive these corollaries (with more precise asymptotics) and give formal definitions of the relevant graphs in Section~\ref{sec:hyper}.

Finally, by extending our methods we are also able to obtain a result for the Johnson graph,\footnote{For $n,\ell,\alpha$, the $(n,\ell,\alpha)$ Johnson graph has the vertices $\binom{[n]}{\ell}$ with $S \sim S'$ if $|S\cap S'|\geq (1-\alpha)\ell$, see Definition~\ref{def:johnson}.} despite the fact that it is not a small-set expander.

\begin{theorem}[Unique Games on the Johnson Graph]\label{thm:intro-j}
For every $0.001 > \eps >0$,  $\tfrac{1}{2} > \alpha > 0$, and integer $\ell\in \N$ with $\ell\alpha \in \N$, there is a polynomial-time algorithm $A$ and a constant $\tau=\tau(\eps,\alpha,\ell)>0$ with the following guarantee: for $n \in \N$ sufficiently large, if $I$ is an affine unique games instance over the $(n,\ell,\alpha)$-Johnson graph with $\val(I) \ge 1 - \eps$, then $A(I)$ returns an assignment to $I$ of value at least $\tau$.
\end{theorem}

The parameter $\tau$ is of the form $\poly\left(\frac{\eps}{\binom{\ell}{r}\exp(c'r)}\right)$ for $r = c\eps/\alpha$ and $c,c'>0$ fixed constants; the runtime is polynomial in $n$ with exponent that depends on $\ell$, $\alpha$, and $\eps$.
We prove Theorem~\ref{thm:intro-j} in Section~\ref{sec:johnson}, where we also give more precise quantitative guarantees.

Theorem~\ref{thm:intro-j} suggests that it may be possible to generalize Theorem~\ref{thm:intro-main} to establish that unique games is easy not only on graphs $G$ that are certifiably small-set expanders, but even on graphs that are not small-set expanders but whose expansion profile has some ``nice characterization''  captured by low degree \sos proofs.
Finding a formal notion of such a ``nice characterization'' is an interesting open question that can lead to a general understanding of the easy instances of unique games.
It is also open whether the standard (i.e., non noisy) Boolean cube possesses such a characterization, and indeed it is not known whether constant-degree \sos (or any other polynomial-time algorithm) can solve unique games on the Boolean cube (see \cite{AgarwalKKT15}).

\subsection{Comparison with prior work}  \label{sec:related}

There has been an extensive prior literature on rounding sum-of-squares programs, solving unique games on restricted instances, and relating the unique-games and small-set expansion problems.
In this section we discuss this literature and how our results relate to it. 

\paragraph{Worst-case rounding techniques for higher-degree SoS} 
Our main technical contribution is a new rounding technique that gives a new way to use higher-degree Sum-of-Squares relaxations for worst-case optimization problems. Despite the proliferation of uses of the sum-of-squares method in \emph{average-case algorithm design} (e.g. \cite{BarakKS14,ma2016polynomial}, see recent survey~\cite{TCS-086}), there are relatively few general techniques  that harness the power of the higher-degree SoS relaxations for \emph{worst-case} combinatorial optimization problems. 
The main examples for such techniques are \cite{BarakRS11}'s \emph{global correlation rounding}  and the generalization via \emph{reweightings} developed in~\cite{BKS17}.
In this work we suggest a new way to round solutions to SoS relaxations by considering a potential function that measures the {\em entropy} of the distribution over non-integral SoS solutions via a proxy for the weighted collision probability of each variable's assignment. When the entropy is low (or collision probability is high), we show that it is easy to round to a solution with high value.
We expect that this technique may find applications for other worst-case combinatorial optimization problems, including other CSPs, coloring, and the like.

\paragraph{Solving UG on restricted families of constraint graphs} Our work naturally fits into the long-standing investigation of efficient algorithms for various restricted families of instances of the Unique Games problem including expander graphs~\cite{AroraKKSTV08,MakarychevM10}, perturbed random graphs~\cite{KollaMM11}, and graphs with small ``threshold rank''~\cite{Kolla10,ABS,BarakRS11,GuruswamiSinop11}. 
In addition to yielding new algorithmic techniques and pointing out differences between hard CSPs such as 3-SAT (for which we do not know of any faster algorithm for such restricted families), such works  constitute the best known ``evidence'' against the truth of the UGC. Our guarantees unify and extend these results by noting that each of these restricted families admit (special kinds of) low-degree sum-of-squares certificates of the constraint graph being a small-set expander.\footnote{For low threshold-rank graphs, we get an algorithm as a direct corollary only when they are small-set expanders.}


\paragraph{UG algorithms for general instances}
The best currently-known algorithmic result for general instances of Unique Games is due to~\cite{ABS} and runs in time  $\exp(n^{\poly(\eps)})$ for all $1-\eps$ satisfiable instances. 
This algorithm was shown to be captured by the SoS hierarchy (and extended to apply to other related problems) by \cite{BarakRS11,GuruswamiSinop11} .
The algorithms of \cite{ABS,BarakRS11,GuruswamiSinop11} have better running times when the constraint graph's adjacency matrix has few large eigenvalues: if there are at most $r$ eigenvalues larger than $1-\poly(\eps)$, then they run in time $\exp(r)$.

Our work improves upon the guarantees of \cite{ABS,BarakRS11,GuruswamiSinop11} for instances which have super-logarithmically many large eigenvalues, yet have a constant degree sum-of-squares certificates of hypercontractivity.
In particular, prior to our work, no polynomial-time algorithms were known for unique games instances over the noisy hypercube, noisy short code, and the Johnson graphs - the best known algorithm for the noisy-hypercube ran in quasi-polynomial time and for the noisy-short code ran in subexponential time. 

\paragraph{Noisy Hypercube, Short Code, and Johnson graphs} Starting with Khot and Vishnoi \cite{KhotV15}, the noisy hypercube (and more recently, the shortcode graph) has been intensely studied to construct integrality gaps for natural SDPs for UG.  These works constructed specific instances of unique games over the noisy cube~\cite{RaghavendraS09} and shortcode graph~\cite{BGHMRS15} that on one hand are very far from satisfiable but on the other hand cannot be certified to be so by certain weak SDP and LP hierarchies.
\cite{BarakBHKSZ12} showed that the particular instances of \cite{KhotV15,RaghavendraS09,BGHMRS15} are in fact ``easy'' for \sos in the sense that they can be certified to be far from satisfiable by $O(1)$-degree \sos  (see also \cite{ODonnellZ13}).
However, the analysis of \cite{BarakBHKSZ12,ODonnellZ13} was tailored to the particular instances (specified by both the constraint graph {\em and} the edge constraints) of \cite{KhotV15,RaghavendraS09,BGHMRS15}, and did not yield an algorithm for general instances over these constraint graphs.\footnote{Specifically, \cite{BarakBHKSZ12} ported the analysis of the unsatisfiability proof from the works on integrality gaps into the \sos framework. However, this analysis was specific to the constructed instances. Moreover,  \cite{BarakBHKSZ12} did not provide any \emph{rounding algorithm} and is not directly applicable to analyzing satisfiable instances.}

Our analysis of the \sos algorithm for the Johnson graph (Theorem~\ref{thm:intro-j}) uses structural properties of the Johnson graph closely related to those shown by \cite{KMMS}.
Similar structural properties of the \emph{Grassman graph} have been exploited in the recent works \cite{DinurKKMS18, KhotMS18} to prove the so called ``2-to-2 conjecture''.
This has been a recurring motif in works on unique games. In the noisy hypercube, short code, and now in the Johnson graph, structure that was exploited to prove soundness for reductions was later found useful in giving efficient algorithms for the same instances.

\paragraph{Reductions from UGC to SSEH}
Raghavendra, Steurer and Tulsiani \cite{RaghavendraST12}  (building on \cite{RaghavendraS10}) reduced the task of solving unique-games on small-set expanders to the  small-set expansion problem  (see also  \cite[Chap.~6]{Steurer-thesis}). 
Theorem~\ref{thm:intro-main} can be viewed as a ``point-wise'' version of their reduction for integrality gaps.
Specifically,  \cite{RaghavendraST12}  gave a reduction which maps any unique-game instance $(\Pi,G)$ (where $G$ is a small set expander), into an instance $G'$ of the small-set expansion problem, where $G'$ is polynomially larger than $G$.
In contrast, Theorem~\ref{thm:intro-main}  shows that for \emph{every} graph $G$, if $O(1)$-degree \sos certifies the small-set expansion of $G$ then $O(1)$-degree \sos can also approximate unique games instances over \emph{the same graph} $G$, which implies that if $G$ a small-set expander and an integrality gap instance  for the Degree $D$ SoS relaxation of UG, then $G$ is also an integrality gap instance for the Degree $\Omega(D)$ SoS relaxation for SSE. Our algorithm for the Johnson graph suggests that there might be a way to extend this result to a reduction from UG to SSE even when the constraint graph is \emph{not} a small set expander but whose expansion profile has a nice characterization captured by SoS proofs, and hence is an easy instance of the small set expansion problem.

\section*{Organization}

In Section~\ref{sec:alg}, we give a high-level overview of our algorithm and our proofs.
In Section~\ref{sec:low-ent}, we prove that if a certain potential function in the sum-of-squares relaxation has large value, then a simple algorithm produces assignments of value $\Omega(1)$.
In Section~\ref{sec:cert-sse} we prove that this potential is always large for certifiable small-set expanders, and in Section~\ref{sec:hyper} we derive corollaries for the hypercube and short code graphs.
Finally, in Section~\ref{sec:johnson} we give the proof of Theorem~\ref{thm:intro-j} for the Johnson graph.
Section~\ref{sec:apx-ind} describes low-degree polynomials that approximate step functions, which we employ to define our potential.
Appendix~\ref{sec:prelims} contains background on \sos, Appendix~\ref{sec:expansion-red} reproduces for completeness a proof of a lemma relating small-set expansion to hypercontractivity, and Appendix~\ref{sec:fourier} contains \sos proofs of structural properties of Johnson graphs.

\section*{Preliminaries and Notation}
For a (weighted) graph $G=(V,E)$, we use $(u,v) \sim E$ to denote an edge $(u,v)$ sampled with probability proportional to its weight. We use $A_G$ to denote the transition matrix of the random walk on $G$, $L_G = I - A_G$ to denote the Laplacian and $\pi_G$ to denote the corresponding stationary distribution over $V$ (we take $\pi_G$ to be the distribution where each vertex is sampled proportional to the sum of weights on its incident edges \footnote{$A_G$ might not have a unique stationary measure, for instance when $G$ is bipartite or disconnected, but $\pi_G$ is always a stationary measure of $A_G$.}); we will drop the subscript when $G$ is clear from context. It is easy to see that picking a random edge from $E$, is equivalent to picking a random vertex $v \sim \pi$ and a random neighbor $w$ of $v$ with probability proportional to the weight of the edge $(w,v)$. For $v \in V(G)$, we use $\deg_G(v)$ to denote $v$'s (weighted) degree inside $G$.
If $A$ is some probabilistic event or condition, we use $\Ind(A)$ to denote the indicator random variable of $A$ (i.e., $\Ind(A)=1$ if $A$ occurs and $\Ind(A)=0$ otherwise). 

\begin{definition}[Unique games] \label{def:uniguegames}
A \emph{unique games instance} is a pair $I=(G,\Pi)$ where $G=(V,E)$ is a graph and $\Pi$ is a collection $\{ \pi_{u,v} \}_{(u,v) \in E}$ such that $\pi_{u,v}$ is a permutation over some finite set $\Sigma$. The graph $G$ is known as the \emph{constraint graph} of $I$.

Given an instance $I=(G,\Pi)$ of unique games and an assignment $x\in \Sigma^V$ of values to the vertices of $G=(V,E)$, the   \emph{value} of $x$ with respect to $I$ is $\val_I(x) = \E_{(u,v) \sim E} \Ind(\pi_{u,v}(x_u) = x_v)$.
 The value of $I$ is the maximum of $\val_I(x)$ over all $x\in \Sigma^V$.
We may drop the subscript $I$ when the instance is clear from context.

We say that $(G,\Pi)$ is an \emph{affine} unique games instance if $\Sigma$ is an additive group and all the functions $\pi_{u,v}$ are of the form $\pi_{u,v}(x) = x -  a_{u,v}$ for some $a_{u,v} \in \Sigma$. That is, all constraints correspond to $x_u - x_v = a_{u,v}$. 
\end{definition}

It is known that the UGC is equivalent to its restriction on affine instances~\cite{KhotKMO07}. 
In this paper we restrict attention to affine instances only.
For the sake of simplicity, we will drop the qualifier ``affine'' in future discussion, but all of our results are for this family of constraints.

\paragraph{Sum of squares proofs.}
Given a set of axioms $\cA = \{q_i = 0\}_i \cup \{g_j \ge 0\}_j$ for polynomials $q_i,g_j \in \R[x]$, we say that ``there is a degree-$d$ sum-of-squares proof that $f \ge h$ modulo $\cA$'' if $f = h + s + \sum_i c_i\cdot q_i + \sum_j r_j \cdot g_j$ with real polynomials $s,\{c_i\}_i,\{r_j\}_j \in \R[x]$ such that $s$ and $\{r_j\}_j$ are sums of squares, and if the maximum degree among $s,\{c_i q_i\}_i, \{r_j g_j\}_j$ is at most $d$. 
We will use the notation $\cA \vdash_d f(x) \ge h(x)$ to denote the existence of such an equality.
We also sometimes use $f(x) \succeq h(x)$ to denote that the inequality is a \sos inequality.
See Appendix~\ref{app:sos} for more.

\paragraph{Other notation.} We use the standard big-$O$ and big-$\Omega$ notation.
We will also use $f = \tilde{O}(x)$ to denote that there exists some $c,C$ independent of $x$ such that $\lim_{x \to \infty} \frac{f}{C x \log^c x} \le 1$.
For a positive integer $k$, we denote $[k] = \{1,\ldots,k\}$ and  $\binom{S}{\ell}$ to denote the set of unordered simple $\ell$-element subsets of $S$.
For a vector of variables $x$, we let $x^{\le D}$ denote the set of monomials of degree at most $D$ in the variables.
For a measure $\pi$ on $S$ and $f,g:S\to\R$, we use $\langle f, g \rangle_\pi = \E_{v \sim \pi} f(v) g(v)$ and the corresponding $p$-norms $\|f\|_{\pi,p} = \left(\E_{v \sim \pi} |f(v)|^p\right)^{1/p}$.
For a function $f(x)$ and $k \in \R$, we will use $f^{\circ k}(x) = f(x)^k$ to denote the element-wise $k$-th power of $f$.

\section{Overview of our techniques}\label{sec:alg}

We now describe our algorithm and give an overview of its analysis.
Our algorithm is based on the \sos semidefinite programming (SDP) relaxation, and in particular its view as optimizing over \emph{pseudo expectation} operators, see the surveys \cite{BarakS14,RaghavendraSS18,TCS-086} and Appendix~\ref{app:sos}.

Given a unique games instance $I= (G,\Pi)$ over alphabet $\Sigma$, with $G=(V,E)$, the value of $I$ can be computed by the following integer program over zero-one variables $\{ X_{u,a} \}_{u \in V, a \in \Sigma}$:

\begin{align}
\max_{X}& \E_{(u,v) \in E} \sum_{a \in \Sigma} X_{u,a} X_{v,\pi_{uv}(a)} \label{eq:ip}\\
s.t. 
&\quad X_{u,a}^2 = X_{u,a} \quad \forall u \in V, a \in \Sigma \qquad\nonumber\\
&\quad X_{u,a}X_{u,b} = 0 \quad \forall u \in V, a \neq b \in \Sigma \qquad\nonumber\\
&\quad \sum_{a} X_{u,a} = 1 \quad \forall u \in V\nonumber
\end{align}

The variables $X_{i,a}$ are the $0/1$ indicator variables that vertex $i \in V$ takes label $a \in \Sigma$. 
The objective function asks us to maximize the fraction of edge constraints satisfied.
Our algorithm is obtained by considering the  degree $D=O(1)$ \sos relaxation of the above program, obtaining a pseudo-expectation operator $\pE:X^{\le D} \to \R$, where $X^{\le D}$ is the set of all monomials in the $X$ variables up to degree $D$, and $\pE$ satisfies the above equality constraints and the Booleanity constraints $\{X_{u,a}^2 = X_{u,a}\}$ as axioms.  For brevity, we will refer to this set of axioms as $\A_I$, dropping the subscript when $I$ is clear from context.
The \emph{value} of such a pseudo-expectation operator whose corresponding pseudodistribution is $\mu$, with respect to the instance $I$ is denoted by $\val_\mu(I) = \pE[\val_I(X)] = \pE [ \E_{(u,v) \in E} \sum_{a \in \Sigma} X_{u,a} X_{v,\pi_{uv}(a)}]$. (Note that this is the pseudo expectation of a degree two polynomial in the variables $\{ X_{u,a} \}$.)

\subsection{Our rounding algorithm}

The \sos SDP relaxation is standard, and the novelty of our work is in the \emph{rounding} algorithm for it.
A $(1-\eps,\delta)$ rounding algorithm for the \sos relaxation is an algorithm that takes as an input an instance $I=(G,\Pi)$ and a pseudo-expectation operator $\pE$ (satisfying $\A_I$) of value at least $1-\eps$  and outputs an assignment $x\in \Sigma^V$ with $\val_I(x) \geq \delta$.
In this paper (and in the context of the UGC in general) we are interested in finding $(1-\eps,\delta)$ rounding algorithms for $\eps, \delta$ that are bounded away from zero by some constant which is independent of the alphabet size $|\Sigma|$.

Our rounding algorithm can be described as follows.
We will define some low-degree polynomial $\Phi^{I}_{\eps}:\R^{V \times \Sigma} \rightarrow [0,\infty)$ (which we call the ``approximate shift partition potential'' for reasons explained below).
We then show (roughly speaking) the following three statements:

\begin{enumerate}
    \item  There is a rounding algorithm that given an instance $I$ and a pseudo-expectation operator $\pE$ such that $\pE[\val_I(X)] \ge 1-\eps$ and $\pE [\Phi^I_{\eps}(X)] \geq \delta$, outputs an assignment $x$ for $I$ with $\val_I(x) \geq \poly(\eps,\delta)$.
    
    \item For every $I=(G,\Pi)$, if $G$ is a $(\delta,100\eps)$-small-set expander,\footnote{That is, every set of $G$ with size at most $\delta$ has expansion at least $100\eps$.} and if $X$ is a random variable sampled from an actual distribution over vectors in $\{0,1\}^{V \times \Sigma}$ with expected value $1-\eps$ for the integer program (\ref{eq:ip}), then $\E[ \Phi^{I}_{\eps}(X)] \geq \poly(\delta)$. 
    
    \item There is an $O(1)$-degree \sos \emph{proof} for Statement 2. 
\end{enumerate}

Using the standard ``\sos paradigm,'' the three steps above suffice to obtain algorithms for graphs that are certifiably small set expanders.
For such graphs we can combine the expansion certificate with the \sos proof of Statement 2 to show that any pseudo-distribution over $X$ obtained as a solution of the \sos program will have to satisfy $\pE \Phi^{I}_{\eps} \geq \Omega(1)$ and hence use the algorithm from Statement 1 to obtain an actual solution with value bounded away from zero.

In the case of the Johnson graph, which is not a small set expander, we have to work harder.
In this case we use the characterization of non expanding sets in the Johnson graph to show that if the value is sufficiently large then the potential $\Phi^{I}_{\eps}$ must be large on some ($o(1)$-sized) \emph{subgraph} of the Johnson graph (itself a Johnson graph with different parameters).
We solve for a partial assignment on this subgraph and iterate, and we are able to show that this process can continue until we have obtained an assignment with value independent of the alphabet size. 

\subsection{Rounding for certified small set expanders}

Since our algorithm for the Johnson graph is more complex, we will start by describing our algorithm for certified small set expanders.
In this section we will focus on the case that the pseudo expectation operator corresponds to an \emph{actual distribution} and the graph $G$ is simply a small set expander (with or without a certificate).
This case is sufficient to illustrate the main ideas behind our algorithm. The full analysis is presented in Sections~\ref{sec:low-ent}~and~\ref{sec:cert-sse}.

Throughout this section we fix an instance $I=(G,\Pi)$ of unique games, with $G=(V,E)$.
We let $\mu$ be a distribution over strings $X$ in $\{0,1\}^{V \times \Sigma}$ satisfying the constraints $\A_I$.
We will also identify $X$ with assignments in $\Sigma^V$ and so write $X_u$ for the unique element $s\in \Sigma$ such that $X_{u,s}=1$.

For every vertex $u\in V$ and symbol $s \in \Sigma$, we define the following random variable
\[
Z_{u,s} = \sum_{a \in \Sigma} X_{u,a} X'_{u,a+s} = \Ind(X_u - X_u' = s) \;,
\]
where $X$ and $X'$ are two independent samples from the distribution.\footnote{Given a degree $D$ pseudo-expectation operator corresponding to some pseudodistribution $X$, we can find in linear time a degree-$D/2$ pseudodistribution that satisfies the constraints corresponding to taking two independent samples of $X$. 
See Appendix~\ref{sec:prelims} and Fact~\ref{fact:indep}.}

We think of $Z_s$ as a subset of $V$, with $Z_{u,s}$ as the indicator variable for the membership of vertex $u$ in $Z_s$.
The $Z_{u,s}$'s satisfy partition constraints, hence they induce a partition of the graph into components on which the solutions $X,X'$ agree up to a shift, so that $Z_{u,s} = 1$ when $X_{u} - X_{u}' = s$.
We refer to this partition as the ``shift partition.''
If we were to assign labels to the vertices arbitrarily, then each part $Z_s$ in the partition would have size roughly $\approx \tfrac{1}{|\Sigma|}$.
On the other hand, if there is a part in the partition of fractional size $\Omega(1)$, this means the labels of two independent assignments are more correlated than one would expect, in that they agree up to shift on a non-trivial fraction of vertices.
This inspires our potential function.

We start by considering the following simplified version of our potential function:

\begin{definition}\label{def:sq-mass}
For any $\beta \in (0,1)$, define the {\em shift-partition potential} to be the quantity
\[
\Phi_\beta(X,X') = \sum_{s \in \Sigma} \left(\E_{u} \left(Z_{u,s} \cdot \Ind(\val_u(X)\ge \beta)\right)\right)^2,
\]
for $\val_u(X)$ the ``local objective'' at $u$, $\val_u(X) = \E_{v \sim u} \sum_{a \in \Sigma} X_{u,a}X_{v,\pi_{uv}(a)}$ where $v \sim u$ denotes a neighbor of $u$ sampled according to the edge weight of $(u,v)$.
\end{definition}
This potential measures the average square size of components in the shift partition, where the indicator ensures that we only include vertices which satisfy at least a $\beta$ fraction of incident edges. 
A convenient parameter setting will be to take $\beta = \eps$.

\paragraph{Rounding from high shift partition potential.}
If $\mu$ is an actual distribution with respect to an instance $I$,  and  $\E_\mu \Phi_\eps(X,X') \geq \Omega(1)$, then the following simple algorithm (see also Algorithm~\ref{alg:low-ent}) will find in expectation an assignment $y$ for $I$ with $\val_I(y) \geq \Omega(1)$:

\begin{enumerate}
    \item Pick $u_0\in V$ uniformly at random.
    
    \item Sample $y_1,\ldots,y_V \in \Sigma$ independently by letting $\Pr[ y_u = a] =  \E[ X_{u,a} | X_{u_0,0}=1 ]$. (That is, $y$ is sampled from the product distributions whose marginals correspond to $X | X_{u_0}= 0$.)
    
\end{enumerate}

The intuition behind the above is as follows: When $\E[\Phi_{\eps}(X,X')] \ge \delta$, then for a ``typical'' pair of independent assignments drawn from $\mu$, there will be a subset of vertices $S$ of measure $\ge \Omega(\delta)$ on which $X$ and $X'$ agree up to a shift in $\Sigma$. This implies that a random pair of vertices, will satisfy that the collision probability of the random variable $(X_v - X_u)$, i.e. $\Pr[X_u - X_v = X'_u - X'_v]$, is at least $\delta$. Since we have symmetry over the labels, the distribution of $X_v - X_u$ is the same as the distribution of $(X_v | X_u = 0)$, hence we get that the collision probability of $(X_v | X_u = 0)$ is high. Since we now choose $u_0$ at random and condition the distribution on $X_{u_0} = 0$, we have that the marginal distribution of a random vertex has high collision probability. If in addition the value on the vertices with high collision probability is close to $1$, we would immediately get that independent rounding gives high value on these vertices; this is because if a vertex's value \emph{and} collision probability are high, the vertex's neighbors must also have high collision probability on the corresponding satisfying labels. Following this logic, independent rounding will satisfy a $\Omega(\delta^2)$-fraction of the edges incident on these vertices, so that in total we satisfy at least a $\Omega(\delta^3)$-fraction of the edges of the graph. The term $\Ind(\val_u(X) \geq \eps)$ in the function $\Phi_\eps$ ensures that the high-collision-probability vertices also correspond to high value vertices (since those vertices that have low value do not even contribute to the potential), and this suffices for us to make the above intuition go through. This argument is made formal in Section~\ref{sec:low-ent} (see Algorithm~\ref{alg:low-ent} and Theorem~\ref{thm:round}).

\paragraph{Low degree polynomials.}
The function $\Phi_\beta$ above cannot be used for rounding pseudo-expectation operators, because it is not a low degree polynomial in the variables $X$.
To tackle this issue, we introduce the {\em approximate shift-partition potential}, replacing the high-degree indicator  $\Ind(\val_u(X) \ge \beta)$  with an approximating low-degree polynomial: 
\begin{definition}\label{def:sq-mass-apx}
For any $\nu,\beta \in (0,1)$, define the {\em approximate shift-partition potential} to be the quantity
\[
\Phi_{\beta,\nu}(X,X') = \sum_{s \in \Sigma} \left(\E_{u} \left(Z_{u,s} \cdot p_{\beta,\nu}(\val_u(X))\right)\right)^2,
\]
for $p_{\beta,\nu}(x)$ the degree-$\tilde{O}(1/\nu)$ polynomial which SoS-certifiably $\nu$-approximates the indicator $\Ind[x\ge\beta]$ for $x \in [0,1]$ described in Theorem~\ref{thm:step-approx}.
\end{definition}

The function $\Phi^{I}_{\eps}$ will be set as $\Phi_{\beta,\nu}$ for a suitable parameter setting $\beta = \eps$ and $\nu = \poly(\eps)$.

\paragraph{Small set expansion and the shift partition potential.}
The $Z_{u,s}$ variables define a partition of the graph.
Edges which cross this partition cannot be satisfied in {\em both} $X$ and $X'$ variables, since in an affine UG instance the labels of a satisfied edge's endpoints agree up to a shift: if $(u,v)$ is an edge with $u$ in the $s$ shift component (that is, $X_u = X'_u + s$), and $v$ in the $t$ shift component ($X_v = X_v' + t$), then $X_{u} - X_{v} \neq X_{u}' - X_v'$ unless $s = t$.
Therefore, the shift partition corresponds to a partition induced by removing the (on average) $\le 2\eps$ fraction of edges that are unsatisfied in at least one of the two solutions, $X$ or $X'$.
This means that if $G$ is a $(\delta,100\eps)$-small-set expander, then on average the partition induced by $Z_{u,s}$ has parts of $\Omega(\delta)$ size.
Since in an assignment of value $1-\eps$ there are at most $O(\eps)$ vertices with local objective $\le \eps$, removing such vertices by introducing the indicators $\Ind[\val_u(X) \ge \eps]$ removes at most $O(\eps)$ edges and therefore the above reasoning is unaffected: the parts remain of size $\Omega(\delta)$, so that $\E[\Phi_\eps] = \Omega(\delta)$.
We make this intuition formal in Section~\ref{sec:cert-sse} (see Theorem~\ref{thm:main-cert} and its proof).

\subsection{Johnson Graphs}

The Johnson Graph is {\em not} a small-set expander.
However, we are able to use its spectral structure  to obtain a nontrivial approximation ratio. 
We start by formally defining this graph:

\begin{definition}[Johnson Graph] \label{def:johnson}
For any $1 > \alpha > 0$ and $\ell,q \in \N$ with $\alpha\ell \in \N$ and $n > \ell$, we define the \emph{$(n,\ell,\alpha)$-Johnson graph $J_{n,\ell,\alpha}$} to be the graph whose vertex set is $\binom{[n]}{\ell}$ and where edges are between pairs of vertices $U,V \in \binom{[n]}{\ell}$ if and only if $|U \cap V| = (1-\alpha)\ell$. We refer to $\alpha$ as the noise parameter (analogous to the $\alpha$-noisy hypercube).
\end{definition}

The $(n,\ell,\alpha)$-Johnson graph contains other Johnson graphs as subgraphs: consider the subgraph induced by vertices which contain some $S \subset [n]$ with $|S| < \ell$.
We call such subgraphs {\em $|S|$-restricted subcubes}. 
It is not hard to see that such an $r$-restricted subcube contains at least an $\eta := (1-\alpha)^{r}$ fraction of its incident edges---this is because neighbors $(U,V)$ differ in each element with probability $\approx \alpha$, and so for a random neighbor $V$ of $U$, the chance that none of the elements of $S$ are changed is $\approx (1-\alpha)^{|S|}$.
Notice that when $r <  O(\frac{c}{\alpha})$ and $r \ll \ell$, the fraction of internal edges in an $r$-restricted subcubes is at least $\eta \ge 1 - O(c)$.

$\cite{KMMS}$ showed that in the Johnson graph, every non-expanding set that has expansion $\eps$ is correlated with some $r$-restricted subcube, for $r = O(\eps/\alpha)$, that has expansion $O(\eps)$. We show a ``distribution-version'' of this theorem: for any distribution over non-expanding sets, there exists an $r$-restricted subcube that is correlated with these sets in expectation. Moreover, we give an SoS proof of this fact (Theorem~\ref{thm:structure-johnson}), so that the same statement holds for pseudodistributions too. 

We then use this structure theorem to show that given a high value pseudodistribution for a unique games instance $I$, there must exist at least one $r$-restricted subcube, so that the approximate shift partition potential restricted to that subcube is high.

\begin{lemma*}[Large potential on a subcube: special case of Lemma~\ref{lem:round-partial-johnson}]
If $I$ is a unique games instance on the $(n,\ell,\alpha)$-Johnson graph and $X$ is sampled from a distribution over solutions with $\E[\val_I(X)] \ge 1-\eps$, then there exists an $O(\frac{\eps}{\alpha})$-restricted subcube $C$ such that the expected shift potential of the subgraph induced by $C$ is at least $\delta = \delta(\ell,\eps,\alpha)>0$.  Furthermore, this is certifiable by a degree-$\tilde{O}(1/\delta)$ \sos proof.
\end{lemma*}

\medskip

The Johnson graph only has $\binom{n}{r} \leq \binom{n}{\ell}$ $r$-restricted subcubes, and so in $n^r = \poly(n)$ time we can enumerate over the cubes to find one cube $C$ with a large shift-partition potential (i.e., satisfying  $\E[\Phi_{\eps}^{C}] \ge \delta$). We can then find a $\delta$-satisfying solution for the internal edges of $C$ by using our rounding algorithm (Algorithm~\ref{alg:low-ent}). Since 
the fractional mass of $C$, $\mu(C) := \binom{n}{\ell - r}/\binom{n}{\ell} \approx \frac{\ell^r}{n^r}$, we only satisfy a negligible fraction of edges this way.
On the other hand, since $C$ is just a $o(1)$-fraction of the graph, the unique games instance restricted to the rest of the graph $\overline{C}$ must have high value too.
Since the value remains high even after removing $C$, we may iteratively repeat this process to find a sequence of $r$-restricted subcubes $C_1,\ldots,C_T$, while ensuring that each cube $C_t$ does not intersect too much with the previous subcubes $C_1,\ldots,C_{t-1}$. 
At each iteration, we fix an assignment on $C_t$ satisfy an $\Omega(\delta^2)$-fraction of $C_t$'s internal edges, which in turn is an $\Omega(\delta^2\eta)$-fraction of all edges incident on $C_t$; the remaining $(1-\delta^2\eta)$ fraction of edges incident on the cube (including outgoing edges) may be unsatisfied. 
But since the ratio of satisfied to unsatisfied edges incident on $C_t$ is at least $\delta^2\eta$, the objective value drop (on the unassigned part of the graph) in every step is proportional to the fraction of edges we satisfy in that step. We repeat the process until the value drops by $\eps$, so we end up satisfying an $\Omega(\delta^2\eta\eps)$-fraction of all the edges.

Modulo the proof of the ``large potential on subcube'' Lemma (which will be a corollary of Lemma~\ref{lem:round-partial-johnson}), this is nearly the complete argument. 
The only detail that remains is to apply the above lemma iteratively (we cannot simply apply it on $J \setminus C$ since that graph is not a Johnson graph) and to ensure that the subcubes we find at each iteration do not overlap too much.
To handle both these issues, as we iterate we take additional measures. The full proof is in Section~\ref{sec:johnson}; see Algorithm~\ref{alg:j} and Theorem~\ref{thm:main-johnson}.

\section{Rounding instances with large shift potential}
\label{sec:low-ent}

In this section, we will show that when the objective value is large and the approximate-shift-partition potential $\Phi$ has large pseudoexpectation, then the Condition \& Round Algorithm  (Algorithm~\ref{alg:low-ent}) succeeds in returning a good assignment for the unique games instance.

\begin{algorithm-thm}[Condition \& Round]\label{alg:low-ent} ~\\ 
\textbf{Input:} A degree-$D$ (for $D \geq 2$) 
shift-symmetric pseudodistribution\footnote{Any pseudodistribution can be efficiently transformed into a shift-symmetric one without losing value. See Definition~\ref{def:sym} and Lemma~\ref{lem:sym-basic} for details.} $\mu$ for a UG  instance $I=(G=(V,E),\Pi)$ over alphabet $\Sigma$.  \\
\textbf{Goal:}  Return an assignment $x \in \Sigma^V$ satisfying $\Omega(1)$ fraction of the constraints in expectation.
\\\\
Sample a random solution $Y$:
\begin{compactenum}
\item Sample a vertex $u \sim \pi$ and condition on $X_{u} = 0$ to obtain the new marginals $\pE_\mu[\cdot ~|~ X_{u} = 0]$.
\item Sample a solution $Y$ by choosing each collapsed variable's labels independently according to its marginals: $Y_v \sim \pE_\mu[X_v ~|~ X_{u} = 0]$.
\end{compactenum}
\end{algorithm-thm}

Recall the approximate shift-mass potential $\Phi_{\beta,\nu}(X,X')$ from Definition~\ref{def:sq-mass-apx}.
We define the potential of a pseudo distribution $\mu$ to be the expectation of $\Phi_{\beta,\nu}$ over $\mu$:

\begin{definition}[Approximate shift mass potential of a pseudodistribution]
For a pseudodistribution $\mu$ of degree at least $2\deg(\Phi_{\beta,\nu}) + 2$, define the approximate shift mass potential of $\mu$ to be the quantity
\[
\Phi_{\beta,\nu}(\mu) = \pE_\mu[\Phi_{\beta,\nu}(X,X')].
\]
\end{definition}

We will prove the following theorem:

\begin{theorem}\label{thm:round}
Let $I = (G, \Pi)$ be an affine instance of Unique Games over the alphabet $\Sigma$. 
Let $\mu$ be a  degree-$O(\deg(\Phi_{\beta,\nu})))$  shift-symmetric pseudodistribution satisfying the axioms $\A_I$ specified by program (\ref{eq:ip}).
If $\Phi_{\beta,\nu}(\mu) \geq \delta$, then on input $\mu$ Algorithm~\ref{alg:low-ent} runs in time $\poly(|V(G)|)$ and returns an assignment of expected value at least $(\delta - \nu)(\beta-\nu)$ for $I$.
\end{theorem}

While Algorithm~\ref{alg:low-ent} is randomized, we can derandomize it and obtain a deterministic polynomial-time algorithm with the same guarantee on the approximation factor. To derandomize we can use standard techniques such as the method of conditional expectations~\cite{Vadhan12}. We will refer to such an algorithm as derandomized Condition \& Round.

\begin{proof}[Proof of Theorem~\ref{thm:round}]
Throughout this proof, we let $\mu$ be a pseudo-distribution satisfying the conditions of the theorem, and all pseudo-expectations are taken with respect to $\mu$.
Our overall strategy will be as follows: we will define an alternate potential function $\Psi(\mu)$, relate its value to $\Phi(\mu)$, and then show that when $\Psi(\mu)$ is large a single step of conditioning and independent rounding gives a large expected objective value.

To define our alternate potential, let us introduce some concise notation. 
For an event $\cE$ whose indicator $\Ind(\cE)$ has degree at most $\deg(\mu)$ define $\pPr[\cE] = \pE[\Ind(\cE)]$ (see Definition~\ref{def:pseudoprob} for a formal definition). 
Similarly, for conditional probabilities, for events $\cE$ and  $\cF$ with $\deg(\Ind(\cE \wedge \cF))\le \deg(\mu)$, let $\pPr[\cE \mid \cF] := \frac{\pPr[\cE \wedge \cF]}{\pPr[\cF]}$. 
For simplicity of notation, when $\pPr[\cF] = 0$, we define $\pPr[\cE \mid \cF] := 0$.

Now we define the conditioned shift potential $\Psi(\mu)$:
\begin{definition}
The {\em conditioned shift potential} of a degree-$D \ge 4$ pseudodistribution $\mu$ is given by
\[
\Psi(\mu) := \E_{u,v \sim \pi} \left[\sum_{s \in \Sigma} \pPr_\mu[X_v - X_u = s]^2 \cdot \pE[\val_v(X) \mid X_v - X_u = s] \right],
\]
where $\pi$ is the stationary measure on $G$ and $\val_v(X)$ is the ``local objective'' at the vertex $v$, $\val_v(X) = \E_{w \sim v}[\pPr_\mu[X \text{ satisfies } (v,w)]]$ for $w \sim v$ a  neighbor of $v$ sampled proportional to the weight on $(v,w)$.
\end{definition}

Roughly, the conditioned shift potential measures the average collision probability of the random variable $(X_u - X_v)$, but it gives more preference to those pairs $(u,v)$ that have high local objective value. 

We will show that when $\Phi(\mu)$ is large, then $\Psi(\mu)$ is also large:
\begin{lemma}\torestate{\label{lem:relating-ent}
If the approximate shift mass potential of $\mu$ is large, then the conditioned shift potential of $\mu$ must be large as well:
\[\Phi_{\beta,\nu}(\mu) \leq \frac{\Psi(\mu)}{\beta-\nu} + \nu.\]}
\end{lemma}
We prove this lemma in Section~\ref{sec:lems} below.
Next, we will show that when the conditioned shift potential is large, a single step of conditioning and rounding returns a solution of high objective value:
\begin{lemma}\torestate{\label{lem:rounding}
Let $I = (G, \Pi)$ be an affine instance of Unique Games over the alphabet $\Sigma$. 
Let $\mu$ be a degree-$4$ shift-symmetric pseudodistribution for $I$. When $\Psi(\mu) \geq \delta$, then the Condition \& Round algorithm (Algorithm~\ref{alg:low-ent}) returns a solution of expected value at least $\delta$.}
\end{lemma}

We prove this lemma below in Section~\ref{sec:lems} as well.
Given the two lemmas, the first statement of the theorem clearly follows. 
\end{proof}

\subsection{Relating the potentials and rounding}\label{sec:lems}
In this section, we will prove Lemmas~\ref{lem:relating-ent}~and~\ref{lem:rounding}.
Before we dive into these lemmas, let us define a symmetrization operation on pseudodistributions. Intuitively it makes sense
for a pseudodistribution on an affine unique games instance $I$ to be symmetric with respect to shifts, since if $X$ is a $(1-\eps)$-satisfying solution for $I$, then so is $X_{+s}$ for all $s \in \Sigma$. Pseudodistributions obtained by symmetrization will satisfy useful symmetry properties that are amenable to the analysis of Algorithm~\ref{alg:low-ent}.

\begin{definition}[Symmetrization]\label{def:sym}
Given a pseudodistribution $\mu$, we define the corresponding symmetrized pseudodistribution $\mu_{sym}$ as: For each $s \in \Sigma$, define $\mu_{+s}$ to be the pseudodistribution in which the labels receive the global affine shift $+s$, so that 
\[
\pE_{\mu_{+s}}[X_{u_1,a_1} \cdots X_{u_t,a_t}]
:= \pE_{\mu}[X_{u_1,a_1 - s}\cdots X_{u_t,a_t -s}]
\]
for all $\{(u_1,a_1), \ldots, (u_t,a_t)\} \in ([n] \times \Sigma)^{\le D}$. Now, define $\mu_{sym}$ to be the uniform mixture over $\mu_{+s}$ with $s \in \Sigma$. We say that a pseudodistribution is shift-symmetric if it is invariant under the symmetrization operation defined above, that is, $\mu = \mu_{sym}$.
\end{definition}

Firstly note that this operation can be efficiently performed on $\mu$. Furthermore it yields a valid pseudodistribution that has the same value as $\mu$.

\begin{lemma}[Symmetrization]\label{lem:sym-basic}
Let $\mu$ be a degree-$D$ pseudodistribution satisfying the unique games axioms $\A_I$ given by (\ref{eq:ip}) for an affine unique games instance $I$. Let $\mu_{sym}$ be a pseudoexpectation operator obtained by symmetrizing $\mu$. Then we have that,
\begin{enumerate}
\item $\mu_{sym}$ is a valid pseudoexpectation operator of degree-$D$ that satisfies the unique games axioms $\cA_I$. 
\item The time taken to perform symmetrization on $\mu$ is subquadratic in the description of $\mu$.
\item The objective value of $\mu$ and $\mu_{sym}$ are equal, i.e. $\val_\mu(I) = \val_{\mu_{sym}}(I)$.
\end{enumerate}
\end{lemma}

The proof of this lemma is fairly straightforward, so we omit it. Since the value is invariant under symmetrization and performing the operation is efficient, all our algorithms symmetrize the pseudodistributions obtained by solving the degree $D$ SoS relaxation, and hence in our analysis we always work with shift-symmetric pseudodistributions henceforth.

Symmetrized distributions satisfy some nice symmetry properties with respect to shifts, such as, every vertex has uniform marginals, and value of $\mu$ conditioned on $X_u = s$ for any shift $s$, is the same as the original value without conditioning. Additionally we have the following:

\begin{lemma}[Shift-Symmetry properties]\torestate{\label{lem:sym}
Let $\mu$ be a degree-$D$ shift-symmetric pseudodistribution satisfying the unique games axioms $\A_I$ given by (\ref{eq:ip}) for an affine unique games instance $I$.
Then $\mu$ satisfies the following symmetry properties:
\begin{enumerate}
    \item For all vertices $u,v \in V(G)$ and shifts $s \in \Sigma$,
    $\pPr[X_v = s \mid X_u = 0] = \pPr[X_v - X_u = s]$.
    \item For all polynomials $f(X)$ with $\deg(f) \le D-2$, such that $f(X) = f(X+s)$ for every global shift $s \in \Sigma$,
    \[\pE[f(X) \mid X_v-X_u=s] = \pE[f(X) \mid X_u = 0, X_v = s].\]
\end{enumerate}}
\end{lemma}
This lemma follows easily from the fact that $\mu$ is invariant under global shifts. See Appendix~\ref{sec:prelims} for a proof. 

We first prove that when the potential $\Psi$ is large, conditioning and then independently rounding succeeds.
\restatelemma{lem:rounding}
\begin{proof}
Suppose that $\Psi(\mu) \geq \delta$.
Define the following,
\[
\Psi_{u}(\mu) := \E_{v \sim \pi} \left[\sum_{s \in \Sigma} \pPr[X_v - X_u = s]^2 \cdot \pE[\val_v(X) \mid X_v - X_u = s] \right],
\]
so that $\Psi(\mu) = \E_{u \sim \pi}[\Psi_u(\mu)]$. Suppose we sample a random vertex $u \sim \pi$ and condition the pseudodistribution on $X_u = 0$, then pick a random label $Y^u_v$ for every vertex $v \in V(G)$ according to its marginal $Y^u_v \sim \pE[X_v \mid X_u = 0]$.
We have that in expectation, after conditioning on $u$ the rounded value is equal to:
\[
\E_{Y^u}[\val(Y^u)] = \E_{v\sim \pi}\E_{w\sim v}\left[\sum_s \pPr[X_v = s \mid X_u = 0]\pPr[X_w = \pi_{vw}(s) \mid X_u = 0] \right].\]

We will now lower bound this quantity by $\Psi_u(\mu)$. 
We have that
\begin{align*}
\Psi_{u}(\mu)
&= \E_{v \sim \pi} \left[\sum_{s \in \Sigma} \pPr[X_v - X_u = s]^2 \cdot \pE[\val_v(X) \mid X_v - X_u = s] \right] \\
&= \E_{v} \left[\sum_{s \in \Sigma} \pPr[X_v = s | X_u = 0]^2 \cdot \pE[\val_v(X) \mid X_u = 0, X_v = s] \right] 
\intertext{where we have applied Lemma~\ref{lem:sym} along with the shift-symmetry of $\mu$ and of $\val_u(X)$, where the latter is a shift-symmetric function because $I$ is an affine unique games instance. Now, by definition of the local value,}
&=\E_{v\sim \pi} \left[\sum_{s \in \Sigma} \pPr[X_v = s | X_u = 0]^2 \cdot \E_{w \sim v}\left[\pPr[X \text{ satisfies }(v,w) \mid X_u = 0, X_v = s]\right] \right] \\
&= \E_{v \sim \pi} \left[\sum_{s \in \Sigma} \pPr[X_v = s | X_u = 0]^2 \cdot \E_{w \sim v}\left[\frac{\pPr[X_v = s, X_w = \pi_{vw}(s) \mid X_u = 0]}{\pPr[X_v = s \mid X_u = 0]}\right] \right] \\
&= \E_{v\sim \pi} \left[\sum_{s \in \Sigma} \pPr[X_v = s | X_u = 0] \cdot \E_{w \sim v}\left[\pPr[X_v = s, X_w = \pi_{vw}(s) \mid X_u = 0]\right] \right] \\
&\leq \E_{v \sim \pi} \E_{w \sim v} \left[\sum_{s \in \Sigma} \pPr[X_v = s | X_u = 0] \cdot \pPr[X_w = \pi_{vw}(s) \mid X_u = 0] \right] \\
&= \E_{Y^u}[\val(Y^u)]
\end{align*}
Further note that the expected value of rounding of Algorithm~\ref{alg:low-ent} is $\E_{u \sim \pi}[\E_{Y^u}[\val(Y^u)]]$ which is greater than $\Psi(\mu)$ by the above inequality. Since $\Psi_{u}(\mu) \geq \delta$, we sample a solution with expected value at least $\delta$.
\end{proof}

Now, we will relate the two potentials.

\restatelemma{lem:relating-ent}
\begin{proof}
We begin by recalling that in the definition of $\Phi_{\beta,\eta}$, we used an $\eta$-additive polynomial approximation $p(x)$ of degree $\tilde{O}(1/\eta)$ to the indicator function $\Ind[x \ge \beta]$ on the interval $x \in [0,1]$, guaranteed by Theorem~\ref{thm:step-approx}.

We begin by expanding the definition of $\Phi_{\beta,\nu}(\mu)$:
\begin{align*}
    \Phi_{\beta,\nu}(\mu)
    &= \pE\left[\sum_{s \in \Sigma} \left(\E_{u \sim \pi} \Ind[X_u - X_u' = s]\cdot p(\val_u(X))\right)^2\right]\\
    &= \pE\left[\sum_{s \in \Sigma} \E_{u,v \sim \pi} \Ind[X_u - X_u' = X_v - X_v' = s]\cdot p(\val_v(X))\cdot p(\val_u(X))\right]\\
    &= \sum_{s \in \Sigma} \E_{u,v \sim \pi}\pE\left[ \Ind[X_u' - X_v' = s]\right] \cdot \pE\left[\Ind[X_u- X_v = s]\cdot p(\val_v(X))\cdot p(\val_u(X))\right],
\end{align*}
where in the last step we have replaced the condition on the difference of $X_u,X_u'$ with a condition on the difference of $X_u,X_v$ (and the same for $v$).
Now, we use that $0\le p(x)\le 1$ and $p(x) \le \frac{x}{\beta-\nu} + \nu$ for all $x \in [0,1]$, and furthermore this is \sos-certifiable (see Fact~\ref{fact:bdd-markov}).
Therefore, we can pull out a factor of $p$ and apply this inequality to the second one to obtain
\begin{align*}
    \Phi_{\beta,\nu}(\mu)
    &\le \left(\sum_{s \in \Sigma} \E_{u,v \sim \pi} \pE\left[\Ind[X_u - X_v = s]\right]\cdot \pE\left[\Ind[X_u- X_v = s]\cdot\frac{\val_u(X)}{\beta - \nu}\right]\right) + \nu\\
    &= \left(\frac{1}{\beta-\nu}\E_{u,v\sim\pi} \sum_{s \in \Sigma} \pE[\Ind[X_u - X_v = s]]^2 \cdot \pE[\val_u(X) \mid X_u - X_v = s]\right) + \nu\\
    &= \frac{1}{\beta-\nu}\Psi(\mu) + \nu,
\end{align*}
where we have applied the definition of conditional pseudoexpectation.
This completes the proof of the lemma.
\end{proof}

\section{Certifiable Small-Set Expanders}\label{sec:cert-sse}

In this section, we give an algorithm for unique games on certifiable small set expander graphs, when the certificate is via 2-to-4 hypercontractivity.
To state our theorem, we will require the following definition:

\begin{definition}{(Certifiable 2 to 4 hypercontractivity)}\label{def:hyp}
For $C \in \R_+$, $\lambda \in (0, 2)$, and $D \ge 2 $ an integer, 
a graph $G = (V,E)$ is said to be {\em $(\lambda,C,D)$-certifiably 2 to 4 hypercontractive } if for any $f : V \to \R$,
\[
 \vdash_{D} \quad \|\Pi_{\lambda} f\|_{\pi,4}^4 \le C \cdot \|f\|_{\pi,2}^4,
\]
where $\|f\|_{\pi,p} = \left(\E_{v \sim \pi} f(v)^p\right)^{1/p}$, and $\Pi_{\lambda}$ is the projection to the right eigenspace of eigenvalues at most $\lambda$ of $G$'s normalized Laplacian.
\end{definition}
We will also say that a graph is a $(\eps,\delta,D)$-certifiable SSE if there is a degree-$D$ \sos proof that sets of size $\le \delta$ have expansion at least $\eps$.

Our main theorem is the following (more fleshed out version of Theorem~\ref{thm:intro-main}):

\begin{theorem}\label{thm:cert-sse}\label{thm:main-cert}
For any $(\lambda,C,D)$-certifiable 2 to 4 hypercontractive graph $G$ and for all $\eps < \frac{1}{100}\lambda^2$, given a degree-$(D + \tilde{O}(C/\eps\lambda^4))$ shift-symmetric pseudodistribution $\mu$ of value $\ge (1-\eps)$ for an affine Unique Games instance $I = (G,\Pi)$ on $G$, Algorithm~\ref{alg:low-ent} 
runs in time $\poly(|V(G)|)$ and
outputs an assignment with expected value at least $\frac{\eps\lambda^4}{64 C}$.
\end{theorem}
\begin{proof}
We start with the fact that a graph which is certifiably 2 to 4 hypercontractive is also a certifiable small-set expander.
This was shown in \cite{BarakBHKSZ12}, but we will state and use stronger guarantees about the form of the certificate which were implicit in their proof (we give a proof in Appendix~\ref{sec:expansion-red} for completeness).

\begin{lemma}[Lemma 6.7 in \cite{BarakBHKSZ12}]\torestate{\label{lem:sse}
If $G = (V,E)$ is $(\lambda,C,D)$-certifiably 2 to 4 hypercontractive, $G$ is a $(\lambda/2, \lambda^4/(16C),D)$-certifiable small-set expander: for any $f:V \to \R$,
\[
\left\{\|\Pi_{\lambda} f\|_{\pi,4}^4 \le C \cdot \|f\|_{\pi,2}^4\right\} \cup \left\{ f(v)^2 = f(v)\right\}_{v \in V} \cup \left\{\E_{\pi} f \le \frac{\lambda^4}{16C}\right\}
\,\, \vdash_{4+D} \quad
\langle f, L f\rangle_\pi \ge \frac{\lambda}{2} \E_\pi [f],
\]
Where $\Pi_{\lambda}$ is the projector to the right eigenspace of eigenvalue $\le \lambda$ in $G$'s normalized Laplacian.
Further, 
\[
\left\{\|\Pi_{\lambda} f\|_{\pi,4}^4 \le C \|f\|_{\pi,2}^4\right\} \cup \left\{ 0 \le f(v) \le 1\right\}_{v \in V}
 \vdash_{4+D} \,
\langle f, L f\rangle_{\pi} \ge \frac{\lambda}{2} \E_{\pi}[f] + c \left( \frac{\lambda^4}{16C} \E_\pi[f] - \E_\pi[f]^2\right) + B(f)
\]
For $c$ a positive constant and $B(f) =
2(\E_\pi[f^{\circ 2} - f]) + \langle f^{\circ 3} -f, \Pi_{\lambda}f \rangle_{\pi}.
$
}
\end{lemma}

Letting $\alpha := \frac{\lambda}{2}$ and $\gamma := \frac{\lambda^4}{16C}$, our assumptions together with Lemma~\ref{lem:sse} give us a small-set expansion certificate of the following form:
\begin{equation}
SSE_{\alpha,\gamma}(G) \equiv \{0 \le f(u) \le 1\}_{u \in V} \simp{D+4} \langle f, L f\rangle_{\pi} \geq \alpha\E_{\pi}[f] + c_1 \cdot \left(\gamma \E_\pi[f] - \E_\pi[f]^2\right) 
+ B(f),\label{eq:cert}
\end{equation}
for $c_1$ a positive constant, $B(f) =
2(\E_\pi[f^{\circ 2} - f]) + \langle f^{\circ 3} -f, Pf \rangle_{\pi}$ and $P$ a projection operator.

Next, we will show that if a graph has such a certificate of small-set expansion, then one can also obtain a lower bound on the approximate shift potential $\Phi_{\beta,\nu}(X,X')$ (whose definition we now recall), which gives a condition under which we can round.
Theorem~\ref{thm:step-approx} guarantees the existence of a family $P_{\beta,\nu}$ of degree-$\tilde{O}(1/\nu)$ polynomials \sos-certifiably which approximate $\Ind[x \ge \beta]$ within an additive $\nu$ in the intervals $[0,\beta - \nu] \cup [\beta+\nu,1]$.
Fix $ p \in P_{\beta,\nu}$ to be one such polynomial.
The functions $\{f_s:V\to\R[X,X']\}_{s \in \Sigma}$ defined such that 
\begin{equation}
f_s(u) = \Ind(X_u - X'_u = s) \cdot p(\val(X_u))\label{eq:fs}
\end{equation}
give disjoint {\em approximate} vertex subsets of $G$ (approximate only because $p$ is not exactly an indicator). 
Recall the definition of the approximate shift-mass potential (Definition~\ref{def:sq-mass-apx}):
\[
\Phi_{\beta,\nu}(X,X') = \sum_{s \in \Sigma} \left(\E_u f_s(u)\right)^2
 = \sum_{s \in \Sigma} \left(\E_{u} \left(\one(X_u - X'_u = s) \cdot p(\val(X_u))\right)\right)^2.
\]
Edges crossing this partition must be unsatisfied in either $X$ or $X'$ (see the discussion in Section~\ref{sec:alg} and Fact~\ref{fact:ug-axioms}).
In a certifiable small-set expander with large objective value, this partition cannot cut too many edges, and therefore its pieces must be large. 
We will make this formal via the following lemma:

\begin{lemma}\torestate{\label{lem:sp-pseudo}
Let $I$ be a unique games instance over a graph $G=(V,E)$ in which functions $f:V \to [0,1]$ with support $\le \gamma$ are SoS-certifiably $\alpha$-expanding via the following certificate:
\[SSE_{\alpha,\gamma}(G):\equiv \{0 \le f(v) \le 1\}_{v \in V}\,\, \vdash_D\,\, \left\{\langle f, L f\rangle_{\pi} \ge \alpha \E_{\pi}[f] + c \cdot \left( \gamma \E_\pi[f] - \E_\pi[f]^2\right) 
+ B(f)\right\},\]
where $B(f) = 2(\E_\pi[f^{\circ 2} - f]) + \ip{f^{\circ 3} - f, P f}_\pi$, $c$ is a fixed positive constant and $P$ is a projection operator.

Then we have that, for all $\beta \in (0,1)$, $\nu \in (0,\frac{1}{3}(1-\beta))$, and $\eta \in \R^+$, there is an \sos lower bound on the approximate shift mass potential $\Phi_{\beta,\nu}$:
\begin{align*}
\A_I \cup \{p \in P_{\beta,\nu} \} \cup \SSE_{\alpha,\gamma}(G) &\vdash_{D+\tilde{O}(1/\nu)}
\Phi_{\beta,\nu}(X,X') \geq \gamma\left(1 - \frac{\viol(X)}{1 - \beta - \nu} - \nu\right) + K_{\beta,\nu}^{\alpha,\eta}(X,X'),
\end{align*}
where $\A_I$ are the axioms defined for $I$ by program (\ref{eq:ip}), $\viol(X) = 1 - \val(X)$ is the fraction of constraints $X$ violates, and $K_{\beta,\nu}^{\alpha,\eta}(X,X') = c'\cdot \left(\alpha - \left(4 + \alpha + \eta\right)\left(\frac{\viol(X)}{1 - \beta - \nu} + \nu \right) - \frac{1}{2\eta} - (\viol(X) + \viol(X'))  \right)$ for $c' \in \R^+$.}
\end{lemma}
We give the proof in Section~\ref{sec:shift-pot}.
Informally, the quantity $K_{\beta,\nu}^{\alpha,\eta}(X,X')$ can be made non-negative when the fraction of violations $\viol(X)$ and $\viol(X')$ are small relative to the expansion $\alpha$.

From equation (\ref{eq:cert}) and Lemma~\ref{lem:sp-pseudo}, we may choose $\beta = \eps \le .01$, $\nu = \eps \gamma$, and $\eta = \frac{1}{2\sqrt{\eps}}$, and the conditions of our theorem imply that we have a degree-$(D+ \tilde{O}(1/\eps\gamma))$ sum-of-squares proof that 
\begin{align}
 \Phi_{\eps,\eps\gamma}(X,X') 
&\ge \gamma\left(1 - \frac{\viol(X)}{1 - \eps - \eps\gamma} - \eps\gamma\right) + K_{\eps,\eps\gamma}^{\alpha,\sqrt{1/4\eps}}(X,X').\label{eq:non-neg}
\end{align}

In order to apply our rounding Theorem~\ref{thm:round}, we require that the pseudoexpectation $\pE[\Phi_{\eps,\eps\gamma}(X,X')]$ is large, where $\pE$ is the pseudoexpectation operator corresponding to the pseudodistribution $\mu$ given to us.
Since by assumption $\pE[\viol(X)] = \pE[\viol(X')] \le \eps$, $\pE$ has degree $(D+\tilde{O}(C/\eps\lambda^4)) = D + \tilde{O}(1/\eps\gamma)$ and $\pE$ satisfies $\cA_I$, we take the pseudoexpectation of (\ref{eq:non-neg}) to get
\begin{align}
\pE\left[\Phi_{\eps,\eps\gamma}(X,X')\right]
&\ge\gamma\left(1 - \frac{\eps}{1 - \eps - \eps\gamma} - \eps\gamma\right) + \pE\left[K_{\eps,\eps\gamma}^{\alpha,\sqrt{1/4\eps}}(X,X')\right].\label{eq:bd-p}
\end{align}

We show now that for our chosen parameters, $\pE[K_{\eps,\eps\gamma}^{\alpha,1/\sqrt{4\eps}}(X,X')] \ge 0$.
Expanding the expression for $K$ and using our bound on $\pE[\viol(X) + \viol(X')]$,
\begin{align*}
\pE\left[K_{\eps,\eps\gamma}^{\alpha,1/\sqrt{4\eps}}(X,X')\right]
&\ge c' \cdot \left(\alpha - \left(4 + \alpha + \frac{1}{2\sqrt{\eps}}\right) \left(\frac{\eps}{1-\eps-\eps\gamma} + \eps\gamma\right) - \sqrt{\eps} - 2\eps\right) \ge c' (\alpha - 5 \sqrt{\eps}),
\end{align*}
where to obtain the final inequality we have used that $\gamma < \frac{1}{2}$, $\eps < \frac{1}{25}$, and $\alpha < 1$.
Since $\alpha = \frac{\lambda}{2} \ge 5 \sqrt{\eps}$ by assumption and since $c' \in \R_+$, this quantity is non-negative.

Returning to (\ref{eq:bd-p}) and simplifying with our upper bounds $\eps < \frac{1}{25},\gamma < \frac{1}{2}$, we have that
\[
\pE\left[\Phi_{\eps,\eps\gamma}\right] \ge \frac{3}{4}\gamma.
\]
Applying Theorem~\ref{thm:round}, we conclude that conditioning and rounding a degree-$(D + \tilde{O}(C/\eps\lambda^4))$ pseudodistribution according to Algorithm~\ref{alg:low-ent} results in a solution of expected value $\ge (\frac{3}{4}\gamma - \eps \gamma)(\eps - \gamma \eps) \ge \frac{1}{4}\eps \gamma = \frac{\eps\lambda^4}{64C}$, as desired.
\end{proof}

\subsection{Bounding the shift potential in certifiable SSE graphs}\label{sec:shift-pot}

In this section, we will use that in a small-set expander, when the expansion of the approximate partition defined the $f_s$ is low and the objectives $\val(X),\val(X')$ are high, then the shift-partition potential $\Phi(X,X')$ (which is a proxy for the size of the partition parts) is large. 
Further, we will show that this fact has an SOS proof when the graph has an SOS certificate of expansion. 

\restatelemma{lem:sp-pseudo}

\begin{proof}[Proof of Lemma~\ref{lem:sp-pseudo}]
Given assignments $(X,X')$ consider the approximate partition defined by the $\{f_s\}_{s \in \Sigma}$ as in (\ref{eq:fs}) and identify $f_s$ with an approximate component $C_s$.
We note that the $f_s$ are close to indicator functions, as they are the product of an indicator and an approximate indicator $p$.
As noted after equation (\ref{eq:fs}), $\Phi_{\beta,\nu}(X,X') = \sum_s \E_{u \sim \pi}[f_s(u)]^2$.
Further, our axioms easily imply that $f_s$ are bounded functions,

\begin{claim}\torestate{\label{claim:bounded}
From our Unique Games axioms and the axiom that $p \in P_{\beta,\nu}$, we may conclude that the $f_s$ are bounded:
\[
\A_G \cup \{p \in P_{\beta,\nu}\} \vdash_{\tilde{O}(1/\nu)} \{0\le f_s(v) \le 1\}_{s \in \Sigma,v\in V}.
\]}
\end{claim}

We provide the proof below in Section~\ref{sec:claim-pfs}.
Thus, we may apply the SSE certificate $\SSE_{\alpha,\gamma}(G)$ guaranteed by the condition of the lemma to all the functions $f_s$ and sum up the equality over $s \in \Sigma$. 
This gives us,
\begin{equation*}
\sum_s \ip{f_s, L f_s}_{\pi} \ge \alpha \sum_s \E_\pi[f_s] + c \left(\gamma \sum_s \E_\pi[f_s] - \sum_s \E_\pi[f_s]^2 \right) - \left(\sum_s 2B_1(f_s) + \sum_s B_2(f_s)  \right).
\end{equation*}
For $B_1(f) = \E_\pi[f - f^{\circ 2}]$ and $B_2(f) = \ip{f - f^{\circ 3}, P f}_\pi$, and $c \ge 0$.
Substituting $\sum_{s} \E_\pi[ f_s]^2 = \Phi_{\beta,\nu}(X,X')$ and re-arranging the expression,
\begin{equation}\label{eq:partition}
\Phi_{\beta,\nu}(X,X') \ge \gamma \sum_s \E_\pi[f_s] + \frac{1}{c}\left(\alpha \sum_s \E_\pi[f_s]  - \left(\sum_s 2B_1(f_s) + \sum_s B_2(f_s)  \right)- \sum_s \ip{f_s, L f_s}_{\pi}\right) .
\end{equation}
We now bound and simplify the remaining terms.
Our goal will be to obtain as large as possible a quantity on the right-hand side.

First, we would like a lower bound on $\sum_{s \in \Sigma} \E_\pi[f_s]$, which measures the total number of vertices included in the approximate partition.
If we were working with the pure shift partition $\Ind[X_u = X_u' + s]$, then this quantity would be $1$; since we have dropped vertices of low objective value, we must prove that we did not remove too many.
\begin{claim}\torestate{\label{claim:cover}
Under the axioms guaranteed by our lemma, the total number of vertices participating in the approximate partition $\{f_s\}_{s \in \Sigma}$ is large,
$$\A_G \cup \{p \in P_{\beta,\nu}\} \,\, \vdash_{\tilde{O}(1/\nu)} \,\,
    \sum_s \E_\pi[f_s] \geq 1 - \frac{\viol(X)}{1 - \beta - \nu} - \nu.
$$}
\end{claim}
This claim follows easily from an averaging argument if we replace $p(x)$ with $\Ind[x \ge \beta]$, since this amounts to removing vertices with at least $ 1-\beta$ incident violated edges in $X$. 
Below, we will show that this claim still holds as an SoS inequality when we use the $\eta$-approximate indicator $p$. 
See Section~\ref{sec:claim-pfs}.

Second, we must argue that the total expansion of the approximate partition is not too large.
The following claim shows that the expansion is bounded by the total violations of $X$ and $X'$:
\begin{claim}\torestate{\label{claim:expans}
Under the axioms guaranteed by our lemma, the total expansion of the partition is bounded as a function of the total violations in $X$ and $X'$:
$$\A_G \cup \{p \in P_{\beta,\nu}\} \,\, \vdash_{\tilde{O}(1/\nu)} \,\,
    \sum_s \ip{f_s,Lf_s}_\pi \leq \viol(X)+\viol(X') + 2\left(\frac{\viol(X)}{1 - \beta - \nu}\right) + 2\nu.$$}
\end{claim}

The proof of this claim uses the fact that in any satisfying assignment for an edge $(u,v)$, $X_u = X_v + s$ for a fixed $s \in \Sigma$, and therefore an edge that crosses the shift partition must be violated in either $X$ or $X'$ since the endpoints differ by a different shift in each assignment.
To account for vertices dropped because their violations are $\ge 1-\beta$, we again use an averaging argument.
We will prove this formally below.

Finally, if the $f_s$ were $0/1$-valued functions, $B_1(f_s)$ and $B_2(f_s)$ would have value $0$.
Since $f_s$ are instead approximately $0/1$ valued, we must show that $B_1(f_s)$ and $B_2(f_s)$ are close to $0$:

\begin{claim}\torestate{\label{claim:B1}
Under the axioms of our lemma, the $B_1(f_s)$ are small,
$$\A_G \cup \{p \in P_{\beta,\nu}\} \,\, \vdash_{\tilde{O}(1/\nu)} \,\,
\sum_s \E_\pi[f_s - f_s^{\circ 2}] \leq \frac{\viol(X)}{1 - \beta - \nu} + \nu.$$}
\end{claim}
\begin{claim}\torestate{\label{claim:B2}
Under the axioms of our lemma, for any $\eta \in \R_+$ and $\nu < \frac{1}{3}(1-\beta)$, the $B_2(f_s)$ may be bounded by
$$\A_G \cup \{p \in P_{\beta,\nu}\} \,\, \vdash_{\tilde{O}(1/\nu)} \,\,
    \sum_s \ip{f_s - f_s^{\circ 3}, P f_s}_\pi \leq \frac{1}{2\eta} + \eta\left(\frac{\viol(X)}{1 - \beta - \nu} + \nu\right).$$}
\end{claim}

When we combine these claims with equation~(\ref{eq:partition}) and simplify, we have the desired inequality, where the parenthesized right-hand side term becomes $K_{\beta,\nu}^{\alpha,\eta}$.
We prove our claims below in Section~\ref{sec:claim-pfs}
\end{proof}

\subsection{Proofs of Claims}\label{sec:claim-pfs}

We now prove the outstanding claims.
We first record some consequences of our unique games axioms $\A_G$, which will be useful to us:
\begin{fact}\label{fact:ug-axioms}
The unique games constraints $\A_G$ imply the following bounds:
\begin{enumerate}
    \item The local values and violations of variables are in $[0,1]$: $\A_G \vdash_2 \{0 \le \val_u(X) \le 1\} \cup \{0 \le \viol(X) \le 1\}$
    \item The variables $\{\Ind[X_u - X_u' = s]\}_{s\in\Sigma,u \in V(G)}$ satisfy Booleanity and partition constraints,
    \[
    \A_G \vdash_4 \{\Ind[X_u-X_u' = s]^2 = \Ind[X_u - X_u' = s]\}_{u \in V(G), s \in \Sigma} \cup \{\sum_{s \in \Sigma}\Ind[X_u - X_u' = s] = 1\} 
    \]
    \item The partition crossing edges are bounded by the sum of violations:
    \[
    \A_G \vdash_8 \E_{(u,v) \sim E(G)} \Ind[X_u - X_v \neq X_u' - X_v'] \le \viol(X) + \viol(X').
    \]
\end{enumerate}
\end{fact}
See Fact~\ref{fact:z-vars} in the appendix for a proof (the guarantees are phrased in terms of the variables $Z_{u,s} = \Ind[X_u - X_u' = s]$.\footnote{The proof of the final claim follows from Fact~\ref{fact:z-vars} sub-claim ``Crossing edges violate an assignment'' and from noting that from the Booleanity and partition constraints, $\Ind[X_u - X_u' \neq X_v - X_v'] = \sum_{s \neq t} Z_{u,s}Z_{v,t} = \sum_{s\neq t} Z_{u,s}Z_{v,t}(Y_{(u,v)} + (1-Y_{u,v}))(Y'_{(u,v)} + (1-Y'_{u,v})) \le (1-Y_{u,v}) + (1-Y'_{u,v}$); the claim is required for the final inequality.})

Much of the work in these proofs will consist of arguing that the approximate indicator $p$ behaves like a true indicator.
We will appeal to the following facts, which are proven later in Section~\ref{sec:apx-ind}:

\begin{fact}[Approximate Markov Inequality]\label{fact:apx-markov}
Under the axioms of the lemma, $p \in P_{\beta,\nu}$ approximately obey Markov's inequality over $[0,1]$:
\[
\A_G \cup \{p \in P_{\beta,\nu}\} \cup \{0 \le x \le 1\} \,\,\vdash_{\tilde{O}(1/\nu)}\,\, p(x) \ge 1 - \frac{1-x}{1-\beta-\nu} - \nu.
\]
\end{fact}
See Fact~\ref{fact:bdd-markov} for a proof of a slightly more general statement.

\begin{fact}[Approximate Union Bound]\label{fact:apx-union}
The approximate events $p(x), p(y)$ satisfy the union bound:
\[
\{0 \le x,y \le 1\} \cup \{p \in P_{\beta,\nu}\} \,\, \vdash_{\tilde{O}(1/\nu)}  1 - p(x) p(y) \le (1-p(x))(1-p(y)).
\]
\end{fact}
See Fact~\ref{fact:apx-ub} for a proof of a slightly more general statement.

Now, we are ready to prove our claims.
\restateclaim{claim:bounded}
\begin{proof}[Proof of Claim~\ref{claim:bounded}]
By definition, $f_s(v) = \Ind[X_v - X_v' =s] \cdot p(\val_v(X))$.
From Fact~\ref{fact:ug-axioms} we have the axioms $\Ind[X_v -X_v' = s]$ and $0 \le \val_v(X) \le 1$ in degree-$4$, and from Theorem~\ref{thm:step-approx} we have the axiom that $0 \le p(x) \le 1$ in degree $\tilde{O}(1/\nu)$.
The conclusion follows as a consequence of these axioms, since for $0 \le A,B \le 1$,
$(1-A)B \ge 0$ and $(A-0)B \ge 0$.
\end{proof}

\restateclaim{claim:cover}
\begin{proof}[Proof of Claim~\ref{claim:cover}]
By the partition constraints (Fact~\ref{fact:ug-axioms}), for each $v \in V(G)$ 
\[
\sum_{s \in \Sigma} f_v(s) = p(\val_v(X)).
\]
From Fact~\ref{fact:apx-markov} we further have that
\[
p(\val_v(X)) 
\ge \left(1-\frac{\viol_v(X)}{1-\beta-\nu}-\nu\right),
\]
where $\viol_v(X) = 1 - \val_v(X)$, and the inequality is a sum-of-squares inequality of degree $\deg(p)+2$. 
Finally, we use that $\pi$ is the stationary measure to conclude that $\E_{v \sim \pi} \viol_v(X) = \viol(X)$, and the conclusion follows.
\end{proof}

\restateclaim{claim:expans}
\begin{proof}[Proof of Claim~\ref{claim:expans}]

We begin by expanding the left-hand side.
By definition of the Laplacian,
\begin{equation}
\sum_{s \in \Sigma} \langle f_s, L f_s \rangle_\pi
=  \sum_{s \in \Sigma} \E_{(u,v) \sim E(G)} \frac{1}{2} f_s(u)^2 + \frac{1}{2} f_s(v)^2 - f_s(u) f_s(v).\label{eq:lap-ex}
\end{equation}
We now apply the fact that the $\Ind(X_u - X_u' = s)$ satisfy Booleanity and partition axioms (Fact~\ref{fact:ug-axioms}) to obtain that $
\sum_{s} f_s(v)^2 = \sum_s \Ind(X_v - X_v'=s) p(\val_u(X))^2 = p(\val_u(X))^2 \le 1$, where the inequality is a sum-of-squares inequality, and also that $\sum_s f_s(u) f_s(v) = p(\val_v(X))p(\val_v(X))\Ind(X_u - X_v = X'_u - X'_v)$.
Combining these, we have the sum-of-squares inequality
\begin{equation}
(\ref{eq:lap-ex}) 
\le \E_{(u,v) \sim E(G)}\left( 1- p(\val_u(X))p(\val_v(X))\Ind[X_u - X_v = X_u' - X_v']\right)\label{eq:lap-2}
\end{equation}
Now we can add and subtract $\Ind[X_u - X_v = X_u' - X_v']$ to the right hand side and then apply the approximate union bound Fact~\ref{fact:apx-ub} to obtain
\begin{align}
(\ref{eq:lap-2}) 
&\le \E_{(u,v) \sim E(G)} \left(\Ind[X_u - X_v \neq X_u' - X_v'] + \Ind[X_u - X_v = X_u' - X_v'](1-p(\val_u(X))p(\val_v(X)))\right)\nonumber\\
&\le \E_{(u,v) \sim E(G)}\left( \Ind[X_u - X_v \neq X_u' - X_v'] + (1-p(\val_v(X))) + (1-p(\val_u(X)))\right),\label{eq:lap-3}
\end{align}
with both inequalities certifiable by $O(\deg(p))$ sum-of-squares proofs.
To bound the first term $\Ind[X_u - X_v \neq X_u' - X_v']$, we use the third claim of Fact~\ref{fact:ug-axioms}, and to bound the remaining terms we apply our approximate Markov's inequality Fact~\ref{fact:apx-markov}, concluding that
\[
(\ref{eq:lap-3}) \le \viol(X) + \viol(X') + \E_{(u,v) \sim E(G)} \left(\frac{\viol_u(X)}{1-\beta-\nu} + \frac{\viol_v(X)}{1-\beta-\nu} +2\nu\right),
\]
and finally applying the property of the stationary measure that $\E_{u \sim \pi} g(u) = \E_{(u,v) \sim E(G)} g(u)$, and that $\E_{u \sim \pi} \viol_u(X) = \viol(X)$, we obtain our conclusion.
\end{proof}

\restateclaim{claim:B1}
\begin{proof}[Proof of Claim~\ref{claim:B1}]
For any $v \in V(G)$, the Booleanity and partition constraints (Fact~\ref{fact:ug-axioms}) give us that
\begin{align*}
\sum_{s \in \Sigma} f_s(v) - f_s(v)^2
&= \sum_{s \in \Sigma} \Ind[X_v - X'_v = s]\left(p(\val_v(X)) - p(\val_v(X))^2\right)\\
&=p(\val_v(X)) - p(\val_v(X))^2
\le 1- p(\val_v(X)),
\end{align*}
where we note the final inequality is an \sos inequality by applying the axiom that $p(x) \in [0,1]$ for $x \in [0,1]$, and that $\val_v(X) \in [0,1]$.
Now applying our approximate Markov's inequality (Fact~\ref{fact:apx-markov}) and the fact that $\val_v(X) = 1-\viol_v(X)$, and finally noting that $\E_{v\sim \pi}\viol_v(X) = \viol(X)$ by definition of the stationary measure, we have our conclusion.
\end{proof}

\restateclaim{claim:B2}
\begin{proof}[Proof of Claim~\ref{claim:B2}]
We apply Cauchy-Schwarz,
\[
\sum_{s \in \Sigma} \langle f_s - f_s^{\circ 3}, Pf_s\rangle_\pi \le \frac{1}{2\eta} \sum_{s} \|Pf_s\|_{\pi,2}^2 + \frac{\eta}{2}\sum_{s} \|f_s - f_s^{\circ 3}\|_{\pi,2}^2.
\]
To bound the first term on the right-hand side, we note that $P$ is a projection matrix, and therefore we can bound the sum $\sum_{s} \|Pf_s\|_{\pi,2}^2 \le \sum_{s}\|f_s\|_{\pi,2}^2 \le 1$.
Further, the inequality is an \sos inequality since the $f_s$ satisfy approximate partition constraints and we can certify that $f_s(v)\in[0,1]$ (Fact~\ref{fact:ug-axioms} and Claim~\ref{claim:bounded}).
To bound the second term on the right-hand side, we expand,
\begin{align*}
    \sum_{s \in \Sigma} \|f_s - f_s^{\circ 3}\|_{\pi,2}^2
    &= \sum_{s \in \Sigma} \E_{u \sim \pi} (f_s(u) - f_s(u)^3)^2\\
    &= \sum_{s \in \Sigma} \E_{u \sim \pi} \Ind[X_u - X_u' = s] p(\val_u(X))^2 (1 - p(\val_u(X))^2)^2,\\
    &= \E_{u \sim \pi}p(\val_u(X))^2 (1 - p(\val_u(X))^2)^2,
\end{align*}
where we have used the Booleanity and partition constraints from Fact~\ref{fact:ug-axioms}.
The same fact ensures that we have as an \sos axiom that $\val_u(X) \in [0,1]$ and therefore $p(\val_u(X))^1(1-p(\val_u(X))^2) \le 1$, so we have as an \sos inequality of degree $O(\deg(p))$,
\begin{align*}
    \sum_{s \in \Sigma} \|f_s - f_s^{\circ 3}\|_{\pi,2}^2
    &\le 1 - \E_{u \sim \pi}p(\val_u(X))^2.
\end{align*}
Now applying Observation~\ref{obs:squaring}, we have that $p^2$ shares all of the characteristics of $P_{\beta,2\nu}$ save for the degree bound, and combining this with our approximate Markov inequality (Fact~\ref{fact:apx-ub}) we get that
\begin{align*}
    \sum_{s \in \Sigma} \|f_s - f_s^{\circ 3}\|_{\pi,2}^2
    &\le \E_{u \sim \pi}\left(\frac{\viol_u(X)}{1-\beta-2\nu} + 2\nu\right).
\end{align*}
The conclusion now follows by noting that $\E_{u \sim \pi}\viol_u(X) = \viol(X)$, and by using our bound $\nu < \frac{1}{3}(1-\beta)$ to argue that $\frac{1}{2}\cdot \frac{1}{1-\beta-2\nu} \le \frac{1}{1-\beta-\nu}$.
\end{proof}

\section{UG on Noisy-Hypercube and Short-code graphs}
\label{sec:hyper}

Here, we derive two corollaries of Theorem~\ref{thm:main-cert}:
we show that polynomial-time sum-of-squares relaxations solve Unique Games on the noisy hypercube graph and the short-code graph.
These results follow easily by combining our results with the prior results of Barak et al. \cite{BarakBHKSZ12}, who showed that these graphs are certifiably 2 to 4 hypercontractive in sum-of-squares degree 4.

We first treat the noisy hypercube:
\begin{definition}[Noisy Hypercube Graph]
For each $\eps \in [0,1]$ and $d \in \N_+$, the $\eps$-noisy $d$-dimensional hypercube is the graph on $\{\pm 1\}^d$, with weighted edges $\{w_{uv}\}_{u,v \in \{\pm 1\}^d}$ where $w_{u,v} = (\eps)^{(d - \langle u,v \rangle)/2} (1-\eps)^{(d + \langle u,v \rangle)/2}$.
\end{definition}

Motivated by breaking known Unique Games integrality gaps, the work of \cite{BarakBHKSZ12} showed that the classical proof of hypercontractivity for the noisy hypercube (see e.g. \cite{O'Donnell-ABF}) can be recast as a degree-4 sum-of-squares proof.
\begin{theorem}[Noisy-Hypercube Certificate (\cite{BarakBHKSZ12}, Lemma 5.1)]\label{thm:noisy-hyp}
Suppose $G$ is the $d$-dimensional $\alpha$-noisy hypercube.
Then for any $t \in [d]$, $G$ is $(1-(1-2\alpha)^t,9^t,4)$-certifiably 2 to 4 hypercontractive.
\end{theorem}

In the same work, Barak et al. \cite{BarakBHKSZ12}, building on \cite{BGHMRS15}, noted that the same argument shows that the short code graph is also SOS-certifiably 2 to 4 hypercontractive.
\begin{definition}[Short Code Graph]
For each $d < n \in \N_+$, the $(d,n)$-shortcode graph is a graph whose vertex set is the set of degree-$d$ polynomials over $\F_2^n$ and with edges between each pair of polynomials $p,q$ such that $p - q$ is a product of $d$ linearly independent affine forms.
For any $\alpha \in [0,1)$, the {\em $\alpha$-noisy $(d,n)$-shortcode graph} is the graph with the random walk transition matrix $(G_{d,n})^{1 + \alpha 2^d}$.
\end{definition}

\begin{remark}
The noisy version of the short code is qualitatively similar to the noisy hypercube, since the transition probabilities in the $\alpha$-noisy $n$-dimensional cube are similar to performing an $\alpha n$-step random walk on the hypercube graph.
In \cite{BGHMRS15}, a different notion of noise is used, where they instead consider the graph with adjacency matrix $\exp(-\alpha 2^d(I - G_{d,n}))$; our results can be reformulated for this notion of noise as well.
\end{remark}

\begin{theorem}[Short-Code Certificate~\cite{BarakBHKSZ12}]\label{thm:short-code}
Suppose $G$ is the $\alpha$-noisy $(d,n)$-shortcode graph, and let $\ell = \lfloor \eta \cdot 2^{d}\rfloor$ for $\eta$ a universal constant.
Then for any $t \in [\ell]$, $G$ is $(1 - (1- t2^{-d})^{1+\alpha 2^d},9^t,4)$-certifiably 2 to 4 hypercontractive.
\end{theorem}

Combining Theorem~\ref{thm:cert-sse} with these results, we show that Unique Games instances on the Noisy Hypercube and Short Code graphs are easy.

\begin{theorem}[UG on Noisy-Hypercube, re-statement of Corollary~\ref{cor:intro-hc}]
For every $\eps \in [0,\frac{1}{400})$, $\alpha \in (0,\frac{1}{4})$, and $d\in \N$ sufficiently large, there exists an algorithm $A$ with the following guarantee: if $I = (G,\Pi)$ is an instance of Unique Games on the $d$-dimensional $\alpha$-noisy hypercube $G$ with $\val(I) \ge 1 - \eps$, then in time $|V(G)|^{\poly(\tau,1/\eps)}$, $A(I)$ returns an $\Omega(\eps^3/\tau)$-satisfying assignment for $I$ for $\tau = \exp(O(\sqrt{\eps}/\alpha))$.
\end{theorem}
\begin{proof}
From Theorem~\ref{thm:noisy-hyp}, for any $t \in [d]$, $G$ is $(1-(1-2\alpha)^t,9^t,4)$-certifiably hypercontractive.
For convenience, denote $\lambda_t = (1-(1-2\alpha)^t)$.
We now wish to apply Theorem~\ref{thm:main-cert}, so we will verify that its conditions hold. 

First, suppose that $\alpha > 5 \sqrt{\eps}$. In this case, let $\mu$ be the pseudodistribution obtained by symmetrizing the pseudodistribution given by the degree-$\poly(1/\eps,1/\alpha)$ SoS relaxation. Choosing $t = 1$, Theorem~\ref{thm:main-cert} guarantees that Algorithm~\ref{alg:low-ent} when run on $\mu$ returns a solution of value $\Omega(\eps^3)$.

Otherwise, suppose that $\alpha < 5\sqrt{\eps}$.
Then, we choose $t = \left\lceil\frac{\log (1-10\sqrt{\eps})}{\log(1-2\alpha)}\right\rceil$ so that $\eps < \frac{1}{100} \lambda_t^2$, and since $\eps \le 1/400$ and from our condition that $\alpha < \frac{1}{4}$ we have that $t=  O(\frac{\sqrt{\eps}}{\alpha})$. 
In this case, let $\mu$ be the pseudodistribution obtained by symmetrizing the pseudodistribution given by the degree-$\poly(1/\eps,\exp(\sqrt{\eps}/\alpha))$ SoS relaxation. 
Theorem~\ref{thm:main-cert} now guarantees that Algorithm~\ref{alg:low-ent} when run on an SoS relaxation of degree-$\tilde{O}(\frac{9^t}{\eps\lambda_t^4}) = \poly(1/\eps,\exp(\sqrt{\eps}/\alpha))$ returns a solution of expected value $\Omega(\frac{\eps \lambda_t^4}{9^t})= \Omega(\eps^3 \cdot\exp(-O(\sqrt{\eps}/\alpha)))$, as desired.
Using standard derandomization techniques we get a deterministic algorithm that runs in polynomial time and obtains a solution with the same guarantees.
\end{proof}

\begin{theorem}[UG on Short-Code, re-statement of Corollary~\ref{cor:intro-sc}]
There exist $\eps_0\in \R_+$ such that for every $n \in \N$ sufficiently large and $d \in \N$ with $2d < n$, $\eps \in (0,\eps_0)$, and $\alpha \in (0,1)$, there is an algorithm $A$ with the following guarantee: if  $(G,\Pi)$ is an instance of Unique Games on the $alpha$-noisy $(d,n)$-shortcode graph with $\val(G,\Pi) \ge 1-\eps$, then in time $|V(G)|^{\poly(1/\eps,\tau)}$, $A(G,\Pi)$ returns a solution of value $\Omega(\eps^3/\tau)$ for $(G,\Pi)$ for $\tau =\min\left( \exp(O(\sqrt{\eps}/\alpha)), \exp(O(\sqrt{\eps}2^d))\right)$.
\end{theorem}
\begin{proof}
Define $\lambda_t = 1-(1-t2^{-d})^{1 + \alpha 2^d}$.
From Theorem~\ref{thm:short-code}, for any $t \le \eta 2^d$, $G$ is $(\lambda_t, 9^t,4)$-certifiably 2 to 4 hypercontractive. 
We now wish to apply Theorem~\ref{thm:main-cert}, so we will verify that its conditions hold.
Choosing $t =\left\lceil \frac{20\sqrt{\eps}}{\alpha + \frac{1}{2^d}}\right\rceil$, by requiring $\eps \le \eps_0 \le \frac{1}{1600}$ we have that 
\[
1-\lambda_t \le \left(1-\frac{20\sqrt{\eps}}{\alpha 2^d + 1}\right)^{1 + \alpha 2^d} \le 1 - 10\sqrt{\eps},
\]
where we have used that $(1-2\delta x)^{1/x} \le 1-\delta$ for any $x \in (0,1)$ and $\delta \in (0,\frac{1}{2})$.
Therefore $\eps < \frac{1}{100} \cdot \lambda_t^2$. 
Further $t\le \eta 2^d$ for $\eta$ the universal constant in Theorem~\ref{thm:short-code} by our upper bound $\eps < \eps_0 = \min(\frac{\eta^2}{400},\frac{1}{1600})$.
Let $\mu$ be the pseudodistribution obtained by symmetrizing the pseudodistribution given by the degree-$\poly(1/\eps,1/\lambda_t,9^t) = \poly(\exp(O(\sqrt{\eps}/\alpha)),1/\eps)$ SoS relaxation.
Now we may apply Theorem~\ref{thm:main-cert} to conclude that Algorithm~\ref{alg:low-ent} finds a solution of expected value $\Omega(\frac{\eps\lambda_t^{4}}{9^t}) = \Omega(\eps^3\exp(-O(\sqrt{\eps}/\alpha))$ when run on $\mu$. Using standard derandomization techniques we get a deterministic algorithm that runs in polynomial time and obtains a solution with the same guarantees.
\end{proof}

\section{Johnson graphs}\label{sec:johnson}
In this section we'll prove that Algorithm~\ref{alg:j} succeeds in producing an assignment with good value for unique games instances of sufficiently high value over the Johnson graph.

\begin{algorithm-thm}[Unique Games on the Johnson Graph]\label{alg:j} Takes as input an affine UG instance on a $(n,\ell,\alpha)$-Johnson graph $I = (J,\Pi)$ over labels $\Sigma$ with $\val(I) = 1-\eps$, returns a $\Omega_{\eps,\alpha,\ell}(1)$ satisfying assignment.
\begin{enumerate}
\item Fix $r =\left\lfloor\frac{32\eps}{\alpha}\right\rfloor$, 
$\delta(\eta) := \frac{\eta}{\exp(cr)\binom{\ell}{r}}$ for all $\eta \in [0,1]$ and $D = \widetilde{O}(\frac{1}{\delta(\eps)})$, for $c>0$ a universal constant. 
Fix $\cA_I$ to be the set of unique games axioms/integer program over the instance $I$ (Program~\ref{eq:ip}).

\item Solve the degree-$D$ \sos SDP relaxation for the integer program $\cA_I$ and symmetrize the pseudodistribution over additive shifts (as described in Lemma~\ref{lem:sym-basic}) to get $\mu_0$. Set $j = 1$.

\item While the SDP value $\val_{\mu_{j-1}}(I) \geq 1 - 2\eps$:
\begin{enumerate}
\item For any $r' \le r$, find an $r'$-restricted subcube $C_j$ (induced subgraph of $J$, defined formally in Definition~\ref{def:subcube}) with high Condition\&Round value\footnote{The quantity $\text{CR-val}_{\mu}(C)$ corresponds to the expected value obtained when Algorithm~\ref{alg:low-ent} is performed on the subgraph $C$ and is formally defined in Definition~\ref{def:cr-val}.}: $\text{CR-val}_{\mu}(C_j) \ge \delta(\eta_{j-1})$ for $\eta_{j-1} = 1 - \val_{\mu_{j-1}}(I)$. 
\item Let $S_j$ be a subgraph of $C_j$ induced by the set of vertices that have not been previously assigned by any partial assignment $f_{k}, k < j$. Perform derandomized Condition\&Round on $V(S_j)$ to get a partial assignment $f_j$\footnote{As noted earlier, derandomization produces an assignment that satisfies $\text{CR-val}_{\mu_{j-1}}(S_j)$-fraction of edges and can be performed in polynomial time using the method of conditional expectations.}.
\item Rerandomize the pseudodistribution $\mu_{j-1}$ on $S_j$ to get $\mu_j$: Make the marginal distribution over the assigned vertices uniform and independent of other vertices, that is, for all degree $\leq D$ monomials define
$\pE_{\mu_{j}}$ as follows, 
\[\pE_{\mu_{j}}[X_{h_1,a_1} \cdots X_{h_t,a_t}X_{u_1,b_1} \cdots X_{u_m,b_m}]
:= \frac{1}{|\Sigma|^t}\pE_{\mu_{j-1}}[X_{u_1,b_1}\cdots X_{u_m,b_m}],
\]
where $\{(h_1,a_1), \ldots, (h_t,a_t)\} \in (V(S_j) \times \Sigma)^{t}$ and $\{(u_1,b_1), \ldots, (u_m,b_m)\} \in (({[n] \choose \ell} \setminus V(S_j)) \times \Sigma)^{m}$.
\item Increment $j$.
\end{enumerate}
\item Output any assignment $f:V \to \Sigma$ that agrees with all partial assignments $f_j$ considered above.
\end{enumerate}
\end{algorithm-thm}

We will prove that this algorithm returns a solution with value independent of the alphabet size.
\begin{theorem}\label{thm:main-johnson}
For every $\eps \in [0,\frac{1}{2000})$, $\alpha \in \Q$ with $\alpha < \frac{1}{2}$, $\ell \in \N$ with $\alpha \ell \in \N$, and integers $k,n$  sufficiently large, Algorithm~\ref{alg:j}
has the following guarantee: if $I$ is an instance of affine Unique Games on the $(n,\ell,\alpha)$-Johnson graph  $J$ with alphabet size $|\Sigma| = k$ and $\val(I) = 1 - \eps$, then in time $|V(J)|^{\poly(\binom{\ell}{r},1/\eps)}$, $A(I)$ returns an $\Omega\left(\frac{\eps^3}{\exp(O(r))\binom{\ell}{r}^2}\right)$-satisfying assignment for $I$ for $r = O(\eps/\alpha)$.
\end{theorem}

The proof of Theorem~\ref{thm:main-johnson} will require some additional ideas beyond that of Theorem~\ref{thm:main-cert}, as the Johnson graph is not a small-set expander. 
Nevertheless, we can characterize the structure of all the non-expanding sets, that is, we can prove that any non-expanding set must be large inside some canonical subgraphs. Using this characterization we prove that the above algorithm succeeds in finding a good assignment. The proof of our main theorem will proceed in the following steps:
\begin{enumerate}
\item We first prove a structure theorem (Theorem~\ref{thm:structure-johnson}) for non-expanding sets of the Johnson graph, similar to the theorem in~\cite{KMMS}. We show an \sos proof of the fact that every non-expanding set must be large when restricted to subcubes of the Johnson graph (Definition~\ref{def:subcube}).
\item Using the structure theorem, in Lemma~\ref{lem:sp-j} we first lower bound the global shift-partition potential $\Phi_{\beta,\nu}(X,X')|_C$ as a function of the violations of $X$ and $X'$. Roughly the global shift-partition potential $\Phi_{\beta,\nu}(X,X')|_C$ corresponds to the shift-component squared sizes when restricted to the subcube $C$ (see Definition~\ref{def:res-shift}). This lemma follows the same outline as that of Lemma~\ref{lem:sp-pseudo} for certifiable small-set expanders.
\item In the next step (Lemma~\ref{lem:round-partial-johnson}), we show that
given a pseudodistribution $\mu$ with objective value $1-\eps$ for unique games over the Johnson graph, one can find a subcube $C$ that has high global shift-partition potential. We then relate the global shift-partition potential to the shift-partition potential on the subgraph induced by $C$, $\pE_{\mu|_C}[\Phi_{\beta,\nu}^C(X,X')]$, to show that this is also high. By our rounding theorem, Theorem~\ref{thm:round} we then conclude that the expected value of the Condition\&Round algorithm, when performed on $C$ must be high. This corresponds to Step 3(a) in Algorithm~\ref{alg:j}.
\item Lastly in Lemma~\ref{lem:subroutine} we show that given a subroutine that finds a subgraph with high Condition\&Round value, there is an algorithm that uses this subroutine and finds a high value assignment to the whole graph. This corresponds to the while loop in Algorithm~\ref{alg:j}. Combining this lemma with Lemma~\ref{lem:round-partial-johnson} (discussed above), we get our main theorem.
\end{enumerate}
We prove the theorem below, after establishing each of these components separately. First let us discuss the structure theorem for Johnson graphs and  define the notion of restrictions.

\begin{definition}[$r$-restricted subcubes of $J$]\label{def:subcube}
Given an $(n,\ell,\alpha)$-Johnson graph $J$ and a set $A \subseteq [n]$ with $|A| = r$ such that $0 \leq r \leq \ell-1$, we let $J|_{A}$ denote the vertex-induced subgraph of $J$ induced by vertices that contain the set $A$. We call such a subset an $r$-restricted subcube of $J$. Note that when $A = \emptyset$ and $r = 0$, $J|_A$ is defined as the whole graph $J$.
\end{definition}

\begin{definition}[Restrictions of Functions]
For the $(n,\ell,\alpha)$-Johnson graph $J$, given a function $F:V(J) \rightarrow \R$ and a set $A \subseteq [n]$ with $A = r$, such that $0 \leq r \leq \ell-1$, we define the restricted function $F|_A: \binom{[n] \setminus A}{\ell - r} \rightarrow \R$ as, 
$$F|_A(X) = F(A \cup X).$$
Further, let $\delta_A(F)$ denote the fractional size of the function restricted to the subcube $J|_A$, that is,
$$\delta_A(F) := \delta(F|_A) = \E\limits_{X \sim \binom{[n] \setminus A}{\ell - r}}[F|_A(X)].$$
When $A = \phi$ and $r = 0$, we have that $F|_A(X) = F(X)$ for all $X \in \binom{[n]}{\ell}$ and $\delta_A(F) = \delta(F) = \E_\pi[F]$.
\end{definition}

We prove that every set in $J$ that is not correlated with any $r$-restricted cube, has high expansion (as a function of $r$). 

\begin{theorem}[Structure theorem for Johnson graphs]\torestate{\label{thm:structure-johnson}
For all $\alpha \in \Q$ with $\alpha < \frac{1}{2}$, all integers $\ell \in \N$ and all large enough integers $n \gg \ell$, the following holds: Let $J$ be a $(n,\ell,\alpha)$-Johnson graph and $\pi$ be the uniform distribution over $V(J)$. For every integer $r$ such that $0 \leq r \leq \ell/2$ and every function $F$ that is not correlated with any $r$-restricted subcube, $F$ has high expansion (as a function of $r$):
\begin{align*}
&\{F(X) \in [0,1]\}_{X \in V(J)} \,\, \vdash_2\,\, \\ 
&\langle F, L F\rangle_{\pi} \ge (1 - (1 - \alpha)^{r+1})\left[\left(1 - O_\ell\left(\frac{1}{n}\right)\right)\E_{\pi}[F] - 8^r \binom{\ell}{r}\left(\sum_{j = 0}^r \E_{Y \in \binom{[n]}{j}}[\delta_{Y}(F)^2]\right) + B(F) \right].
\end{align*}
where $B(F)$ represents the Booleanity constraints and equals $\E_\pi[F^{\circ 2} - F]$.}
\end{theorem}

Let us compare this theorem with~\cite{KMMS} and for simplicity let $\eps < 1.9\alpha$. Roughly, the structure theorem in~\cite{KMMS} implies that for every non-expanding set $S$ with expansion $\eps$, there exists a $1$-restricted subcube $C$ such that the $S$ is large inside $C$: $|S \cap C|/|C| \geq \Omega(1)$. From this theorem, one can derive the fact that
in fact a $\delta(S)/\ell$-fraction of the $1$-restricted subcubes have this property, where $\delta(S)$ denotes the fractional size of $S$ (by applying their theorem iteratively). Further this implies that given a distribution $D$ over non-expanding sets, say of the same size $\delta$, there exists a $1$-restricted subcube $C$ such that, $\E_D[|S \cap C|/|C|] \geq \Omega(\delta/\ell)$.   

But the above line of reasoning is not amenable to a low degree sum-of-squares proof because although each iterative step requires only a constant degree SoS proof, to get the final statement we need to apply the theorem $\Omega(n)$ times and this takes degree $\Omega(n)$. Our final aim is to prove the distribution-version of the statement. Our structure theorem gets around this barrier and directly proves the fact, using a constant degree SoS proof, that given a non-expanding set $S$ with expansion $\leq 1.9\alpha$, many subcubes are such that $S$ is large inside them. That is, rearranging Theorem~\ref{thm:structure-johnson}, as a corollary we have an SoS proof (in the formal indicator variables of membership in $S$) that $\E_C[|S \cap C|/|C|] \geq \Omega(1/\ell)$. Given this, we can easily derive the implication for distributions by applying an expectation over $D$ to the latter expression and exchanging expectations. Since we give an SoS proof, the statement holds true for pseudodistributions over non-expanding sets $S$! Lemma~\ref{lem:round-partial-johnson} carries out precisely this kind of an argument, but in more generality.

The proof ideas of Theorem~\ref{thm:structure-johnson} are similar to those in~\cite{KMMS}, hence we defer the proof of this theorem to Appendix~\ref{sec:fourier}. We will now show that under this theorem we get an algorithm for UG on the Johnson graph $J$. We will first formally define the global shift-partition potential on a subgraph.

\begin{definition}[Global shift-potential restricted to Subgraphs]\label{def:res-shift}
Let $I = (G,\Pi)$ be an instance of affine unique games over alphabet $\Sigma$. For any $\nu,\beta \in (0,1)$ and subgraph $H$ of $G$, define the {\em approximate global shift-partition potential restricted to the subgraph $H$} to be the quantity:
\[
\Phi_{\beta,\nu}(X,X')|_H = \sum_{s \in \Sigma} \E_{u \in H} \left[Z_{u,s} \cdot p(\val_u(X))\right]^2,
\]
for $Z_{u,s} = \one(X_u - X'_u = s)$, $\val_u(X) = \E_{(u,v) \in E(G)}[\one(X \text{ satisfies } (u,v))]$, and $p(x)$ the degree-$\tilde{O}(1/\nu)$ polynomial in the family $P_{\beta,\nu}$, described in Theorem~\ref{thm:step-approx}.
\end{definition}

Note that the global shift-partition potential measures the size of the global partition inside $H$, i.e. the $\val_u(X)$ is a function of all the edges in $E(G)$ that are incident on $u$, not just the edges in $H$. We will now use the structure theorem for Johnson graphs to get a lower bound on the global shift-partition potential restricted to subcubes $C$, $\Phi(X,X')|_C$, when the violations of the assignments $X$ and $X'$ are small. The following lemma is analogous to Lemma~\ref{lem:sp-pseudo} for certifiable small-set expanders and is proved in the same way. The main difference is in the conclusion of the lemma: instead of getting a lower bound on the shift-partition potential of the whole graph, we get a lower bound on the global shift-partition potential restricted to subcubes.

\begin{lemma}\label{lem:sp-j}
For all $\alpha \in \Q$ and all $\ell,n \in \N$ with $\alpha\ell \in \N$ and $\ell \ll n$ sufficiently large, the following holds: 
If $I$ is an affine unique games instance over the $(n,\ell,\alpha)$-Johnson graph $J$, then for all $\beta,\nu \in (0,1)$ and for every integer $r \in [\ell/2]$, there is an \sos lower bound of the following form on the average of the approximate global shift-partition potential $\Phi_{\beta,\nu}$ over $r$-restricted subcubes of $J$:
\begin{align*}
&\A_I \cup \{p \in P_{\beta,\nu} \}  \vdash_{\tilde{O}(1/\nu)} \\
&\qquad \qquad \sum_{j = 0}^r \E_{Y \in \binom{[n]}{j}}[\Phi_{\beta,\nu}(X,X')|_{(J|_Y)}] \geq \frac{1}{8^r \binom{\ell}{r}}\left( 1 - \frac{2 \viol(X)}{1 - \beta - \nu} - 2\nu - o_n(1) - K_{\beta,\nu}(X,X')\right),
\end{align*}
where $\A_I$ are the axioms defined for $I$ by program (\ref{eq:ip}), $\viol(X) = 1 - \val(X)$ is the fraction of constraints $X$ violates, and $K_{\beta,\nu}(X,X') = \frac{1}{1 - (1-\alpha)^{r+1}}\left(\viol(X) + \viol(X') + \frac{2\viol(X)}{1 - \beta - \nu} + 2\nu\right)$.
\end{lemma}

\begin{proof}
This proof proceeds exactly as the proof of Lemma~\ref{lem:sp-pseudo} for certifiable small-set expanders. We define functions $F_s$ corresponding to the components in the shift-partition and apply the structure theorem (Theorem~\ref{thm:structure-johnson}) to them and sum up the inequality over $s \in \Sigma$. For $Y \subseteq [n]$, with $0 \leq |Y| \leq r$, we have that,
\[\Phi_{\beta,\nu}(X,X')|_{(J|_Y)} = \sum_{s \in \Sigma} \E_{u \in J|_Y}[F_s(u)]^2 = \sum_{s \in \Sigma} \delta_Y(F_s)^2.\]

We can now use the same claims from Section~\ref{sec:cert-sse} to bound the terms in the structure theorem to get the conclusion of the lemma. We omit the details of the proof since it is straightforward given the above equality and the proof of Lemma~\ref{lem:sp-pseudo}.
\end{proof}

Using the lemma above, we will now prove that given a pseudodistribution $\mu$ over a highly satisfying instance of unique games over the Johnson graph we can find an $r$-restricted subcube $C$ with high Condition\&Round value. Let us define this precisely:

\begin{definition}[Condition\&Round Value]\label{def:cr-val}
Given a unique games instance $I = (G,\Pi)$ and a degree $4$ shift-symmetric pseudodistribution $\mu$ over $I$, for every subgraph $H$ of $G$, let $\text{ind-val}_\mu(H)$ denote the expected fraction of satisfied edges when independent rounding is performed on $V(H)$ using the marginals of $\mu$, i.e. $\text{ind-val}_\mu(H) := \E_{(v,w) \sim E(H)}[\sum_s \pE_\mu[X_{v,s}]\pE_\mu[X_{w,\pi_{vw}(s)}]]$. Let the Condition\&Round value, denoted by $\text{CR-val}_\mu(H)$ be the value obtained by performing Algorithm~\ref{alg:low-ent} on $H$, i.e. $\text{CR-val}_{\mu}(H) := \E_{u \sim V(H)}[\text{ind-val}_{\mu \mid X_u = 0}(H)]$.
\end{definition}

We will show this by first finding a cube $C$ that has high global shift potential, $\pE_{\mu}[\Phi(X,X')|_C]$, using Lemma~\ref{lem:sp-j} above. We then relate the global shift potential to the shift-partition potential on $C$, which we will denote by $\Phi^C(X,X')$. The only difference between the two potentials is that the latter is measured using the value of a vertex \emph{inside} $C$ and is the usual definition of the shift-partition potential on the graph $C$. We show that the subcube $C$ has small expansion, hence we can relate the global value of a vertex (when averaged over all edges in $E(G)$ incident on it) to the local value of a vertex (when averaged over just the edges in $E(C)$ incident on it), thus relating the global shift-partition potential to the shift-partition potential on $C$.
In particular, we will show that there exists   $C$ that has high shift-potential; using the analysis of the Condition\&Round algorithm, Theorem~\ref{thm:round}, this immediately gives us that there exists an $r$-restricted subcube that has high Condition\&Round value. To find such a cube $C$ algorithmically, one can just  enumerate over all $r$-restricted subcubes in time $n^r$ and check in polynomial time whether $C$ has high Condition\&Round value or not. Let us now make this argument formal.

\begin{lemma}\label{lem:round-partial-johnson}
For all $\eps \in [0,0.001)$, for all $\alpha \in \Q$ and $\alpha < \frac{1}{2}$, all integers $\ell\in \N$ with $\alpha \ell \in \N$ and all integers $k,n \gg \ell$ sufficiently large, the following holds: Let $I$ be an affine unique games instance over the $(n,\ell,\alpha)$-Johnson graph $J$ with alphabet size $|\Sigma| = k$ and $\val(I) = 1 - \eps$. Then for $r = \left\lfloor\frac{32\eps}{\alpha}\right\rfloor$, given a degree-$\tilde{O}\left(\frac{1}{\eps}2^{4r}\binom{\ell}{r}\right)$ shift-symmetric pseudodistribution $\mu$ satisfying the axioms $\A_I$, in time $n^{r}$ we can find a $s$-restricted subcube $C$ with $s\leq r$ such that $C$ has high Condition\&Round value: $\text{CR-val}_{\mu}(C) \geq \Omega\left(\frac{\eps}{2^{4r}\binom{\ell}{r}}\right)$.
\end{lemma}

\begin{proof}
Fix the parameters $\beta = 201\eps$, 
$r = \left\lfloor\frac{32\eps}{\alpha}\right\rfloor$, 
$\gamma = \frac{1}{16^{r+1}\binom{\ell}{r}}$ and 
$\nu = \eps \gamma$.
Since $\eps < 1/1000$ and $\alpha\ell \ge 1$, we have that $r \leq \ell/4$.
So we can now apply Lemma~\ref{lem:sp-j}, with the parameters $\beta,\nu$ and $r$. 
The conditions of our theorem imply that we have a degree-$\tilde{O}(1/\gamma)$ sum-of-squares proof that 
\begin{align}\label{eq:non-neg-j}
\sum_{j = 0}^r \E_{Y \in \binom{[n]}{j}}[\Phi_{\beta,\nu}(X,X')|_{(J|_Y)}] \geq \frac{1}{8^r \binom{\ell}{r}}\left( 1 - \frac{2 \viol(X)}{1 - 201\eps - \eps\gamma} - 2\eps\gamma - o_n(1) - K_{201\eps,\eps\gamma}(X,X')\right),
\end{align}
where $K_{201\eps,\nu}(X,X') = \frac{1}{1 - (1-\alpha)^{r+1}}\left(\viol(X) + \viol(X') + \frac{2\viol(X)}{1 - 201\eps - \eps\gamma} + 2\eps\gamma\right)$.

In order to apply our rounding Theorem~\ref{thm:round}, we require that the pseudoexpectation of the shift-partition potential on $C$, denoted by  $\pE_\mu[\Phi^C_{\eps,\eps\gamma}(X,X')]$ is large, for some $s$-restricted subcube $C = J|_Y$ with $s \le r$. The shift-partition potential on $C$ is just applying $\Phi$ to the graph induced by $C$, whereas the global shift potential restricted to $C$ measures the component sizes of the global shift partition within $C$. Formally,

\[
\Phi^C_{\beta,\nu}(X,X') = \sum_{s \in \Sigma} \E_{u \in C} \left[Z_{u,s} \cdot p(\val^{C}_u(X))\right]^2,
\]
where $\val^C_u(X)$ is the value of $u$ averaged over edges incident on $u$ in $C$ (as opposed to edges in $G$). We will first argue that there is a subcube whose global restricted shift potential is large, and then relate the two. 

Since $\val(I) = 1-\eps$, it follows that $\pE[\viol(X)] = \pE[\viol(X')] \le \eps$ and $\pE$ satisfies $\cA_I$, we take the pseudoexpectation of (\ref{eq:non-neg-j}) to get
\begin{align}\label{eq:bd-p-j}
\sum_{j = 0}^r \E_{Y \in \binom{[n]}{j}}\pE[\Phi_{\beta,\eps\gamma}(X,X')|_{(J|_Y)}] \geq \frac{1}{8^r \binom{\ell}{r}}\left( 1 - \frac{2 \eps}{1 - 201\eps - \eps\gamma} - 2\eps\gamma - o_n(1) - \pE[K_{201\eps,\eps\gamma}(X,X')]\right).
\end{align}

We show now that for our chosen parameters, $\pE[K_{201\eps,\eps\gamma}(X,X')] \leq \frac{1}{2}$. Expanding the expression for $K$ and using our bound on $\pE[\viol(X) + \viol(X')]$,
\begin{align}\label{eq:jlb}
\pE\left[K_{201\eps,\eps\gamma}(X,X')\right]
&\le \frac{1}{1 - (1-\alpha)^{r+1}}\left(2\eps + \frac{2\eps}{1 - 201\eps - \eps\gamma} + 2\eps\gamma\right) \leq \frac{8\eps}{1 - (1-\alpha)^{r+1}}
\end{align}
where to obtain the final inequality we have used that $\eps\gamma < \eps < \frac{1}{4}$ and $1-201\eps \ge \frac{1}{2}$. 
By our choice of parameters, $(1-\alpha)^{r+1} < 1 - 16\eps$, and rearranging gives us that $\pE[K_{201\eps,\eps\gamma}(X,X')] \leq \frac{1}{2}$.

Thus, returning to (\ref{eq:bd-p-j}) and simplifying with our upper bounds $\eps < \frac{1}{48}$ and $o_n(1) < 1/4$, we have that
\[
\sum_{j = 0}^r \E_{Y \in \binom{[n]}{j}}\pE[\Phi_{201\eps,\eps\gamma}(X,X')|_{(J|_Y)}] \ge \frac{1}{8^{r+1}\binom{\ell}{r}}.
\]
We can now apply an averaging argument to conclude that there exists a $\leq r$-restricted subcube $J|_Y$ such that,
\[\pE[\Phi_{201\eps,\eps\gamma}(X,X')|_{(J|_Y)}] \ge \frac{1}{r8^{r+1}\binom{\ell}{r}} \geq \frac{1}{16^{r+1}\binom{\ell}{r}} = \gamma. \]

Finally, we will relate the global restricted potential to the potential on $C$. We have the following claims.
The first states that an $r$-restricted subcube has bounded expansion when $r$ is not too large.

\begin{claim}
\torestate{\label{claim:subcube-expanse}
If $r = \left\lfloor\frac{32\eps}{\alpha}\right\rfloor < \frac{\ell}{4}$ and $s < r$, an $s$-restricted subcube of $J_{n,\ell,\alpha}$ has expansion at most $200\eps$.}
\end{claim}

The proof of this claim is via a direct calculation, and we give it in Section~\ref{sec:claims-j} below.
From this claim, we are able to prove that the local and global restricted potentials are related:

\begin{claim}
\torestate{\label{claim:potentials}
Suppose that $C$ is an $r$-restricted subcube of $J_{n,\ell,\alpha}$ with $r = \left\lfloor\frac{32\eps}{\alpha}\right\rfloor$.
Then if $\Phi^C$ is the shift-partition potential restricted to $C$, for any $\beta \ge 201\eps$ and $\nu < \eps$,
\[
\Phi^C_{\beta -200\eps,\nu}(X,X') \ge \Phi_{\beta,\nu}(X,X')|_C - 2\nu,
\]
and furthermore this is certifiable in degree $\tilde{O}(1/\nu)$ \sos.}
\end{claim}
The proof of this claim is based on the fact that the fraction of neighbors of every vertex $v \in C$ which lie outside of $C$ cannot be too large when $r$ is bounded, and therefore if the value in $J$ at a vertex is $\beta$, the value restricted to $C$ is still $\beta - \phi(C)$.
We give the proof in Section~\ref{sec:claims-j} below.

From Claim~\ref{claim:potentials} and (\ref{eq:jlb}) we have that there exists a subcube $C = J|_Y$ such that the local potential on $C$ is large,
\[
\pE[\Phi^C_{\eps,\eps\gamma}(X,X')] \ge \gamma -2\nu = \gamma(1-2\eps).
\]

We can now apply Theorem~\ref{thm:round} to get that the condition and round algorithm when applied to the vertices in $C$, would produce a good satisfying assignment for $C$ in expectation, i.e.  $\text{CR-val}(C)$ is high. Concretely we get that conditioning and rounding a degree-$\tilde{O}(1/\gamma\eps) = \tilde{O}\left(\frac{16^{r+1}\binom{\ell}{r}}{\eps}\right)$ pseudodistribution on the subcube $C = J|_Y$ according to Algorithm~\ref{alg:low-ent} results in a solution of expected value $\ge (\gamma(1-2\eps) - \eps \gamma)(\eps - \gamma \eps) \ge \frac{1}{4}\eps \gamma = \frac{\eps}{4 \cdot 16^{r+1}\binom{\ell}{r}}$ within $C$.
\end{proof}

Using the above theorem, we can find a subcube $C$ with high value, say $\geq \delta = \Omega_{\ell,\alpha}(1)$, and then perform derandomized Condition\&Round algorithm to get a $\delta$-satisfying assignment to the vertices of $C$. But this may be a negligible fraction of edges of the whole graph (since even a $1$-restricted subcube is a $o(1)$-fraction of $J$), and we need to satisfy $\Omega_{\eps,\alpha,\ell}(1)$ constraints.
To achieve this, after setting the vertices of the subcube $C$, we alter the pseudodistribution $\mu$ and apply our algorithm iteratively: we randomize $\mu$ on $V(C)$, so that these vertices are completely uncorrelated with any other vertex. This ensures that the value of any edge incident on $V(C)$ is $1/|\Sigma|$, which is much smaller than $\delta$, under the modified pseudodistribution $\mu'$. Then, we run the algorithm again on $\mu'$ to find a subcube $C'$ with high Condition\&Round value. Since edges that are incident on previously assigned vertices have very low value, we can show that the subcube $C'$ has low intersection with $C$. Furthermore the subcubes we find have low expansion, so we get that, the derandomized Condition\&Round algorithm when performed on $C' \setminus C$ satisfies a constant fraction of edges incident on $C'$. We continue in this way until the modified pseudodistribution's value drops by $\Omega(\eps)$. We show that at each iteration of the while loop, by modifying the pseudodistribution we lower the value by an amount that is proportional to the fraction of edges we satisfy in that step. Thus, after sufficiently many iterations we lower the value of the pseudodistribution by $\Omega(\eps)$ and hence satisfy an $\Omega_{\ell,\alpha}(\epsilon)$ fraction of the edges in the graph. We make this argument formal below.

\begin{lemma}\label{lem:subroutine}
Let $\eps_0 \in (0,1)$ be a universal constant and $\delta: [0,1] \rightarrow [0,1]$ be a function. Let $\eps < \eps_0/2$ be any constant, and let $\delta_{\min} = \min_{\eta \in [\eps,2\eps]}(\delta(\eta))$. 
Let $G$ be a regular graph and $I$ be any unique games instance on $G$ with alphabet size $|\Sigma| = k \geq \Omega(\frac{1}{\delta_{\min}})$ and value $1 - \eps$. 

Suppose we have a subroutine $\cA$ which given as input $\mu$, a shift-symmetric degree-$D$ pseudodistribution satisfying $\cA_I$ with $\val_\mu(I) = 1 - \eta \ge 1-\eps_0$, returns a vertex-induced subgraph $H$ such that, 1) $\text{CR-val}_{\mu}(H) \geq \delta = \delta(\eta)$ and 2) the edge-expansion of $H$ is $O(\eta)$. 

Then if $\cA$ runs in time $T(\cA)$, there is a $|V(G)|(T(\cA) + |V(G)|^{O(D)})$-time algorithm which finds a solution for $I$ that satisfies an $\Omega(\delta_{\min}^2\eps)$-fraction of the  edges of $G$. 
\end{lemma}

\begin{proof}
We will use the algorithm $\cA$ as a subroutine. To get a full assignment, our algorithm below is a generalized version of the Algorithm~\ref{alg:j}, where we've replaced the steps 2 to 4 in Algorithm~\ref{alg:j} with an arbitrary subroutine $\cA$ that finds a subgraph with high Condition\&Round value with respect to $I$. We include it here for completeness.

\begin{algorithm-thm}[Partial to Full Assignment]\label{alg:partial} ~\\ 
\begin{compactenum}
\item Solve the degree-$D$ SoS SDP relaxation for the integer program $\cA_I$ and make the pseudodistribution shift-symmetric to get a pseudodistribution $\mu_0$. Set $j = 1$.

\item While $\eta_{j-1} := 1 - \val_{\mu_{j-1}}(I) \leq 2\eps$:
\begin{compactenum}
\item Run subroutine $\cA$ on $\mu_{j-1}$ to find a subgraph $H_{j}$ with $\text{CR-val}_{\mu_{j-1}}(H_j) \geq \delta(\eta_{j-1})$.
\item Let $S_j$ be a subgraph of $H_j$ induced by the set of vertices that have not been previously assigned by any partial assignment $f_{k}, k < j$. Perform derandomized Condition\&Round on $V(S_j)$ to get a partial assignment $f_j$.
\item Rerandomize the pseudodistribution $\mu_{j-1}$ on $S_j$ to get $\mu_j$: Make the marginal distribution over the assigned vertices uniform and independent of other vertices, that is, for all degree $\leq D$ monomials define
$\pE_{\mu_{j}}$ as follows, 
\[\pE_{\mu_{j}}[X_{h_1,a_1} \cdots X_{h_t,a_t}X_{u_1,b_1} \cdots X_{u_m,b_m}]
:= \frac{1}{|\Sigma|^t}\pE_{\mu_{j-1}}[X_{u_1,b_1}\cdots X_{u_m,b_m}],
\]
where $\{(h_1,a_1), \ldots, (h_t,a_t)\} \in (V(S_j) \times \Sigma)^{t}$ and $\{(u_1,b_1), \ldots, (u_m,b_m)\} \in (({[n] \choose \ell} \setminus V(S_j)) \times \Sigma)^{m}$.
\item Increment $j$.
\end{compactenum}
\item Output any assignment to $V(G)$ that agrees with all partial assignments $f_j$ considered above.
\end{compactenum}

\end{algorithm-thm}

Let us first check that the algorithm is well-defined. The initial pseudodistribution $\mu_0$ by definition satisfies axioms $\cA_I$ and is shift-symmetric. It has value $= 1 - \eps > 1 - \eps_0$. In subsequent iterations of the while loop all these properties are satisfied by the modified pseudodistributions: 1) the rerandomizing operation on pseudodistributions produces a valid pseudodistribution operator that satisfies the axioms $\cA_I$ and is also shift-symmetric, 2) At iteration $j$ of the while-loop, since the while condition is met, we know that $\mu_{j-1}$ has value $\geq 1 - 2\eps \geq 1 - \eps_0$ and furthermore we can show that since the value only decreases at each step, it is always $\leq 1 - \eps$, so that $\eta_j \in [\eps,2\eps]$. So inside the while-loop, $\cA$ will always find a non-empty subgraph $H_j$ with high Condition\&Round value. Next, we find an assignment $f_j$ to the set of vertices $V(S_j)$ that by definition don't intersect previously assigned vertices. Since $f_j$ doesn't reassign any vertices, in the final step of the algorithm it is possible to output an assignment that is consistent with all previously considered partial assignments. We will now show that our final partial assignment satisfies a large fraction of the edges, where we say that an edge $(u,v)$ is satisfied by a partial assignment $f_j$, if both vertices $u,v$ have been assigned labels under $f_j$ and the labels satisfy the edge. We claim the following two facts:

\begin{claim}
The drop in value in every iteration satisfies that:
\[\val_{\mu_{j-1}}(I) - \val_{\mu_j}(I) \leq \frac{2|V(H_j)|}{|V(G)|},\]

where $\val_\mu(I)$ denotes the SDP value of $I$ with respect to the pseudodistribution $\mu$. 
\end{claim}

\begin{proof}
For any edge $(h,v)$ where $h \in H_j$, we have that $\val_{\mu_j}((h,v)) = \frac{1}{k}$, whereas $\val_{\mu_{j-1}}((h,v)) \leq 1$. For any edge whose both endpoints lie outside $H_j$, the value remains unchanged under rerandomizing. Noting that the fraction of edges incident on vertices in $H_j$ is at most $\frac{2|V(H_j)|}{|V(G)|}$ the conclusion follows.
\end{proof}

\begin{claim}
The value of the partial assignment found at iteration $j$ satisfies:
\[\val(f_j) \geq \delta^2(\eta_{j-1})(1 - O(\eta_{j-1}))\frac{|V(H_j)|}{2|V(G)|},\]
where $\val(f_j)$ denotes the fraction of edges (in $E(G)$) satisfied by the partial assignment $f_j$.
\end{claim}

\begin{proof}
We will first prove that $\text{CR-val}_{\mu_{j-1}}(S_j) \geq \delta(\eta_{j-1}) := \delta$, where $S_j$ is the subgraph induced by the unassigned (by previous partial assignments $f_k$, $k < j$) vertices of $H_j$. For notational simplicity we will drop the subscript $j$ from $H_j,S_j$ and $\mu_{j-1}$. We know by the guarantees of the subroutine $\cA$ that $H$ is such that, $\text{CR-val}_{\mu}(H) = \E_{u \sim V(H)}[\text{ind-val}_{\mu | X_u = 0}(H)] \geq \delta$ (see Definition~\ref{def:cr-val} for $\text{CR-val}$ and $\text{ind-val}$).  First note that the marginals of every vertex are uniform, due to the shift-symmetry of $\mu$. Moreover we have that conditioning on previously assigned vertices, i.e. any vertex $u \in H \setminus S$, maintains this property, since the distribution of $u$ is completely uncorrelated with the other vertices. So we get that, $\text{ind-val}_{\mu | X_u = 0}(H) = \frac{1}{|\Sigma|} < \delta$ for all $u \in H \setminus S$. This implies that, 
\begin{equation}\label{eq:cr1}
\E_{u \sim S}[\text{ind-val}_{\mu | X_u = 0}(H)] \geq \text{CR-val}_{\mu}(H) \geq \delta.    
\end{equation}
Again we have that, $\text{ind-val}_{\mu | X_u = 0}(e) = \frac{1}{|\Sigma|}$, for any edge $e$ which has at least one endpoint in $H \setminus S$, so we get that, $\E_{u \sim S}[\text{ind-val}_{\mu | X_u = 0}(S)] \geq \delta$. Now we can perform derandomized Condition\&Round on $S$ to get an assignment $f_j$ that satisfies at least a $\delta$-fraction of the edges of $S$.

Next we will show, by an averaging argument, that the edges of $S$ constitute a large fraction of the edges incident on the vertices of $H$, which would imply that $f_j$ satisfies a large fraction of these edges. Let $u_0 \in S$ be a vertex for which $\text{ind-val}_{\mu | X_{u_0} = 0}(H) \geq \delta$ (we know such a vertex exists by equation~(\ref{eq:cr1})) and let $\mu'$ be the pseudodistribution $(\mu | X_{u_0} = 0)$. First note that the set of edges $E(H) \setminus E(S)$, have independent rounding value $1/|\Sigma|$ under $\mu'$, since at least one endpoint of such edges has been assigned previously. Since $\text{ind-val}_{\mu'}(H) \geq \delta$, a simple averaging argument gives us that the set $E(H) \setminus E(S)$ can be at most a $\frac{1 - \delta}{1 - (1/|\Sigma|)}$-fraction of $E(H)$. So the set $E(S)$ is at least a $\frac{\delta - (1/|\Sigma|)}{1 - (1/|\Sigma|)} \geq \delta/2$-fraction of $E(H)$. Since the expansion of $H$ is at most $O(\eta_{j-1})$, we have that $E(H)$ is a $(1 - O(\eta_{j-1}))$-fraction of the total edges incident on $H$, which in turn is at least a $\frac{|V(H)|}{|V(G)|}$-fraction of $E(G)$. Combining these facts we get that $f_j$ satisfies a $\delta \cdot \frac{\delta}{2} \cdot (1 - O(\eta_{j-1}))\cdot \frac{|V(H)|}{|V(G)|}$-fraction of the edges of $G$. 
\end{proof}

Once we have these facts, the conclusion is immediate. Firstly there cannot be more than $V(G)$ iterations of the while-loop, since  at each iteration we set the value of at least one new vertex to $1/|\Sigma|$. The rerandomization operation in the while loop as well as the symmetrization operation (Lemma~\ref{lem:sym}) can be done in time polynomial in the description of $\mu_0$. So each iteration takes time $T(\cA) + |V(G)|^{O(D)}$, hence the algorithm runs in time $|V(G)|(T(\cA) + |V(G)|^{O(D)})$. 

Moreover, combining the claims above, we get that the partial assignment at any iteration is proportional to the drop in value of the pseudodistribution. That is, 

\[\val(f_j) \geq \delta^2(\eta_{j-1})(1 - O(\eta_{j-1}))\frac{|V(H_j)|}{2|V(G)|} \geq \delta^2(\eta_{j-1})(1 - O(\eta_{j-1}))\left(\frac{\val_{\mu_{j-1}}(I) - \val_{\mu_j}(I)}{4}\right).\]

At the last iteration, we know that the pseudodistribution value has dropped by at least $\eps$ (compared to $\val_{\mu_0}(I)$), hence summing the above over all iterations $j$, we get that the value of the partial assignment returned by the algorithm is at least ${\frac{1}{4}(\min_j(\delta(\eta_j)))^2(1 - O(\eps))\eps}$ as required.
\end{proof}

This completes the analysis of Algorithm~\ref{alg:j}. Combining the lemmas above, Theorem~\ref{thm:main-johnson} easily follows.

\begin{proof}[Proof of Theorem~\ref{thm:main-johnson}]
Given a UG instance $(J,\Pi)$ on the Johnson graph and a shift-symmetric pseudodistribution $\mu$ of degree $D = \widetilde{O}(\frac{1}{\eta}2^{4r}\binom{\ell}{r})$ with value $1 - \eta$, for $\eta < 0.001$, Lemma~\ref{lem:round-partial-johnson} gives us a subgraph of $J$ with high Condition\&Round value. This subgraph has expansion $\leq 200 \eta$ (by Claim~\ref{claim:subcube-expanse}) and Condition\&Round value at least $\delta(\eta) = \Omega(\frac{\eta}{\exp(c'r)\binom{\ell}{r}})$, where $r = c\eta/\alpha$ for universal constants $c,c'$.
To bound $\delta_{\min}$, we take the derivative 
\[
\frac{\partial}{\partial\eta} \delta(\eta)
=\left(1 - \frac{(c' + \ln \ell)\eta}{\alpha}\right)\exp\left(-(c' +\ln \ell) c \frac{\eta}{\alpha}\right)
\]
and we can see that $\frac{\partial}{\partial\eta}\delta(\eta)$ has at most one sign change from positive to negative in the interval $[\eps,2\eps]$, which means that it is minimized at one of the endpoints $\delta(\eps)$ or $\delta(2\eps)$ which are both bounded below by $\Omega(\eps/\exp(O(r))\binom{\ell}{r})$. 
Furthermore, the subroutine for finding a subcube runs in time $|V(J)|^{O(r)}$.
Now observe that the algorithm stated in the proof of Lemma~\ref{lem:subroutine}, instantiated with the subroutine for finding an $r$-restricted subcube of the Johnson graph, is the same as Algorithm~\ref{alg:j}.
So we can apply the algorithm guarantees outlined in Lemma~\ref{lem:subroutine}, to complete the analysis for Algorithm~\ref{alg:j}.  
\end{proof}

\subsection{Proofs of outstanding claims}
\label{sec:claims-j}
Here we prove some of the claims that we have used in the proof of Theorem~\ref{thm:main-johnson} and supporting lemmas.

\restateclaim{claim:subcube-expanse}
\begin{proof}
Let $J|_Y$ be an $s$-restricted subcube.
We have that,
\[1 - \phi(J|_Y) = \frac{\binom{\ell - |Y|}{\alpha\ell}}{\binom{\ell}{\alpha\ell}} \geq \frac{\binom{\ell - r}{\alpha\ell}}{\binom{\ell}{\alpha\ell}} = \left(\frac{\ell - \alpha\ell}{\ell}\right)\left(\frac{\ell - \alpha\ell - 1}{\ell - 1}\right)\ldots \left(\frac{\ell - \alpha\ell - r + 1}{\ell - r + 1}\right).\]
Now since $r \leq \ell/4$ by assumption, each of the parenthesized terms is at least $\left(\frac{3\ell/4 - \alpha\ell}{3\ell/4}\right) = (1 - 4\alpha/3)$, so
\[1 - \phi(J|_Y) \geq \left(1-\frac{4\alpha}{3}\right)^r \geq 1 - \frac{4r\alpha}{3}.\]

Since $r = \left\lfloor\frac{32\eps}{\alpha}\right\rfloor < \frac{75\eps}{\alpha}$, we get that $\phi(J|_Y) < 200\eps$
as desired.
\end{proof}

\restateclaim{claim:potentials}
\begin{proof}
When $r = \left\lfloor\frac{32\eps}{\alpha}\right\rfloor$, the expansion of $C$ is at most $1 - (1-4\alpha/3)^r \le 200\eps$ by Claim~\ref{claim:subcube-expanse}.
Furthermore, from the definition of the Johnson graph this holds vertex-by-vertex; every $v \in C$ has at most a $200\eps$-fraction of its neighbors outgoing.
Therefore,
\[
\Ind[\val_u^C(X) \ge \beta - 200\eps] \ge \Ind[\val_u(X) \ge \beta],
\]
and furthermore since $\nu < \eps$,
\[
p_{\beta - 200 \eps,\nu}(\val_u^C(X)) + \nu \ge p_{\beta,\nu}(\val_u(X)) - \nu.
\]
Therefore, by definition,
\begin{align*}
    \Phi_{\beta - 200\eps,\nu}(X|_C,X'|_C)
    &= \sum_{s \in \Sigma} \E_{u \in C}\left[\left(Z_{u,s}\cdot p_{\beta-200\eps,\nu}(\val_u^C(X))\right)\right]^2\\
    &\ge \sum_{s \in \Sigma} \E_{u \in C}\left[\left(Z_{u,s}\cdot \left(p_{\beta,\nu}(\val_u(X))\right) - 2\nu\right)\right]^2\\
    &\ge \Phi_{\beta,\nu}(X,X')|_C - 2\nu,
\end{align*}
where each inequality is a sum-of-squares inequality of degree at most $2\deg(p)$.
\end{proof}
\section{Approximating indicator functions with low-degree polynomials}\label{sec:apx-ind}
In this section, we note that there is a low-degree polynomial which provides an SOS-certifiably good approximation to a step function.
This will be a consequence of the existence of low-degree approximations to step functions that appear in the literature, as well as the theory of univariate sums-of-squares.

The following theorem, due to Diakonikolas et al., provides a low-degree approximation to a step function.
Though similar statements may be proven using classical results in approximation theory, we use Diakonikolas et al.~\cite{diakonikolas2010bounded} as their degree bounds are sharper (though ultimately this does not qualitatively change our result).

\begin{theorem}[Corollary of Theorem 4.5 in \cite{diakonikolas2010bounded}]\label{thm:step-approx}
Define $s_{\alpha}(x)$ to be the step function at $\alpha \in (0,1)$, so that $s_{\alpha}(x) = 0$ if $x < \alpha$ and $1$ otherwise.
Then for each $0 < \delta < \alpha$ and $\epsilon > 0$ there is a univariate polynomial of $p_{\alpha}^{\epsilon,\delta}$ of degree $O(\frac{1}{\delta}\log^2\frac{1}{\epsilon})$ such that
\begin{enumerate}
\item $|p_{\alpha}^{\epsilon,\delta}(x) - s_{\alpha}(x)| \le \epsilon$ for all $x \in [0,\alpha-\delta] \cup [\alpha + \delta,1]$
\item $0 \le p_{\alpha}^{\epsilon,\delta}(x) \le 1$ for all $x \in [0,1]$
\item $p_{\alpha}^{\epsilon,\delta}$ is monotonically increasing on $(\alpha - \delta, \alpha + \delta)$.
\end{enumerate}
Further, given axioms $A = \{x \ge 0\}\cup \{x \le 1\}$, there is an SoS proof that 
\[ A \vdash_{O(\frac{1}{\delta}\log^2 \frac{1}{\epsilon})} \{0 \le p_{\alpha}^{\epsilon,\delta}(x) \le 1\}.
\]
\end{theorem}
\begin{remark}
Though the statement is not identical to that of Theorem 4.5 of \cite{diakonikolas2010bounded}, it is an easy corollary.
To switch from their $\text{sign}(y)$ polynomial for $y \in [-1,1]$ to $s_\alpha(x)$ for $x \in [0,1]$, we can do a simple change of variables, taking $y = x-\alpha$.
Shifting by a constant and rescaling changes the bounds so that $p(x) \in [0,1]$.
The third item is not explicitly written in the statement of Theorem 4.5 of \cite{diakonikolas2010bounded}, but it can be easily extracted from the proof.
The \sos-certifiability follows from Luk\'{a}cs' Theorem.
\end{remark}

We here recall Luk\'{a}cs' theorem and a simple corollary, which easily establish the \sos-certifiability of the step function approximation.
\begin{theorem}[Luk\'{a}cs Theorem, see e.g. \cite{szeg1939orthogonal}]
If $p$ is a degree-$d$ univariate polynomial with $p(x) \ge 0$ for $x \in [-1,1]$, then $p$ can be written as 
\[
p(x) = \begin{cases}
s(x)^2 + (1-x^2) t(x)^2 & \text{ if $d$ even}\\
(1+x)s(x)^2 + (1-x) t(x)^2 & \text{ if $d$ odd}
\end{cases}
\]
for $s,t$ real polynomials of degree at most $d$.
\end{theorem}

The following easy corollary is well-known (though we include the proof for completeness).
\begin{corollary}\label{cor:Lukacs}
Let $q$ be a degree-$d$ polynomial which is non-negative on $[a,b]$.
Then given the axioms $A = \{x \ge a\} \cup\{x \le b\}$, there is a degree-$2d$ \sos proof that $q$ is non-negative, $A \vdash_{2d} q(x) \ge 0$.
\end{corollary}
\begin{proof}
We claim that Luk\'{a}cs theorem implies
\[
q(x) = \begin{cases}
s(x)^2 + (x-a)(b-x)t(x)^2 & \text{ if $d$ even}\\
(x-a)s(x)^2 + (b-x)t(x)^2 & \text{ if $d$ odd}
\end{cases}
\] 
for $s,t$ real polynomials of degree at most $2d$, and this implies our corollary.
To get the claim, we perform a change of variables, taking $x' =\frac{2}{b-a} (x - a) -1$.
Let $q(x) = h(x')$. 
We now have that $h(x')$ is a degree-$d$ polynomial which is non-negative on $[-1,1]$.
From Luk\'{a}cs Theorem, we have that
\[
h(x') = \begin{cases}
u(x')^2 + (1-x'^2) v(x')^2 & \text{ if $d$ even}\\
(1+x')u(x')^2 + (1-x') v(x')^2 & \text{ if $d$ odd}
\end{cases}
\]
for $u,v$ real polynomials of degree at most $d$.
But now,
\[
\tfrac{b-a}{2} \cdot (1 - x') = b-x, \qquad \tfrac{b-a}{2}\cdot (1 + x') = x-a, \qquad \left(\tfrac{b-a}{2}\right)^2\cdot (1-x'^2) 
= (x-a)(b-x)
\]
and so by applying a change of variables to the polynomials $u(x'),v(x')$ to obtain $s(x),t(x)$, the conclusion follows.
\end{proof}

When we have \sos certificates that polynomials are bounded within $(0,1) \pm \epsilon$, \sos can also certify that they behave roughly like indicator functions.
\begin{fact}[Union bound for Approximate Indicators]\torestate{\label{fact:apx-ub}
Suppose that $h,g$ are polynomials of degree at most $d$, and suppose furthermore that from the axioms $A$, there is an \sos proof that $A \vdash_d \{0 \le g \le 1\}\cup\{ 0 \le h \le 1\}$.
Then,
\[
A \quad \vdash_{2d}\quad gh \ge g + h - 1
\]}
\end{fact}
\begin{proof}
We have as a polynomial equality that $(g + (1-g))(h + (1-h)) = 1$.
Expanding then re-arranging,
\begin{align*}
gh 
&= 1 - g(1-h) - h(1-g) - (1-g)(1-h)\\
&\succeq 1 - g(1-h) - h(1-g) - (1-g)(1-h) - (1-h)(1-g)\\
&= 1 - (1-h) - (1-g) 
\end{align*}
where in the second line we have used the \sos bounds $g,h \preceq 1$.
Simplifying gives the conclusion.
\end{proof}

\begin{fact}[Markov Inequality for Bounded Polynomials]
\label{fact:bdd-markov}
Let $p := p_{\alpha}^{\varepsilon,\delta}$ be the degree-$D = O(\frac{1}{\delta}\log^2\frac{1}{\varepsilon})$ polynomial guaranteed by Theorem \ref{thm:step-approx}.
Then $p$ satisfies Markov's inequality:
\begin{equation*}
\{0 \le x \le 1\}
\vdash_{\deg(p)} 
\{ p(x) \ge 1 - \frac{1-x}{1-\alpha-\delta} - \eps\} \cup \{p(x) \le \frac{x}{\alpha-\delta} + \eps\}
\end{equation*}
\end{fact}
\begin{proof}
We will perform case analysis on $x$, throughout using Corollary~\ref{cor:Lukacs} to obtain our \sos inequalities.
We prove the first inequality first.
For $x \in [0,\alpha+\delta)$,
\[
p(x) \succeq 0 \succeq 1 - \frac{1-x}{1-\alpha - \delta},
\]
where we have used that $p(x) \succeq 0$ and $\frac{1-x}{1-\alpha - \delta}\succeq 1$.
Now for $x \in [\alpha +\delta, 1]$,
\[
p(x) \succeq 1 - \eps \succeq 1 - \eps - \frac{1-x}{1-\alpha-\delta},
\]
where we have used that $x \in [0,1]$ so that we are subtracting a positive quantity.
Combining these claims concludes the proof of the first claim.

To see the second claim, notice that for $x \in [0,\alpha-\delta]$, $p(x) \le \eps$, and for $x \in (\alpha - \delta, 1]$, $p(x) \le 1 \le \frac{x}{\alpha - \delta}$.
This concludes the proof.
\end{proof}

\begin{observation}\label{obs:squaring}
Let $p := p_{\alpha}^{\varepsilon,\delta}$ be the degree-$D = O(\frac{1}{\delta}\log^2\frac{1}{\varepsilon})$ polynomial guaranteed by Theorem \ref{thm:step-approx}.
Then $p^2$ is a polynomial of degree $2D$ which enjoys the same guarantees as the polynomial $p_{\alpha}^{2\varepsilon,\delta}$.
\end{observation}
\begin{proof}
The polynomial $p^2$ is bounded in $[0,\eps^2]$ on $[0,\alpha -\delta]$, inherits the monotonic increasing property on $(\alpha-\delta,\alpha+\delta)$, and is bounded by $[(1-\eps)^2,1]$ on $[\alpha + \delta,1]$.
\end{proof}

\bibliographystyle{amsalpha}
\bibliography{references}

\appendix

\section{Sum-of-squares Background}\label{sec:prelims} \label{app:sos}

Given a polynomial optimization program $P = \{\max_{x} p(x) \ s.t. ~ q_i(x) = 0, \forall i \in [m] \}$, the degree-$D$ sum-of-squares semidefinite programming relaxation of $P$ is a semidefinite program of size $n^{O(D)}$ that returns a {\em pseudoexpectation operator} $\pE: x^{\le D} \to \R$. This operator can be uniquely extended to give a pseudo-expectation operator on the set of all polynomials of degree at most $D$ by linearity (defined precisely below).
This operator satisfies four properties:
\begin{itemize}
\item Scaling: $\pE[1] = 1$.
\item Linearity: $\pE[a \cdot f(x) + b \cdot g(x)] = a\cdot \pE[f(x)] + b \cdot \pE[g(x)]$, for all $a,b \in \R$ and all degree $\leq D$ polynomials $f,g$. 
\item Non-negativity of low-degree squares: $\pE[s(x)^2] \ge 0$ for all polynomials $s(x)$ with $\deg(s) \le \tfrac{D}{2}$.
\item Program constraints: $\pE[f(x) \cdot q_i(x)] = 0$ for all $i \in [m]$ and polynomials $f(x)$ such that $\deg(fq_i) \le D$.
\end{itemize}
Additionally, we will have $\pE[p(x)] \ge value(P)$.
We refer to these as {\em pseudomoments} of a {\em pseudodistribution}.

\subsection{Reweighing and conditioning} We will sometimes {\em reweigh} or {\em condition} our degree-$D$ pseudodistribution by a sum-of-squares polynomial $s(x)$ of degree $d < D$; this simply means that we define a new pseudoexpectation operator $\pE'$  of degree $D-d$ by taking, for every monomial $x^{\alpha}$ of degree at most $D - d$, $\pE'[x^{\alpha}] = \frac{\pE[x^{\alpha} \cdot s(x)]}{\pE[s(x)]}$.
One can show that reweighing preserves the four properties of the pseudodistribution up to degree $D - d$.
When $s(x)$ is a $0/1$ function, this is also called ``conditioning'', and we may denote $\pE'$ by $\pE[\cdot ~|~ s(x)]$.
See \cite{BarakRS11,BKS17} for further discussion.

\subsection{Independent samples}
Throughout the paper, we make use of ``shift partition'' variables $\{Z_{u,s}\}_{u \in V, s \in \Sigma}$ which we define as
\[
Z_{u,s} = \sum_{a \in \Sigma} X_{u,a} X'_{u,a+s}
\]
for $X,X'$ ``independent copies'' of $X$.
Formally, given a pseudoexpectation operator $\pE: X^{\leq D} \rightarrow \R$, we define a pseudoexpectation operator $\pE_{X,X'}$ on monomials of degree $\leq D$ in variables $X,X'$: for any monomial $X^{\alpha}(X')^{\beta}$ in $X,X'$,  $\pE_{X,X'}[X^{\alpha}(X')^{\beta}] := \pE_{X}[X^{\alpha}] \cdot \pE_X[X^{\beta}]$.
Similar constructs have been used in the literature, see e.g. \cite{BarakKS14}.
We denote the resulting ``product'' pseudodistribution by $\pE_{X,X'}$ and call $X,X'$ as independent samples.
We will use the following facts about polynomials in independent samples, several of which regard the $Z_{u,a}$ specifically.

\begin{fact}\label{fact:indep}
If $\pE_{X}$ is a valid pseudodistribution of degree $D$ in variables $X$, then $\pE_{X,X'} $ is a valid pseudodistribution of degree $D$. Furthermore, if there are additional SOS inequalities that are true for $\T_X$, they also hold for $\T_{X,X'}$.
\end{fact}
\begin{proof}
By definition, $\pE_{X,X'}$ satisfies scaling and linearity.

We next check that $\pE_{X,X'}$ satisfies the non-negativity of squares.
This fact follows from the fact that the degree-$D$ pseudomoment matrix of $\pE_{X,X'}$ is a principal minor of the Kronecker square of the pseudomoment matrix of $X$, that is, of $(\pE X^{\le D})^{\otimes 2}$.
Since $\pE_{X}$ is a valid pseudoexpectation matrix, $\pE_{X}[X^{\le D}]$ is a PSD matrix, and therefore its Kronecker square and any principal minor thereof. 
Finally, in the standard manner any degree-$D$ square polynomial $s$ in variables $X,X'$ can be written as a quadratic form of $s$'s coefficient vector with the submatrix of the Kronecker square.
Thus $\pE_{X,X'}$ satisfies the degree-$S$ \sos inequalities.

Finally, to see that $\pE_{X,X'}[q_i(X) \cdot f(X,X')] = 0$ for any $q_i$ for which we have the constraint $q_i(x) = 0$ and any $f$ of degree at most $D - \deg(q_i)$, we write $f$ in the monomial basis, $f(X,X') = \sum_{\alpha, \beta} \hat f_{\alpha,\beta} \cdot X^{\alpha}(X')^\beta$, and then we have by linearity
\[
\pE_{X,X'}[q_i(X) f(X,X')] = \sum_{\alpha,\beta} \hat f_{\alpha,\beta} \cdot \pE[q_i(X) \cdot X^{\alpha}] \cdot \pE[(X')^{\beta}] = 0,
\]
since $\pE_X[ q_i(X) \cdot X^{\alpha}] = 0$.
This concludes the proof.
\end{proof}

Now, we prove some properties specific to the $Z$ variables.
\begin{fact}\label{fact:z-vars}
Define the shift variable $Z_{u,s} = \sum_{a \in \Sigma} X_{u,a}X'_{u,a+s}$ to be the indicator that $X_u - X_u' = s$, for $X,X'$ degree-$8$ solutions to the \sos relaxation of the UG integer program (\ref{eq:ip}).
Define as well for each edge $(u,v)$ the variables $Y_{(u,v)} = \sum_{a} X_{u,a} X_{v, \pi_{uv}}(a)$ to be the indicator that the edge $(u,v)$ is satisfied.

Then the $Z$ variables satisfy:
\begin{enumerate}
\item Booleanity: $Z_{u,a}^2 = Z_{u,a}$. 
\item Partition constraints: $Z_{u,a}Z_{u,b} = 0$ for $a \neq b$, $\sum_{s} Z_{u,s} = 1$.
\item Crossing edges violate an assignment: $Z_{u,a} Z_{v,b} Y_{(u,v)} Y'_{(u,v)} = 0$ for every edge $(u,v) \in E$ and $a \neq b$. 
\end{enumerate}
\end{fact}
\begin{proof}
The first two items are easily verified via direct computation, using properties of the $X_{u,a}$s.
We prove that the final property holds.
Since our UG instance is affine, we have that for each $i,j \in E$, $\pi_{ij}(a) = a + h_{ij}$ for some $h_{ij} \in \Sigma$.
 Therefore,
 \begin{align}
 Z_{i,s} Z_{j,t}Y_{(i,j)}Y'_{(i,j)}
 &= \sum_{a,b,c,d \in \Sigma} X_{i,a}X'_{i,a+s} \cdot X_{j,b}X'_{j,b+t} \cdot X_{i,c}X_{j,c + h_{ij}} \cdot X'_{i,d} X'_{j,d+h_{ij}}\\
 &= 0,
 \end{align}
 where we derive the final equality from the disjointness constraints (i.e. that $X_{i,a}X_{i,b} = 0$ whenever $a \neq b$), as for the above term to be nonzero we require $a = c$, $d = a + s$, $b = d - t + h_{ij} = a + s - t + h_{ij}$, and also $b = c + h_{ij}$, which implies $a + s - t = c$, a contradiction since $t \neq s$.
This establishes the final property.
\end{proof}

\subsection{Symmetries}
Here, we will prove the symmetry properties that shift-symmetric pseudodistributions satisfy.
\restatelemma{lem:sym}
\begin{proof}
Recall that since $\mu$ is a shift-symmetric pseudodistribution, we have that,
\begin{equation}\label{eq:props}
\pE_{\mu}[X_{u_1,a_1} \cdots X_{u_m,a_m}]
= \frac{1}{|\Sigma|}\sum_{t}\pE_\mu[X_{u_1,a_1 - t}\cdots X_{u_m,a_m -t}] = \pE_{\mu}[X_{u_1,a_1+s} \cdots X_{u_m,a_m+s}]
\end{equation}
for all $\{(u_1,a_1), \ldots, (u_m,a_m)\} \in ([n] \times \Sigma)^{\le D}$ and $s \in \Sigma$. 

The two items now follow because under $\mu$ the marginal probabilities $\pPr_{\mu}[X_v = s]$ are uniform for all $s \in \Sigma$. We have that for all $s \in \Sigma$ and all $u,v \in V$,
\begin{align*}
\pPr_{\mu}[X_v = s \mid X_u = 0]
&= \frac{\pPr_{\mu}[X_v = s, X_u = 0]}{\pPr_{\mu}[X_u = 0]}\\
&= |\Sigma| \cdot \pPr[X_v = s, X_u = 0] 
= \sum_{t \in \Sigma} \pPr[X_v = s+t, X_u = t],
\end{align*}
where in the third equality we have used the shift-invariance of $\mu$, equation~\ref{eq:props}.

Further for any polynomial which satisfies $f(X) = f(X+t)$ for all $t\in \Sigma$,
\begin{align*}
\pE_{\mu}[f(X) \mid X_v - X_u = s]
&= \frac{\pE_{\mu}[\sum_{t \in \Sigma} f(X) \cdot \Ind[X_v = s + t, X_u = t]]}{\pE_{\mu}[\sum_{t \in \Sigma}\Ind[X_v = s+t, X_u = t]]}\\
&= \frac{\pE_{\mu}[\sum_{t \in \Sigma} f(X + t) \cdot \Ind[X_v = s+t, X_u = t]]}{\pE_{\mu}[\sum_{t \in \Sigma}\Ind[X_v = s+t, X_u = t]]}\\
&= \frac{|\Sigma|\cdot\pE_{\mu}[f(X) \cdot \Ind[X_v = s, X_u = 0]]}{|\Sigma|\cdot\pE_{\mu}[\Ind[X_v = s, X_u = 0]]}
= \pE_{\mu}[f(X) \mid X_v = s, X_u = 0],
\end{align*}
where to obtain the second equality we have used the shift-symmetry of $f$, $f(X) = f(X+(t-s))$, and in the penultimate equality we have used the shift-invariance of $\mu$, equation~\ref{eq:props}.
The conclusion follows.
\end{proof}

\subsection{Pseudoprobabilities}
The following definitions will help to ease notation in our proofs.

\begin{definition}[Pseudoprobability of an event]\label{def:pseudoprob}
Let $\mu$ be a pseudodistribution of degree $D$.
If $\cE(X,X')$ is an event such that $\Ind[\cE(X,X')]$ can be expressed as a degree-$D$ function of $X$ and $X'$, then we define the {\em pseudoprobability of $\cE(X,X')$} to be
\[
\pPr[\cE(X,X')] = \pE[\Ind(\cE(X,X')].
\]
Similarly, if $\cF(X)$ is an event and $\deg(\Ind[\cF(X)]) + \deg(\Ind[\cE(X,X')]) \le D$, then we define the {\em pseudoprobability of $\cE(X,X')$ conditioned on $\cF(X)$} to be
\[
\pPr[\cE(X,X') \mid \cF(X)] = \pE[\Ind(\cE(X,X')) \mid \Ind(\cF(X))] = \frac{\pE[\Ind(\cE(X,X')) \cdot \Ind(\cF(X,X'))]}{\pE[\Ind(\cF(X))]}.
\]
\end{definition}

\subsection{Useful lemmas}
We will state two SOS-versions of Cauchy-Schwarz that we will be useful in the Fourier analysis.
\begin{lemma}[Cauchy Schwarz]\label{CS1-prelim}
For for all $\epsilon \in \R_+$, 
$$
\vdash_2 YZ \leq \frac{\epsilon}{2}Y^2 + \frac{1}{2\epsilon}Z^2.
$$
\end{lemma}

\begin{lemma}[Cauchy Schwarz]\label{CS2-prelim}
For a degree-$D$ pseudoexpectation operator, where $D =  2\max(\deg(Y),\deg(Z))$,
$$\pE[YZ]^2 \leq \pE[Y^2]\pE[Z^2].$$
\end{lemma}
Proofs for both lemmas appear in \cite{BarakKS14}.

We also need the following version of H\"{o}lder's inequality which is proven in e.g.~\cite{ODonnellZ13}.
\begin{fact}[H\"{o}lder's Inequality]\label{fact:sos-hol}
For all real $\nu > 0$ we have that,
$$\vdash_4 Y^3 Z \leq \frac{3\nu}{4} Y^4 + \frac{1}{4\nu^3}Z^4.$$ 
\end{fact}

\begin{claim}\label{claim:spectral-sos}
Let $A$ be the transition matrix for a random walk on an undirected (weighted) graph $G$ and $\pi$ be the stationary measure on $G$,
where $\pi$ samples every vertex proportional to its weighted degree. Then $A$ has real eigenvalues, and moreover if $\Pi$ is the projector to the space of $A$'s right eigenvalues of value at most $\lambda$, then as a degree-$2$ \sos inequality we have
\[
\langle f, A \Pi f \rangle_{\pi} \preceq \lambda \langle f, \Pi f \rangle_{\pi}
\]
and 
\[
\langle f, A (\Id - \Pi) f \rangle_{\pi} \preceq \langle f, (\Id - \Pi) f\rangle_{\pi}.
\]
\end{claim}
\begin{proof}
We use that $A$ is self-adjoint in the inner product space $(\R^n,\langle \cdot \rangle_\pi)$, and therefore it has real eigenvalues and its right eigenspace has orthonormal eigenvectors $v_1,\ldots,v_n$.
We may write $f$ according to its orthogonal decomposition, $f = \sum_{i = 1}^n c_i \cdot v_i$ for $c_i$ linear functions of $f$, and if there are $k$ eigenvalues of value at most $\lambda$ then $\Pi f = \sum_{i = 1}^k c_i \cdot v_i$.
We thus have
\begin{align*}
\langle f, A\Pi f \rangle_\pi
&= \left\langle \sum_{i = 1}^n c_i v_i, \sum_{j=1}^k c_j \cdot A v_j\right\rangle_\pi\\
&= \left\langle \sum_{i = 1}^n c_i v_i, \sum_{j=1}^k c_j \lambda_j \cdot v_j\right\rangle_\pi
= \sum_{j = 1}^k \lambda_j c_j^2
\le \lambda \sum_{j=1}^k c_j^2
= \lambda \langle f, \Pi f \rangle_\pi,
\end{align*}
where the inequality is a degree-2 sum of squares because $\lambda_j \le \lambda$, and the $c_j$ are degree-1 functions of $f$.
A near-identical proof gives the second statement when we observe that $A$'s maximum eigenvalue is $\le 1$.
\end{proof}

\section{Reduction from small-set expansion to hypercontractivity}
\label{sec:expansion-red}

Here, we prove Lemma~\ref{lem:sse} for completeness.

\restatelemma{lem:sse}

\begin{proof}
Since $L = \Id - A$ for $A$ the transition matrix of $G$, we have
\[
\langle f, L f\rangle_{\pi} = \|f\|_{\pi,2}^2 - \langle f, A f\rangle_\pi
\]
Where every right eigenvector $u$ of $L$ with eigenvalue $\lambda_u$ is also an eigenvector of $A$ with eigenvalue $1-\lambda_u$.
We can write $f = f_{\le \lambda} + f_{>\lambda}$, with $f_{\le \lambda} = \Pi_{\lambda} f$.
By linearity,
\[
\langle f, A f\rangle_{\pi}  = \langle f, A f_{\le \lambda}\rangle_{\pi} + \langle f, A f_{> \lambda}\rangle_{\pi},
\]
We can derive an upper bound on the second term,
\begin{align*}
\langle f, A f_{> \lambda}\rangle_{\pi}
&\le  (1-\lambda)\|f\|^2_{\pi,2},
\end{align*}
where the difference between the right- and left-hand side of the inequality is a degree-$2$ sum of squares because $A$'s eigenvalues off the support of $\Pi_{\lambda}$ are bounded by $(1-\lambda)$ (see Claim~\ref{claim:spectral-sos}).

For the first term, we can derive a different upper bound,
\begin{align*}
\langle f, A f_{\le \lambda}\rangle_{\pi}
&\le \langle f, \Pi_{\lambda} f\rangle_{\pi} \\
&= \langle (f^{\circ 3}), \Pi_{\lambda }f\rangle_{\pi} + \langle (f^{\circ 3} - f), \Pi_{\lambda} f\rangle_{\pi},
\end{align*}
where the first inequality follows from the fact that $A$'s eigenvalues are bounded by $1$ (which gives the first line as an SOS inequality, again see Claim~\ref{claim:spectral-sos}), and in the second line we have used $f^{\circ 3}$ to denote the function $f^{\circ 3}:V(G) \to \R$ given by $f^{\circ 3}(v) = f(v)^3$.
Given the Booleanity axioms we have that
$\{f(v)^2 = f(v)\}_{v \in V(G)} \vdash_3 f^{\circ 3} - f = 0$,
so therefore we have from our axioms $\cA$ that
\[
\cA \vdash_4 \langle f, A f_{\le \lambda}\rangle_{\pi}  \le \langle (f^{\circ 3}), \Pi_{\lambda} f\rangle_{\pi}.
\]
Now, using the shorthand $f_v := f(v)$,
\begin{align*}
\langle f^{\circ 3}, \Pi_{\lambda} f \rangle_{\pi}
&\le \frac{3\eta}{4} \|f\|_{\pi,4}^4 + \frac{1}{4\eta^3} \|\Pi_{\lambda} f\|_{\pi, 4}^4\\
&\le \frac{3\eta}{4} \|f\|_{\pi,4}^4 + \frac{1}{4\eta^3} C \|f\|_{\pi,2}^4,
\end{align*}
where the first inequality is an SOS inequality for any $\eta > 0$ (see Fact~\ref{fact:sos-hol}), and the final inequality is guaranteed to be an SOS inequality from our 2-4 hypercontractivity axiom. We can further simplify the inequality above to get that,

\[\langle f^{\circ 3}, \Pi_{\lambda} f \rangle_{\pi} \leq  \frac{3\eta}{4} \E_\pi[f] + \frac{1}{4\eta^3} C \E_{\pi}[f]^2,\]

since, $\frac{3\eta}{4}(\E[f]- \|f\|_{\pi,4}^4) + \frac{C}{4\eta^3}((\E[f]^2 - \|f\|_{\pi,2}^4)$ is a degree-4 sum-of-squares under the axioms $\{f_v \in [0,1]\}_{v \in V(G)}$, as $f_v - f_v^4 \ge 0$ and $f_v f_u - f_v^2 f_u^2 \ge 0$ are SOS inequalities for $f_v \in [0,1]$.

Putting both the upper bounds together, we have that
\begin{align*}
\langle f, L f \rangle_{\pi}
&= \|f\|_{\pi,2}^2 - (1-\lambda)\|f\|_{\pi,2}^2 - \frac{3\eta}{4}\E_\pi[f] - \frac{C}{4\eta^3}\E_{\pi}[f]^2 + \langle (f^{\circ 3} - f), \Pi_{\lambda}f \rangle_{\pi} + S'(f)
\end{align*}
for $S'(f)$ a degree-4 sum of squares in the span of the hypercontractivity and Booleanity axioms. We also have that from the Booleanity axioms, $\|f\|_{\pi,2}^2 = \E_\pi[f]$, so rearranging terms we get that,
\begin{align*}
\langle f, L f \rangle_{\pi}
&= \left(\lambda - \eta \right)\cdot \E_\pi[f] +  \frac{1}{4}\left(\eta \E_\pi[f] - \frac{C}{\eta^3} \E_\pi[f]^2\right) + (\lambda(\|f\|_{\pi,2}^2-\E[f]) + \langle f^{\circ 3} - f, \Pi_{\lambda} f \rangle_\pi) + S'(f) \\
&= \frac{\lambda}{2}\cdot \E_\pi[f] +  \frac{2C}{\lambda^3}\left(\frac{\lambda^4}{16C} \E_\pi[f] - \E_\pi[f]^2\right) + (2(\|f\|_{\pi,2}^2-\E[f]) + \langle f^{\circ 3} - f, \Pi_{\lambda} f \rangle_\pi) + S(f)
\end{align*}
where we have set $\eta = \frac{\lambda}{2}$ and $S(f) = S'(f) + (2 - \lambda)(\E[f] - \|f\|_{\pi,2}^2)$, which is a sum-of-squares because $\lambda \leq 2$ (all eigenvalues of the Laplacian are bounded above by $2$) and $f_v - f_v^2 \geq 0$ is an SOS inequality under the axiom $\{f_v \in [0,1]\}_{v \in V(G)}$. Taking $B(f) = 2(\|f\|_{\pi,2}^2-\E[f]) + \langle f^{\circ 3} - f, \Pi_{\lambda} f \rangle_\pi$ gives us the conclusion.
\end{proof}

\section{Structure Theorem for the Johnson graph}\label{sec:fourier}
In this section, we prove a structure theorem for the non-expanding sets of the Johnson graph. Spectral analysis on the Johnson graph turns out to be complicated, so we move to a closely related Cayley graph, whose eigenstructure is simple to calculate. We will call this the Johnson-approximating graph $C_{n,\ell,\alpha}$. We will prove the following structure theorem about $C_{n,\ell,\alpha}$:

\begin{theorem}\label{thm:structure}
For all $\alpha \in (0,1)$, all integers $\ell \geq 1/\alpha$ and all integers $n \geq \ell$, the following holds: Let $C_{n,\ell,\alpha}$ be the Johnson-approximating graph and $\pi$ be the uniform distribution over $V(C)$. For every positive integer $r \leq \ell/2$ and every permutation-invariant function $F$ that is not correlated with any $r$-restricted subcube, $F$ has high expansion:
\begin{align*}
&\{F(X) \in [0,1]\}_{X \in V(C)} \cup \A_{inv} \,\, \vdash_2\,\, \\ 
&\langle F, L F\rangle_{\pi} \ge (1 - (1 - \alpha)^{r+1})\left[\E_\pi[F] - 8^r \binom{\ell}{r}\left(\sum_{j = 0}^r \E_{Y \in [n]^{j}}[\delta_{Y}(F)^2]\right) + B(F)\right],
\end{align*}
where $B(F)$ represents the Booleanity constraints and equals $\E_\pi[F^{\circ 2} - F]$. 
\end{theorem}

The proof of the theorem above, follows pretty much on the lines of the proof given in~\cite{KMMS}. Since the spectral analysis is much easier on this graph and in the end, we want to prove a weak characterization of non-expanding sets, our proof ends up being simpler. Given this structure theorem, it is straightforward to derive a structure theorem for the Johnson graph and we do so at the end of this section.

\paragraph{Notation:}We will now give some notation that we need for this section.
We use $[n]$ to denote the set $\{0,\ldots,n-1\}$, and also the group $(\Z/n\Z)$, the natural numbers modulo $n$. 
Generally, when we take a set $S$ and raise it to a positive integer power $\ell$, we mean the set of all ordered multisets of elements of $S$ of size $\ell$. We use $\Chi_t$ for $t \in [n]$ to denote the characters of the group $\Z/n\Z$ (or the eigenvectors of the $n$-cycle), where $\Chi_t: [n] \rightarrow \mathbb{C}$ is the function $\Chi_t(x) = e^{\frac{2\pi i tx}{n}}$.
We will use $\lambda_G(v)$ to denote the eigenvalue of $v$ which is an eigenvector of the adjacency matrix of graph $G$. For a string $S \in \Sigma^m$, for some alphabet $\Sigma$, and a set $I \subseteq [m]$, we denote its restriction to the set of coordinates in $I$, by $S|_I$.

\subsection{Preliminaries about the Spectrum}

\begin{definition}
Let $\alpha$ be a number in $(0,1)$ and $\ell$ be a positive integer. Let $n$ be a positive integer such that $n > \ell$. We then define the graph $C_{n,\ell,\alpha}$ as follows:
\begin{enumerate}
\item The vertex set of $C_{n,\ell,\alpha}$ is the set, $[n]^{\ell}$. We will drop the subscript $(n,\ell,\alpha)$ in $C_{n,\ell,\alpha}$ when these parameters are clear from context.
\item The edges are described by showing how to sample a uniformly random neighbor of an arbitrary vertex $X \in [n]^\ell$. Fix a vertex $X = (x_1,\ldots,x_\ell), x_i \in [n]$. Choose $(y_1,\ldots,y_\ell)$ uniformly at random from $[n]^{\ell}$ and $b = (b_1,\ldots,b_{\ell}) \sim \{0,1\}^{\ell}$ such that the Hamming weight of $b$ equals $\alpha \ell$. Let the neighbor of $X$ be $Z = (x_1 + b_1 \cdot y_1,\ldots,x_l + b_\ell \cdot y_\ell)$.
\end{enumerate}
\end{definition}

It is easy to verify that the graph defined above is a weighted Cayley graph with vertex set being the elements of the group $[n]^\ell = (\Z/n\Z)^{\ell}$. The natural group operation associated with this set is component-wise addition modulo $n$, which we will denote by $x+y$ for any two elements $x,y$ in $[n]^\ell$. We will now analyze the spectral properties of the graph. We will overload the notation $C$ to also refer to the normalized adjacency matrix of the graph $C$. Note firstly that the eigenvectors of $C$ are the characters of the group $[n]^\ell$ which we will denote by $\raisebox{2pt}{$\chi$}_{T}$, where $T = (T_1,\ldots,T_\ell) \in [n]^{\ell}$. We have that for all $x \in [n]^\ell$, $\Chi_T(x) = \Chi_{T_1}(x_1) \cdot \ldots \cdot \Chi_{T_{\ell}}(x_\ell)$, where $\Chi_{t}$ denotes the characters of $\Z/n\Z$ or equivalently the eigenvectors of the $n$-cycle.
We will now define a notion of degree for an eigenvector.

\begin{definition}[Degree of $\Chi_T$]
For all $T \in [n]^\ell$, where $T = (T_1,\ldots,T_\ell)$, define the degree of $T$ as:
$$|T| := |\{i \mid T_i \neq 0\}|,$$
The degree of $\Chi_T$ is defined as $|T|$.
\end{definition}

We will now calculate the eigenvalues of $C$. We will show that the eigenvalue corresponding to $\Chi_T$ only depends on $|T|$. Moreover when $|T| \ll \ell$, the eigenvalue of $\Chi_T$ grows exponentially small with $|T|$.

\begin{lemma}\label{lem:eigenval}
Let $\lambda_C(\Chi_T)$ denote the eigenvalue of $C$ corresponding to the eigenvector $\Chi_T$ for $T \in [n]^\ell$. We have that, 
$$
\lambda_{C}(\raisebox{2pt}{$\chi$}_T) = 
\begin{cases}
\frac{\binom{\ell - |T|}{(1 - \alpha)\ell - |T|}}{\binom{\ell}{(1 - \alpha)\ell}}, ~~~~|T| \leq (1 - \alpha)\ell\\
0, ~~~~~~~~~~~~~~~~~~~~~~~\text{otherwise.} 
\end{cases}.
$$
\end{lemma}

\begin{proof}
Let $T = (T_1, \ldots, T_\ell)$. For all $X \in [n]^\ell$, we have that,

\begin{align*}
C \cdot \Chi_T(X) &= \E_{y,b}[\Chi_T(x_1+b_1 y, \ldots, x_l + b_l y_l)] \\
&= \Chi_T(X) \E_{y,b}[\Chi_{T_1,\ldots,T_l}(b_1y_1,\ldots,b_ly_l)] \\
&= \Chi_T(X) \E_{y,b}[\Chi_{(b_1 T_1,\ldots,b_l T_l)}(y)].
\end{align*}

For $y,S \in [n]$ and $S \neq 0$, we know that the eigenvector $\Chi_S$ is orthogonal to the eigenvector $\Chi_0$, equivalently that $\E_y[\Chi_S(y)] = 0$, whereas if $S = 0$ then $\E_y[\Chi_S(y)] = 1$. So we get that,

\begin{align*}
\lambda_{C}(\Chi_T) &= \E_{y,b}[\Chi_{(b_1 T_1,\ldots,b_\ell T_\ell)}(y)] \\
&= \Pr_b[(b_1T_1,\ldots,b_\ell T_\ell) = 0^{\ell}] \\
&= \begin{cases}
\frac{\binom{l - |T|}{(1 - \alpha)l - |T|}}{\binom{l}{(1 - \alpha)l}}, ~~~~|T| \leq (1 - \alpha)\ell\\
0, ~~~~~~~~~~~~~~~~~~~~~~~\text{otherwise.} 
\end{cases}.
\end{align*}
\end{proof}

\subsection{Analyzing non-expanding sets of the Johnson-approximating graph}

Since our main aim in Section~\ref{sec:johnson} is to deal with sets in the Johnson graph $J_{n,\ell,\alpha}$ we will only consider ``permutation-invariant'' sets on $C_{n,\ell,
\alpha}$. Notice that the vertices of the Johnson graph are subsets of $[n]$ of size $\ell$, whereas the vertices of the Johnson-approximating graph are ordered $\ell$-tuples of $[n]$. Therefore, given a set $S$ in the Johnson graph, it has a natural mapping to the set $S'$ which is a subset of the vertices of the Johnson-approximating graph $C$, 
$S' := \{(x_{\pi(1)},\ldots,x_{\pi(l)}) \mid \pi: [l] \rightarrow [l], \{x_1,\ldots,x_l\} \in S\}$. This leads to the following definition:

\begin{definition}[Permutation-invariance]\label{def:perm-inv}
We say that a set $S \subseteq C_{n,\ell,\alpha}$ is permutation-invariant if for all permutations $\pi \in \mathcal{S}_\ell$, the symmetric group on $\ell$ elements, and all $X = (x_1,\ldots,x_l) \in S$, we have that $X_{\pi} = (x_{\pi(1)},\ldots,x_{\pi(l)})$ belongs to $S$. Similarly a function $F: V(C) \rightarrow \R$ is permutation invariant if for all inputs $X = (x_1,\ldots,x_{\ell})$, we have that $F(x_1,\ldots,x_l) = F(x_{\pi(1)},\ldots,x_{\pi(l)})$, for all permutations $\pi$. Further let $\A_{inv}$ denote the set of axioms that $F$ is permutation-invariant, that is, 
\[\A_{inv} := \{F(x_1,\ldots,x_\ell) = F(x_{\pi(1)},\ldots,x_{\pi(l)})\}_{\pi \in \mathcal{S}_\ell, X \in [n]^\ell}.\]
\end{definition}

Since the set of vertices in $C$ that correspond to some set of vertices in $J$ are permutation invariant it will be enough to focus are attention on these special sets and from now on whenever we refer to a set in $V(C)$, the reader can assume that it is permutation-invariant.

To analyze non-expanding sets of $C$, we will consider permutation-invariant functions $F: V(C) \rightarrow [0,1]$. Typically one would consider $0/1$-valued functions $F$, where $F$ is the indicator function of a set $S$, i.e. $F(X) = 1$ when $X \in S$. But since we need to analyze ``approximate-sets'' (the indicator function is approximated by a polynomial that takes values close to $0/1$), $F(X)$ could take any value between $[0,1]$.

Recall that the Fourier decomposition of $F$ gives us that, $F(X) = \sum_T \hat{F}(T)\Chi_T(X)$. We will now define the following for a function $F$:

\begin{definition}
\begin{enumerate}
\item We will expand $F$ as 
\[F = F_0 + \ldots + F_\ell,\] 
where $F_i(X) = \sum_{T: |T| = i} \hat{F}(T) \Chi_T(X)$. We will call $F$ a level $i$ function, if its Fourier decomposition has degree $i$ characters only, i.e. $\hat{F}(T) = 0,$ for all $T$ such that $|T| \neq i$.
\item Let $f_{i,F}: [n]^i \rightarrow \R$ be a function defined as, 
$$f_{i,F}(x_1,\ldots,x_i) := \sum_{T_1,\ldots,T_i \in ([n] \setminus 0)^i}\hat{F}(T_1,\ldots,T_i,0,\ldots)\chi_{T_1,\ldots,T_i}(x_1,\ldots,x_i),$$ 
\end{enumerate}
\end{definition}

Let $X = (x_1,\ldots,x_{j}) \in [n]^{j}$ and $I$ be a subset of $\{1,\ldots,j\}$. Let $I = \{k_1,\ldots,k_{|I|}\}$ where $k_1 < k_2 < \ldots < k_{|I|}$. We will use $X|_I$ to denote the ordered tuple of elements $(x_{k_1},\ldots,x_{k_{|I|}})$. We will now state some simple properties of $F$ that are implied by permutation-invariance.

\begin{lemma}\label{lem:props}
For all functions $F:[n]^\ell \rightarrow \R$ that are permutation-invariant, we have that:
\begin{enumerate}
\item $\hat{F}(T_1,\ldots, T_l) = \hat{F}(T_{\pi(1)},\ldots,T_{\pi(l)})$, for all $(T_1,\ldots,T_l) \in [n]^\ell$ and all permutations $\pi:[l] \rightarrow [l]$.

\item The functions $F_i$ and $f_{i,F}$ are also permutation-invariant.

\item $F_i(X) = \sum\limits_{\substack{I \subseteq [l] \\ |I| = i}} f_{i,F}(X|_{I})$.

\end{enumerate}
\end{lemma}
We skip the proof of this lemma because it follows by a straightforward manipulation of the definitions.

\begin{definition}[$r$-restricted subcubes of $C$]
Given an ordered tuple, $A = (a_1,\ldots,a_r)$ for $a_i \in [n]$ and $r \leq l-1$, we let $C|_{A}$ denote the subset of vertices of $C$ whose first $r$ coordinates are restricted to be $(a_1,\ldots,a_r)$. We call such a subset an $r$-restricted subcube of $C$.
\end{definition}

\begin{definition}[Restrictions]
Given a function $F:[n]^\ell \rightarrow \R$ and an ordered tuple, $A = (a_1,\ldots,a_r)$ for $a_i \in [n]$ and $1 \leq r \leq l-1$, we define the restricted function $F|_A: [n]^{\ell - r} \rightarrow \R$ as, 
$$F|_A(x_1,\ldots,x_{l-r}) = F(a_1,\ldots,a_r, x_1,\ldots,x_{l-r}).$$
Further, let $\delta_A(F)$ denote the mass of the function restricted to $A$, that is,
$$\delta_A(F) := \delta(F|_A) = \E\limits_{X \in [n]^{\ell - r}}[F|_A(X)].$$
For convenience, when $A = \phi$ ($r = 0$), define $F|_A(X) := F(X)$, and $\delta_A(F) := \delta(F) = \E_{X \in [n]^\ell}[F(X)]$.
\end{definition}

The following simple facts hold for restrictions of functions:

\begin{lemma}\label{lem:restriction}
Let $F$ be a permutation-invariant function on $V(C)$. Then we have the following:
\begin{enumerate}
\item For all $a \in [n]$ and for all $i$ such that $0 \leq i \leq \ell-1$, and all $X \in [n]^i$, we have that,
\[f_{i+1,F}(a, X) = f_{i,F|_{\{a\}}}(X) - f_{i,F}(X).\]

\item For all integers $i$ such that $0 \leq i \leq \ell$ and for all $X \in [n]^i$, we get an inclusion-exclusion formula for $f_i(X)$ in terms of restrictions of $F$: 
\[f_{i,F}(X) = \sum_{B \subseteq \{1,\ldots,i\}} (-1)^{i - |B|} \delta_{X|_B}(F),\]
where $X|_B$ is the ordered tuple of elements of $X$ restricted to the indices in $B$.
\end{enumerate}
\end{lemma}

\begin{proof}[Proof of (1)]
Using the definition, we can expand out $f_{i+1,F}$ to get that,

\[f_{i,F}(a, X) = \sum_{(T_1,\ldots,T_{i+1}) \in ([n] \setminus 0)^{i+1}}\hat{F}(T_1,\ldots,T_{i+1},0,\ldots, 0)\Chi_{T_1,\ldots,T_{i+1}}(a, X).\]

We can split this sum into two parts, one where $T_{1}$ can take \emph{any} value (even $0$) and the 
second where $T_{1} = 0$. We get that,

\begin{align}
f_{i,F}(a,X) &= \sum_{ \substack{T_1 \in [n] \\ T \in ([n] \setminus 0)^i}}\hat{F}(T_1,T, 0,\ldots,0,\ldots, 0)\Chi_{T_1,T}(a,X) \\
&- \sum_{ \substack{T_1 = 0 \\ T \in ([n] \setminus 0)^i}}\hat{F}(0, T ,0,\ldots, 0)\Chi_{T}(X).
\end{align}
We will show that the first term equals $f_{i, F|_{\{a\}}}(X)$ and the second term equals $f_{i,F}(X)$. This implies the conclusion needed.

For the first term we have that, 

\begin{align}
&\sum_{ \substack{T_1 \in [n] \\ T \in ([n] \setminus 0)^i}}\hat{F}(T_1, T,0,\ldots,0)\Chi_{T_1,T}(a, X) \nonumber \\
= &\sum_{ \substack{T_1 \in [n] \\ T \in ([n] \setminus 0)^i}} \E_{\substack{Y_1 \in [n], \\ Y \in [n]^{\ell - 1}}}\left[F(Y_1,Y)\Chi_{T_1}(Y_1)\Chi_{(T,0,\ldots,0)}(Y)\right]\Chi_{T_1}(a)\Chi_{T}(X) \nonumber \\
= &\sum_{T \in ([n] \setminus 0)^i} \E_{\substack{Y_1 \in [n], \\ Y \in [n]^{\ell - 1}}}\left[F(Y_1,Y)\Chi_{(T,0,\ldots,0)}(Y)\sum_{T_1 \in [n]}\Chi_{T_1}(a+Y_1)\right]\Chi_{T}(X). \label{eq:last}
\end{align}

We now have that $\sum_{T_1 \in [n]}\Chi_{T_1}(a+Y_1) = 0$ if $T_1 \neq a$ and equals $n$ otherwise. Using this fact we get that equation~\ref{eq:last} equals,

\begin{align*}
&\sum_{T \in ([n] \setminus 0)^i} \frac{1}{n} \cdot \E_{Y \in [n]^{\ell - 1}}\left[F(a,Y)\Chi_{(T,0,\ldots,0)}(Y)\cdot n\right]\Chi_{T}(X) \\
= &\sum_{T \in ([n] \setminus 0)^i} \E_{Y \in [n]^{\ell - 1}}\widehat{F|_{\{a\}}}(T,0,\ldots,0)\Chi_{T}(X) \\
= &f_{i, F|_{\{a\}}}
\end{align*}

For the second term we have that,

\begin{align*}
&\sum_{ \substack{T_1 = 0 \\ (T_2,\ldots,T_{i+1}) \in ([n] \setminus 0)^i}}\widehat{F}(0, T_2,\ldots,T_{i+1},0,\ldots, 0)\Chi_{T_2,\ldots,T_{i+1}}(X)\\
= &\sum_{(T_2,\ldots,T_{i+1}) \in ([n] \setminus 0)^i}\hat{F}(T_2, \ldots,T_{i+1},0,\ldots, 0)\Chi_{T_2,\ldots,T_{i+1}}(X),
\end{align*}
since by Lemma~\ref{lem:props} (1) we have that $\hat{F}(0, T_2,\ldots,T_{i+1}, 0, \ldots) = \hat{F}(0, T_2,\ldots,T_{i+1}, 0, \ldots)$. Since the last equality is the definition of $f_{i,F}(X)$, the conclusion follows.

\end{proof}

\begin{proof}[Proof of (2)]
We will prove this claim by induction on $i$. For the base case of $i = 0$, by definition, we have that,

\[f_{0,F}(\phi) = \hat{F}(0,\ldots,0) = \E_{X \in [n]^\ell}[F(X)] = \delta(F) = \delta_{\phi}(F) = \sum_{B \subseteq \phi} \delta_{\phi|_B}(F).\]

Now let us assume that for all permutation-invariant functions $G$ the claim holds for $i-1$, i.e. for all $X \in [n^{i-1}]$, we have that $f_{i-1,G}(X) = \sum_{B \in \{1,\ldots,i-1\}} (-1)^{i-1-|B|} \delta_{X|_{B}}(G)$. Now we will prove the  claim for $f_{i,F}$, thus completing the induction.

Let $X = (x_1, X')$, where $X \in [n]^i, x_1 \in [n]$ and $X' \in [n]^{i-1}$. Then by property (1) of the same lemma, we have that,
\[f_{i,F}(X) = f_{i,F|_{\{x_1\}}}(X') - f_{i,F}(X').\]

Expanding the RHS using the induction hypothesis on the functions $F|_{\{x_1\}}$ and $F$, we get that,

\begin{align*}
f_{i,F}(X) &= \sum_{B' \in \{1,\ldots,i-1\}} (-1)^{i-1-|B'|} \delta_{X'|_{B'}}(F|_{\{x_1\}}) - \sum_{B' \in \{1,\ldots,i-1\}} (-1)^{i-1-|B'|} \delta_{X'|_{B'}}(F) \\
&= \sum_{B' \in \{1,\ldots,i-1\}} (-1)^{i-(1+|B'|)} \delta_{(x_1,X'|_{B'})}(F) + \sum_{B' \in \{1,\ldots,i-1\}} (-1)^{i-|B'|} \delta_{X'|_{B'}}(F) \\
&= \sum_{B \in \{1,\ldots,i\}: 1 \in B} (-1)^{i-|B|} \delta_{X|_B}(F) + \sum_{B \in \{1,\ldots,i\}: 1 \notin B} (-1)^{i-|B|} \delta_{X|_{B}}(F) \\
&= \sum_{B \in \{1,\ldots, i\}} (-1)^{i-|B|} \delta_{X|_B}(F).
\end{align*}

This completes the inductive step and the proof of the lemma.
\end{proof}

We will first upper bound the Fourier weights on the lower levels. To do so we will use the following relation between the Fourier weight $\eta_i$ and $f_i$'s.

\begin{lemma}\label{lem:weight}
Let $F$ be a permutation-invariant function on the vertices of $C$. Then, we have that,
$$\E_{X \in [n]^{i}}[f_{i,F}(X)^2] =  \frac{\eta_{i}}{\binom{l}{i}},$$
where $\eta_i = \E_{Y \in V(C)}[F_i(Y)^2]$.
\end{lemma}

\begin{proof}
Recall that if $I \subseteq[l]$, we will use $f_i(I)$ to denote $f_i(A|_I)$. Since $F_i(A) = \sum_{I \subseteq [l]}F_i(A|_I)$, we have that, 
\begin{align*}
    \E_{A \sim [n]^{\ell}}[F_i(A)^2] &= \E_A[(\sum_I f_i(I))^2] \\
    &= \sum_I \E_A[ f_i(I)^2] + \sum_{I \neq I'} \E_{A}[f_i(I)f_i(I')] \\
    &= \sum_I \E_A[ f_i(I)^2] + 0 \\
    &= \sum_I \E_{(a_1,\ldots,a_{\ell})}[ f_i(a_{j_1},\ldots,a_{j_i})^2], &\text{where } j_1,\ldots,j_i \in I \\
    &= \E_{(x_1,\ldots,x_i)}[f_i(X)^2] \cdot \binom{\ell}{i},
\end{align*}
rearranging which, immediately implies the lemma. 
\end{proof}

Recall that $\A_{inv}$ denotes the set of axioms that $F$ is permutation-invariant. We will now bound the $i^{th}$-level Fourier weight of a permutation-invariant function $F$.
\begin{lemma}[Upper Bound on Level-$i$ Weight]\label{lem:ub-weight}
Let $F$ be a permutation-invariant function on $V(C)$. Then for all $i$ such that $0 \leq i \leq \ell$, we can bound the Fourier weight of $F$ on the $i^{th}$ level using its restrictions:
\[\A_{inv} \,\, \vdash_2 \,\, \eta_{i} \leq 2^i\binom{\ell }{i} \cdot \left(\sum_{j = 0}^i \binom{i}{j} \E_{Y \in [n]^{j}}\left[\delta_{Y}(F)^2\right]\right),\]
where $\eta_{i} = \ip{F_i,F_i}_\pi$, for $\pi$ equal to the uniform distribution over $V(C)$.
\end{lemma}

\begin{proof}
Firstly, using Lemma~\ref{lem:weight} we get that,
\begin{equation}\label{eq:weight}
\eta_{i} = \binom{\ell}{i} \cdot \E_{X \in [n]^{i}}[f_{i,F}(X)^2].
\end{equation}

Using the expansion of $f_{i,F}$ from Lemma~\ref{lem:restriction} (2), we get that,

\[f_{i,F}(X)^2 = \left(\sum_{B \subseteq \{1,\ldots,i\}} (-1)^{i - |B|} \delta_{X|_B}(F)\right)^2 \leq 2^i \cdot \sum_{B \subseteq \{1,\ldots,i\}} \delta_{X|_B}(F)^2, \]
where in the last step we have used the Cauchy-Schwarz inequality. Noting that this is a degree 2 SOS inequality and substituting this expression into (\ref{eq:weight}) we get that,

\begin{equation}\label{eq:weight1}
\A_{inv} \,\, \vdash_2 \,\, \eta_{i} \leq 2^i\binom{\ell}{i} \cdot \E_{X \in [n]^i}\left[\sum_{B \subseteq \{1,\ldots,i\}} \delta_{X|_B}(F)^2\right].
\end{equation}

We can now simplify the RHS further. We have that,
\begin{align*}
&\E_{X \in [n]^i}\left[\sum_{B \subseteq \{1,\ldots,i\}} \delta_{X|_B}(F)^2\right] \\
= &\sum_{B \subseteq \{1,\ldots,i\}} \E_{X \in [n]^i}\left[\delta_{X|_B}(F)^2\right] \\
= & \sum_{B \subseteq \{1,\ldots,i\}} \E_{Y \in [n]^{|B|}}\left[\delta_{Y}(F)^2\right] \\   
= & \sum_{j = 0}^i \binom{i}{j} \E_{Y \in [n]^{j}}\left[\delta_{Y}(F)^2\right].  
\end{align*}

Plugging in the last equation into equation (\ref{eq:weight1}), we get the conclusion.
\end{proof}

We will now prove the main structure theorem for $C$. This theorem can be interpreted as saying that if $F$ is the indicator function of a permutation-invariant set $S$ such that $S$ is not correlated with any of the $r$-restricted subcubes of $C$ then $S$ has high expansion in $C$ ($\ip{F,LF}$ is large). In our theorem, correlation with a subcube $C|_A$ for $A \in [n]^r$, is measured by the squared-mass of $F|_{A}$, which is equal to $\delta_A(F)^2$ (Note that $\delta_A(F)$ is equal to $\frac{|S \cap (C|_A)|}{|(C|_A)|}$).

\begin{reminder}{Theorem~\ref{thm:structure}:}
For all $\alpha \in (0,1)$, all integers $\ell \geq 1/\alpha$ and all integers $n \geq \ell$, the following holds: Let $C_{n,\ell,\alpha}$ be the Johnson-approximating graph and $\pi$ be the uniform distribution over $V(C)$. For every positive integer $r \leq \ell/2$ and every permutation-invariant function $F$ that is not correlated with any $r$-restricted subcube, $F$ has high expansion:
\begin{align*}
&\{F(X) \in [0,1]\}_{X \in V(C)} \cup \A_{inv} \,\, \vdash_2\,\, \\ 
&\langle F, L F\rangle_{\pi} \ge (1 - (1 - \alpha)^{r+1})\left[\E_\pi[F] - 8^r \binom{\ell}{r}\left(\sum_{j = 0}^r \E_{Y \in [n]^{j}}[\delta_{Y}(F)^2]\right) + B(F)\right],
\end{align*}
where $B(F)$ represents the Booleanity constraints and equals $\E_\pi[F^{\circ 2} - F]$. 
\end{reminder}

\begin{proof}
We know that $\ip{F,LF}_\pi = \E_\pi[F^{\circ 2}] - \ip{F,AF}_\pi$, where $\pi$ is the uniform distribution over $V(C)$. We will now upper bound $\ip{F,AF}$. Let $\lambda_i$ denote the eigenvalue of the level $i$ eigenvectors of $C$. From Lemma~\ref{lem:eigenval}, we have that $\lambda_i = \frac{\binom{\ell - |T|}{(1 - \alpha)\ell - |T|}}{\binom{\ell}{(1 - \alpha)\ell}}$ for $i \leq (1-\alpha)\ell$ and $0$ otherwise. One can check that ${\lambda_i \leq (1-\alpha)\lambda_{i-1}}$ for all $i$ between $1$ and $\ell$. Since $\lambda_0 = 1$, we get that $\lambda_i \leq (1-\alpha)^i$. We will use this upper bound because it is easier to work with in calculations.

Let $\eta_i = \ip{F_i, F_i}_\pi$ be the Fourier weight on level $i$. Expanding out $\ip{F,AF}_\pi$ we get that,

\begin{align*}
\ip{F,AF}_\pi &= \sum_{i = 0}^r \lambda_i \eta_i + \sum_{i = r+1}^\ell \lambda_i \eta_i \\
&\leq  \sum_{i = 0}^r (1-\alpha)^i \eta_i + (1-\alpha)^{r+1} \sum_{i = r+1}^\ell \eta_i \\
&\leq  \sum_{i = 0}^r \eta_i +  (1-\alpha)^{r+1} \left(\E_\pi[F^{\circ 2}] - \sum_{i = 0}^r \eta_i \right),\\
\end{align*}
where in the last step we have used the fact that, $\lambda_i \leq 1$, for all $i \leq r$, for the first summand and $\sum_{i = 0}^\ell \eta_i = \E_{\pi}[F^{\circ 2}]$ for the second. Further note that each of these inequalities is a degree 2 SoS inequality, since $\eta_i$ is a sum-of-squares for all $i \in [\ell]$. Plugging in the above inequality into the expression for the Laplacian and rearranging it we get that,

\begin{align*}
\ip{F,LF}_\pi &\geq (1 - (1 - \alpha)^{r+1})\left[\E_\pi[F^{\circ 2}] - \sum_{i = 0}^r \eta_i \right] \\
&= (1 - (1 - \alpha)^{r+1})\left[\E_\pi[F] + \E_\pi[F^{\circ 2} - F] - \sum_{i = 0}^r \eta_i  \right] \\
&\geq (1 - (1 - \alpha)^{r+1})\left[\E_\pi[F] + B(F) - \sum_{i = 0}^r 2^i\binom{\ell}{i} \cdot \sum_{j = 0}^i \binom{i}{j} \E_{Y \in [n]^{j}}\left[\delta_{Y}(F)^2\right]\right], 
\end{align*}
where in the last step we have applied the upper bound on $\eta_i$ proved in Lemma~\ref{lem:ub-weight} and substituted $B(F) = \E_\pi[F^{\circ 2} - F]$. All the inequalities are therefore degree 2 SoS inequalities.

We can now apply a simplification to the expression inside the summand to get that,
\begin{align*}
\ip{F,LF}_\pi &\geq (1 - (1 - \alpha)^{r+1})\left[\E_\pi[F] + B(F) - \sum_{j = 0}^r \E_{Y \in [n]^{j}}\left[\delta_{Y}(F)^2\right] \cdot \left(\sum_{i = j}^r 2^i\binom{i}{j}\binom{\ell}{i}\right)\right]\\ 
&\geq (1 - (1 - \alpha)^{r+1})\left[\E_\pi[F] + B(F) - 8^r \binom{\ell}{r}\sum_{j = 0}^r \E_{Y \in [n]^{j}}\left[\delta_{Y}(F)^2 \right] \right] , 
\end{align*}
where in the last inequality we have used the fact that $\sum_{i = j}^r 2^i\binom{i}{j}\binom{\ell}{i} \leq r 2^r \binom{r}{j}\binom{\ell}{r} \leq 8^r \binom{\ell}{r}$.

\end{proof}

We will use the structure theorem for the Johnson-approximating graph given above, to derive a structure theorem for the Johnson graph.

\restatetheorem{thm:structure-johnson}

\begin{proof}
We will use the structure theorem for the Johnson-approximating graph $C_{n,\ell,\alpha}$, to obtain a structure theorem for the Johnson graph $J_{n,\ell,\alpha}$. We will drop the subscript and use $C, J$ henceforth. 

Let $F$ be a function on the vertices of $J$ (given by $\ell$-sized subsets of $[n]$) such that $F(X) \in [0,1]$ for all $X$. Define a function $G: V(C) = [n]^\ell \rightarrow [0,1]$ in the following way:

\begin{align*}
G(x_1,\ldots,x_\ell) = \begin{cases}
F(\{x_1,\ldots,x_\ell\}), ~~~~ &\text{if the elements } x_j's \text{ are all distinct.} \\
0,  &\text{otherwise}.
\end{cases}
\end{align*}

One can check that $G$ satisfies the permutation-invariance axioms from Definition~\ref{def:perm-inv}. We also have that $G(X) \in [0,1]$ for all $X \in [n]^\ell$, when $F$ satisfies the same. So we can apply the structure theorem for the Johnson-approximating graph to $G$ to get that,

\begin{align}\label{eq:struct-approx}
&\{F(X) \in [0,1]\}_{X \in V(J)} \,\, \vdash_2\,\, \\ 
&\langle G, L_C G\rangle_{\pi_C} \ge (1 - (1 - \alpha)^{r+1})\left[\E_{\pi_C}[G] - 8^r \binom{\ell}{r}\left( \sum_{j = 0}^r \E_{Y \in [n]^{j}}\left[\delta_{Y}(G)^2\right]\right) + B_{\pi_C}(G) \right],
\end{align}
where $L_C$ is the Laplacian of $C$, $\pi_C$ is the uniform distribution over $V(C)$ and $B_{\pi_C}(G) = \E_{\pi_C}[G^{\circ 2} - G]$. We will now use the close relation between $F$ and $G$ to bound every term in the above expression to get a similar expression for $F$. We will show the following:

\begin{enumerate}
\item  $\langle G, L_C G\rangle_{\pi_C} \leq \langle F, L F\rangle_{\pi} + \left(\frac{2\ell^2 + \ell}{2n}\right)\E_{\pi}[F]$.
\item $\E_{\pi_C}[G] \geq \left(1 - O_\ell(\frac{1}{n})\right)\E_\pi[F]$.
\item $\E_{Y \in [n]^{j}}\left[\delta_{Y}(G)^2\right] \leq \E_{Y \in \binom{[n]}{j}}\left[\delta_{Y}(F)^2\right]$.
\item $B_{\pi_C}(G) \geq B(F)$, for $B(F) = \E_\pi[F^{\circ 2} - F]$.
\end{enumerate}

Plugging these bounds into equation (\ref{eq:struct-approx}) we get that,

\begin{align*}
&\{F(X) \in [0,1]\}_{X \in V(J)} \,\, \vdash_2\,\, \\ 
&\langle F, L F\rangle_{\pi} \ge (1 - (1 - \alpha)^{r+1})\left[\left(1 - O_\ell\left(\frac{1}{n}\right)\right)\E_{\pi}[F] - 8^r \binom{\ell}{r}\left(\sum_{j = 0}^r \E_{Y \in \binom{[n]}{j}}[\delta_{Y}(F)^2]\right) + B(F) \right] \\
&\qquad\qquad- O_\ell\left(\frac{1}{n}\right)\E_\pi[F].
\end{align*}

Absorbing the last term, $O_\ell\left(\frac{1}{n}\right)\E_\pi[F]$, into the first term inside the brackets, we get the conclusion.

Now let us go into the proofs of points 1 to 4. One can check that all the inequalities below are degree 2 SoS inequalities given the axioms $F(X) \in [0,1]$ for all $X$.

\paragraph{Proof of (1):} Let $E_C$ be the probability distribution over the edges of $C$. By the expansion of the Laplacian we know that,

\[\ip{G,L_C G}_{\pi_C} = \frac{1}{2}\E_{(X, Z) \sim E_C}[(G(X) - G(Z))^2].\]

For edges $(X,Z)$ for which both $X$ and $Z$ have repeated coordinates, we have that $(G(X) - G(Z))^2 = 0$. Let $A(X,Z)$ be the event that none of the endpoints of the edge $(X,Z)$ has repeating coordinates. We have that $\Pr_{(X,Z) \sim E_C}[\neg A(X,Z)] \leq \Pr_{Y \sim \pi_C}[\neg A(Y)] \leq \ell^2/n$, where $A(Y)$ is the event that $Y$ has no repeating coordinates. Furthermore, let $B(X,Z)$ be the event that $X$ and $Z$ differ in exactly $\alpha\ell$ elements. Again one can check that, $\Pr[\neg B(X,Z)] \leq \Pr_{Y \sim \pi_C}[Y_i \neq 0, \forall i \in [\ell]] \leq \frac{\ell}{n}$.

When both the events $A(X,Z)$ and $B(X,Z)$ occur, the distribution $E_C$ is the same as sampling an edge $(C,D) \sim E$ (the uniform distribution over $E(J)$) and randomly permuting the sets $C$ and $D$ to get an ordered tuple $(X,Z)$.

For brevity of notation, we will drop the term $(X,Z)$ in $A(X,Z)$ and $B(X,Z)$. Using the above inequalities we get that,

\begin{align*}
2\ip{G,L_C G}_{\pi_C} \leq &\Pr_{E_C}[A \cap B]\E_{(E_C | A \cap B)}[(G(X) - G(Z))^2] \\
&+ \Pr_{E_C}[\neg A]\E_{(E_C | \neg A)}[(G(X) - G(Z))^2] \\
&+ \Pr_{E_C}[\neg B]\E_{(E_C | B)}[(G(X) - G(Z))^2] \\
\leq &\E_{E}[(F(X) - F(Z))^2] + \left(\frac{2\ell^2}{n}\right)\E_{X \sim \pi}[F(X)^2] + \left(\frac{\ell}{n}\right)\E_{X \sim \pi}[F(X)^2]\\
= & 2\ip{F,LF}_\pi +   \left(\frac{2\ell^2 + \ell}{n}\right)\E_{X \sim \pi}[F(X)].
\end{align*}

\paragraph{Proof of (2):}
Let $A(X)$ be the event that $X \sim \pi_C$ has no repeating coordinates. We have that, 

\[\Pr_{X \sim \pi_C}[A(X)] = \frac{\binom{n}{ \ell}\ell!}{n^\ell} \geq \left(1 - \frac{\ell^2}{n} \right).\] 
When $X$ has repeating coordinates $G(X) = 0$ and otherwise $G(X) = F(X)$ (when we apply $F$ on $X$ we think of $X$ as a $\ell$-sized subset of $[n]$ and therefore a vertex of $J$). Also note that, the distribution $\pi_C$ conditioned on the event that $X$ has no repeating coordinates is uniform over all such $X$'s and is therefore the same as drawing a random set $Y \sim \pi$ and choosing a random ordering of the elements. So we have that,

\begin{align*}
\E_{X \in [n]^\ell}[G(X)] &= \Pr_{X \sim \pi_C}[A(X)] \cdot \E_{X \sim (\pi_C | A)}[G(X)] +  \Pr_{X \sim \pi_C}[\neg A(X)] \cdot \E_{X \sim (\pi_C | \neg A(X))}[G(X)] \\ &\geq \left(1 - \frac{\ell^2}{n} \right) \E_{X \sim \pi}[F(X)].
\end{align*}

\paragraph{Proof of (3):} 
We have that $\delta_Y(G) = 0$ if $Y$ has repeating coordinates, so let us first assume that $Y$ does not have repeating coordinates.  
Let $X \sim [n]^{\ell - j}$ and let $A(X)$ be the event that $(Y,X)$ has no repeating coordinates. Then by definition of restrictions, we get that,
\begin{align*}
\delta_Y(G) &= \Pr_{X}[A(X)] \E_{X \sim ([n]^{\ell -j} | A(X))}[G(Y,X)] + \Pr_{X}[\neg A(X)] \E_{X \sim ([n]^{\ell -j} | \neg A(X))}[G(Y,X)] \\
&\leq 1 \cdot \E_{X \sim \binom{[n] \setminus Y}{\ell - j}}[F(Y,X)] \\
&= \delta_Y(F).
\end{align*}

So we also get that, $\delta_Y(G)^2 \leq \delta_Y(F)^2$.

Now we will calculate an upper bound on $\E_{Y \in [n]^j}[\delta_Y(G)^2]$. Let $Y \sim [n]^j$ and let $A(Y)$ be the event that $Y \sim [n]^j$ has no repeating coordinates. We then have that,

\begin{align*}
\E_{Y \in [n]^j}[\delta_Y(G)^2] &= \Pr_{Y \sim [n]^j}[A(Y)] \cdot \E_{Y \sim ([n]^j | A(Y))}[\delta_Y(G)^2] + \Pr_{Y \sim [n]^j}[\neg A(Y)] \cdot \E_{Y \sim ([n]^j | \neg A(Y))}[\delta_Y(G)^2] \\   
&\leq \E_{Y \sim \binom{[n]}{j}}[\delta_Y(F)^2].
\end{align*}

\paragraph{Proof of (4):}  As in the proof of (3), let $A(X)$ be the event that $X \sim \pi_C$ has no repeating coordinates. We have that,

\begin{align*}
\E_{X \in [n]^\ell}[G - G^{\circ 2}] &= \Pr_{X \sim \pi_C}[A(X)] \cdot \E_{X \sim (\pi_C | A(X))}[G - G^{\circ 2}] +  \Pr_{X \sim \pi_C}[\neg A(X)] \cdot \E_{X \sim (\pi_C | \neg A(X))}[G - G^{\circ 2}] \\ 
&\leq  \E_{X \sim \pi}[F - F^{\circ 2}].
\end{align*}

\end{proof}

\end{document}